\documentclass[copyright]{eptcs}
\providecommand{\event}{VPT 2021} 
\usepackage{makeidx,fancybox,tikz}
\usepackage{bm} 
\usetikzlibrary{automata, positioning, arrows,patterns}
\usepackage{amsmath,hhline,mathrsfs,amsthm}
\usepackage{amssymb,array,amscd}
\usepackage{multirow}

\usepackage{xcolor,relsize,graphicx}

\usepackage[all,2cell]{xy}
\title{Program Specialization as a Tool\\for Solving Word Equations}
\author{Antonina Nepeivoda \\
\institute{Program Systems Institute of Russian Academy of Sciences\thanks{The reported study was partially supported by Russian Academy of Sciences, research project No.~AAAA-A19-119020690043-9.}
\\
  Pereslavl-Zalessky, Russia\\}\email{a\_nevod@mail.ru}
}
\def\titlerunning{Program Specialization as a Tool for Solving Word Equations}
\def\authorrunning{A. Nepeivoda} 
\usepackage{color}            

\newcommand{\skipper}{\hspace{-7pt}\textcolor{white}{!_{a_a}}}
\newcommand{\upskipper}{\hspace{-1pt}\textcolor{white}{!^!}}

\newcommand{\ie}{{\it i.e.}~}
\newcommand{\eg}{{\it e.g.}~}
\newcommand{\etc}{{\it etc.} }

\newcommand{\st}{{\it s.t.}~}

\newcommand{\wrt}{{\it w.r.t.}~}

\newcommand{\sectionname}{Sec.}

\newcommand{\Prog}{\mathsf{P}}
\newcommand{\SpecTask}{\mathfrak{M}}
\newcommand{\Node}{\mathsf{N}}

\newcommand{\FunGo}{\mathtt{Go}}
\def\LogLang{\rotatebox[origin=c]{5}{${\mathcal{L}}$}\hspace{-1.5pt}}
\def\IntLang{\rotatebox[origin=c]{5}{${\mathcal{F}}$}\hspace{-1pt}}
\newcommand{\FunG}{\mathtt{G}}
\newcommand{\FunH}{\mathtt{H}}
\newcommand{\FunEq}{\mathtt{Main}}
\newcommand{\FunSubst}{\mathtt{Subst}}
\newcommand{\FunSim}{\mathtt{Simplify}}
\newcommand{\FunSort}{\mathtt{CountEq}}
\newcommand{\FunSimPrim}{\mathtt{Reduce}}
\newcommand{\FunSplit}{\mathtt{Split}}
\newcommand{\FunCountMS}{\mathtt{CountMS}}
\newcommand{\FunAreEqual}{\mathtt{AreEqual}}
\newcommand{\FunElMinus}{\mathtt{CheckElement}}
\newcommand{\FunInclude}{\mathtt{Include}}
\newcommand{\FunSubstAll}{\mathtt{SubstAll}}
\newcommand{\FunAddElstoMS}{\mathtt{AddExprMS}}
\newcommand{\FunYieldCheck}{\mathtt{YieldCheck}}
\newcommand{\FunYieldCheckAux}{\mathtt{YieldCheckAux}}
\newcommand{\FunSubjEq}{\mathtt{CheckLengths}}
\newcommand{\FunSubtractEl}{\mathtt{CheckInclusion}}
\newcommand{\FunCheckInfo}{\mathtt{CheckInfo}}
\newcommand{\FunSplitRight}{\mathtt{SplitR}}
\newcommand{\FunCountMinus}{\mathtt{CmpNumbers}}
\newcommand{\Constdef}{\mathbf{CONST}\;}

\newcommand{\Fun}{\mathrm{F}}

\newcommand{\mscp}{\texttt{MSCP-A}}
\newcommand{\scp}{\mathtt{Spec}}
\newcommand{\Code}{\underline}
\newcommand{\conc}{\,}
\newcommand{\longconc}{\hspace{0.2ex}\raisebox{0.2ex}{{{+\hspace*{-0.3ex}+}}}\hspace{0.2ex}}
\newcommand{\MainStep}{\mathtt{Main}}

\newcommand{\lencode}{\llcorner}
\newcommand{\rencode}{\lrcorner}
\newcommand{\wholeencode}[1]{\raisebox{-0.4ex}{$\lencode\hspace{-1ex}\hspace{-0.2pt}\lencode$}\hspace{-0.8ex}#1\hspace{-0.8ex}\raisebox{-0.4ex}{$\rencode\hspace{-1ex}\hspace{-0.2pt}\rencode$}}
\newcommand{\wholeencodespace}[1]{\raisebox{-0.4ex}{$\lencode\hspace{-1ex}\hspace{-0.2pt}\lencode$}\hspace{-0.4ex}#1\hspace{-0.4ex}\raisebox{-0.4ex}{$\rencode\hspace{-1ex}\hspace{-0.2pt}\rencode$}}

\newcommand{\cA}{\mathbf{A}}
\newcommand{\cB}{\mathbf{B}}

\newcommand{\cX}{\mathbf{X}}
\newcommand{\cY}{\mathbf{Y}}

\newcommand{\pu}{\mathit{u}}
\newcommand{\pv}{\mathit{v}}

\newcommand{\svar}{\mathtt{s}}
\newcommand{\codex}{\mathbf{V}}
\newcommand{\vpar}{\mathit{v}}
\newcommand{\varx}{\mathtt{x}}
\newcommand{\vareq}{\mathtt{x}}
\newcommand{\vary}{\mathtt{y}}
\newcommand{\vvarv}{\mathtt{v}}
\newcommand{\varw}{\mathtt{w}}
\newcommand{\varz}{\mathtt{z}}
\newcommand{\vars}{\mathtt{s_{sym}}}
\newcommand{\vareqlist}{\mathtt{x_{val}}}
\newcommand{\varv}{\mathtt{x_{value}}}
\def\kavych{\hspace{0.2ex}\raisebox{-0.2ex}{\rotatebox[origin=c]{13}{${}^{\bm{\prime}}$}}}
\def\smkavych{\hspace{0.1ex}\raisebox{-0.4ex}{\rotatebox[origin=c]{13}{${}^{{\prime}}$}}}

\newcommand{\varres}{\mathtt{x_{result}}}
\newcommand{\varexpr}{\mathtt{x_{expr}}}
\newcommand{\varprev}{\mathtt{x_{prev}}}
\newcommand{\varpref}{\mathtt{xpref}}
\newcommand{\nnumber}{\mathtt{x_{\rm{num}}}}
\newcommand{\newnumber}{\mathtt{x_{\rm{freshnum}}}}
\newcommand{\varsuff}{\mathtt{xsuff}}
\newcommand{\varms}{\mathtt{xms}}
\newcommand{\varcount}{\mathtt{xnum}}
\newcommand{\othereq}{\mathtt{x_{eqs}}}
\newcommand{\varterm}{t}
\newcommand{\aarg}{\rm{Expr}}
\newcommand{\Pat}{\mathrm{P}}
\newcommand{\varr}{\mathtt{x_{rul}}}
\newcommand{\pref}{\mathrm{Pr}}
\newcommand{\suff}{\mathrm{S}}
\def\WeqInt{\mathrm{WI}_{\hspace{-0.2ex}\rotatebox[origin=c]{3}{${\mathcal{L}}$}}}
\def\WeqIntBase{\mathrm{WIBase}_{\hspace{-0.2ex}\rotatebox[origin=c]{3}{${\mathcal{L}}$}}}
\def\WeqIntSplit{\mathrm{WISplit}_{\hspace{-0.2ex}\rotatebox[origin=c]{3}{${\mathcal{L}}$}}}
\def\WeqIntSym{\mathrm{WICount}_{\hspace{-0.2ex}\rotatebox[origin=c]{3}{${\mathcal{L}}$}}}
\newcommand{\SolBase}{\mathrm{AlgWE_{Base}}}
\newcommand{\SolSplit}{\mathrm{AlgWE_{Split}}}
\newcommand{\SolCount}{\mathrm{AlgWE_{Count}}}
\newcommand{\vx}{\mathtt{x}}
\newcommand{\vy}{\mathtt{y}}
\newcommand{\vz}{\mathtt{z}}

\newcommand{\term}{t}

\newcommand{\VarSet}{\mathit{\mathcal{V}\hspace{-2.6pt}a\hspace{-0.5pt}r}}
\newcommand{\ParSet}{\mathit{\mathcal{P}\hspace{-2.4pt}a\hspace{-0.5pt}r}}
\newcommand{\Ground}{\mathit{\mathcal{G}\hspace{-2.3pt}r\hspace{-0.5pt}d}}
\newcommand{\eq}{E}
\newcommand{\Eq}{E}

\newcommand{\seq}{\;{=}\;}

\newcommand{\Eqs}{\mathcal{E}\hspace{-2pt}\mathit{qs}}
\newcommand{\ProcTree}{\rotatebox[origin=c]{10}{$\mathcal{T}$}}

\newcommand{\Alp}{\Sigma}
\newcommand{\CAlph}{\mathcal{A}}
\newcommand{\VAlph}{\mathcal{V}}
\newcommand{\empt}{\varepsilon}
\newcommand{\bottom}{\bot}
\newcommand{\OutTrue}{\mathbf{T}}
\newcommand{\OutFalse}{\mathbf{F}}
\newcommand{\OutNone}{\mathbf{N}}
\newcommand{\OutGreater}{\mathbf{G}}
\newcommand{\OutLesser}{\mathbf{L}}

\newcommand{\rar}{\mapsto}

\newcommand{\mindia}{\rule{0mm}{1.63ex}}    
\newcommand{\ov}[1]{*+[F-:<44pt>]{#1\mindia}}

\newtheorem{Prop}{Proposition}
\newtheorem{Property}{Property}
\newtheorem{Coroll}{Corollary}
\newtheorem{Definition}{Definition}
\newtheorem{Lemma}{Lemma}
\newtheorem{Example}{Example}

\def\logor{\mathrel{\vee}}
\def\logimpl{\mathrel{\Rightarrow}}
\def\logand{\mathrel{\&}}

\makeatletter
\newcommand\makebig[2]{%
  \@xp\newcommand\@xp*\csname#1\endcsname{\bBigg@{#2}}%
  \@xp\newcommand\@xp*\csname#1l\endcsname{\@xp\mathopen\csname#1\endcsname}%
  \@xp\newcommand\@xp*\csname#1r\endcsname{\@xp\mathclose\csname#1\endcsname}%
}
\makeatother

\makebig{biggg} {3.0}
\makebig{Biggg} {3.5}
\makebig{bigggg}{4.0}
\makebig{Bigggg}{14.7}

\begin{document}
\maketitle

\begin{abstract}
The paper focuses on the automatic generating of the witnesses for the word equation satisfiability problem by means of specializing an~interpreter $\WeqInt\hspace{-0.2ex}\left(\langle\sigma_i\rangle, \Eqs\right)$, which tests whether a composition of variable substitutions $\sigma_i$ of a given word equation system $\Eqs$ produces its solution. We specialize such an interpreter \wrt $\Eqs$, while $\sigma_i$ are unknown. We show that several variants of such interpreters, when specialized using the basic unfold/fold methods, are able to construct the whole solution sets for some classes of the word equations whose left- and right-hand sides share variables. We prove that the~specialization process \wrt the~constructed interpreters gives a simple syntactic criterion of the satisfiability of the equations considered, and show that the suggested approach can solve some equations not solvable by \texttt{Z3str3} and \texttt{CVC4}, the widely-used SMT-solvers.
\end{abstract}

\def\keywords{unfold/fold transformation, specialization, word equations, Nielsen transformation, supercompilation, Jones optimality} 

\section{Introduction}\label{Sect::Intro}

In recent decades, program transformation techniques were applied to verification and analysis of several computational models, including cache-coherence and cryptographic protocols, constrained Horn clauses, Petri nets, deductive databases, control-flow analysis, \etc\hspace{-0.4ex}\cite{Leuschel2008,Pettorossi2019,Gallagher2019,Hamilton15,NemReachability,NepeivodaArt,Vidal12}. On the\hspace{-0.2ex} other hand, the\hspace{-0.2ex} number of works on analysis of string manipulating programs and string constraint solvers is rapidly growing during last years\hspace{-0.2ex}~\cite{Abdulla,Bjorner,Ostrich,Day2019Woorjpe,Le2018,CVC4,Lin2018,Z3,Saxena,Trinh,Yu}. As far as we know, there are few interactions between the two research areas, although some of their methods exploit similar concepts.

One approach to the verification is to apply an~unfold/fold algorithm~\cite{Burst} to  model a~nondeterministic system behaviour by a~deterministic program via introducing an additional path parameter \cite{NemReachability,NemytykhCSR}. That is, given a~specialization algorithm $\scp$ and a non-deterministic program $f(\varx)$, the algorithm $\scp$ solves the specialization task $f'(\vpar,\varx)$ satisfying the condition\footnote{We use the assumption that only the elements $c$ belonging to the function domain are considered, which is expressed by the premise $\exists f(c)$.} $\forall c\left(\exists f(c)\logimpl \exists p\left(f'(p,c)\seq f(c)\right)\hspace{-0.3ex}\right)$. Here the parameter $\vpar$ ranges over the paths determining the ways to compute $f(\varx)$.   

This idea has many applications in computer science. In particular, methods to solve word equations, starting from Matyiasevich's~\cite{Matiyas}, Hmelevskij's~\cite{Hmelevsky}, Makanin's~\cite{Makanin77} algorithms in 1970s, and including Plandowski's~\cite{Plandowski} and Jez's algorithms~\cite{Jez} designed in the recent two decades, all use the~non-deterministic search. A word equation is an equation $\Phi\seq\Psi$, where $\Phi$ and $\Psi$ are finite words in the joint alphabet $\CAlph\cup\VAlph$ of letters and variables, its solution is a substitution $\sigma: \VAlph\rightarrow\CAlph^*$ \st $\Phi\sigma$ is textually equal to $\Psi\sigma$. All the algorithms provide transformation steps, which, applied iteratively to a given equation, generate its solution set. Thus, a (partial) solution tree of the given equation is produced. Some paths of the solution tree may be infinite, and are to be either pruned or represented as loops. For example, given a path in the solution tree and a node along this path, Matyiasevich's algorithm constructs an arc leading to this node from its descendant when the two equations labelling the nodes coincide. \figurename~\ref{fig:twodiags} shows a solution graph for a simple word equation generated by Matyiasevich's algorithm and a graph of states and the state transitions of a~functional program constructed by a basic unfold/fold algorithm. Henceforth, we refer to such graphs as (partial) process trees~\cite{Pettorossi96,Sorensen95}. One may observe that the two graphs coincide modulo the node and arc labels, sharing the general structure, although they are constructed for the different purposes. This paper focuses on the similarity of the methods and aims at adjusting the unfold/fold program transformation algorithm to solving the~word equations.

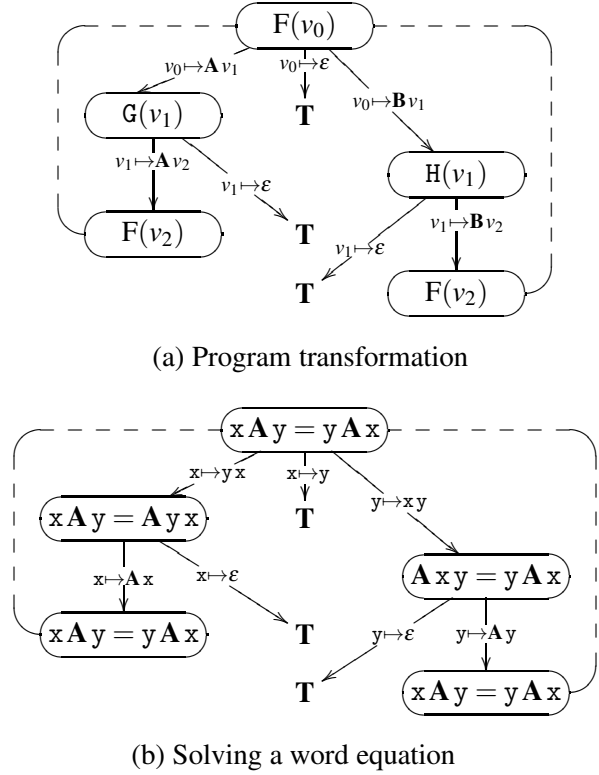
\begin{figure}[t]\centering
\begin{minipage}{0.7\textwidth}
$$
{
\xymatrix @-6.5mm {
&\ov{\quad\Fun(\vpar_0)\quad}\ar[ldd]!<-20pt,3pt>|<(.3){{\skipper\vpar_0\rar \cA\conc\vpar_1\upskipper}}\ar[rddd]|{\upskipper\skipper\vpar_0\rar\cB\conc\vpar_1}\ar[dd]|>(.55){\vpar_0\rar \empt}&&\\
& \\
\ov{\quad\FunG(\vpar_1)\quad}\ar[dd]|>(.48){\vpar_1\rar \cA\conc\vpar_2}\ar[rdd]|<(.48){\upskipper\vpar_1\rar\empt} &*+{\OutTrue}&&\\
&&\ov{\quad\FunH(\vpar_1)\quad}\ar[dd]|>(.48){\upskipper\upskipper\vpar_1\rar \cB\conc\vpar_2}\ar[ddl]|<(.55){\vpar_1\rar\empt}&\\
\ov{\quad\Fun(\vpar_2)\quad}\ar@{--}`l[uuuu]`[uuuurr][uuuur]&*+{\OutTrue}&\\
&*+{\OutTrue}&\ov{\quad\Fun(\vpar_2)\quad}\ar@{--}`r[uuuuu]`[uuuuull][uuuuul]
}}$$

\vspace{4pt}
\centering\;(a) Program transformation
\end{minipage}

\begin{minipage}{0.75\textwidth}
\vspace{10pt}
$$
{
\xymatrix @-6.5mm {
&\ov{\varx\conc\cA\conc\vary\seq\vy\conc\cA\conc\vx}\ar[ldd]|{{\skipper\vx\rar \vy\conc\vx}}\ar[rddd]|{\skipper\vy\rar \vx\conc\vy}\ar[dd]|<(.35){\skipper\vx\rar \vy}&&\\
& \\
\ov{\vx\conc\cA\conc\vy\seq\cA\conc\vy\conc\vx}\ar[dd]|{\upskipper\vx\rar \cA\conc\vx\upskipper}\ar[rdd]|{\skipper\vx\rar\empt} &*+{\OutTrue}&\\
&&\ov{\cA\conc\vx\conc\vy\seq\vy\conc\cA\conc\vx}\ar[dd]|{\vy\rar \cA\conc\vy}\ar[ddl]|{\skipper\vy\rar\empt\upskipper}&\\
\ov{\vx\conc\cA\conc\vy\seq\vy\conc\cA\conc\vx}\ar@{--}`l[uuuu]`[uuuurr][uuuur]&*+{\OutTrue}&\\
&*+{\OutTrue}&\ov{\vx\conc\cA\conc\vy\seq \vy\conc\cA\conc\vx}\ar@{--}`r[uuuuu]`[uuuuull][uuuuul]&
}}$$

\vspace{4pt}
\qquad\qquad\qquad\qquad\quad\; (b) Solving a word equation
\end{minipage}
\caption{Unfold/fold method for solving the word equations and program transformation.}
\label{fig:twodiags}
\end{figure} 

Our contributions are the following. 
\begin{enumerate}
\item We study three interpreters, specialization of which using a~general-purpose tool \wrt a given word equation constructs a~residual program presenting the complete solutions set of this word equation.

\item We prove that an~analogue of Jones-optimality holds for the interpreters considered~\cite{BenAbram,JonesBook}, which guarantees that if the specialization terminates, then the residual programs represent complete solution sets of the given equation systems. Surprisingly, the naive unfolding plus some basic optimisations generate a solution algorithm, for example, for the set of the one-variable word equations, and the algorithm differs from the well-known one given by Hmelevskij~\cite{Hmelevsky}. 

\item We also reveal several other classes of equations sharing variables in right- and left-hand sides, for which the specialization is proved to always terminate. To the best of our knowledge, these classes of equations were not covered by the published works on the string constraint solvers applied to the unbounded-length case. 

\item Finally, we show the application results of the presented approach to the benchmark equation sets developed for the solver \texttt{Woorpje}~\cite{Day2019Woorjpe}, and to a new benchmark of 50 equation systems, and compare the results with the results of the application of the SMT-solvers \texttt{Z3str3} and \texttt{CVC4} to these benchmarks. The systems were generated by the authors of the benchmarks randomly, and the complete solution sets are constructed by our algorithm for the most of the tests considered. Our algorithm is slow, as compared to the algorithms implemented in the SMT-solvers, when used on satisfiable equations because the algorithm finds all the solutions instead of at least one, however for the equations having no solution, our algorithm shows better success rate. 
\end{enumerate}

The remainder of this paper is structured as follows. In \sectionname~\ref{Sect:Language}, we introduce the presentation syntax. We describe the interpreters used in \sectionname~\ref{Sect:WeqInt}, and the verification task in \sectionname~\ref{Sect:VerTask}. The general unfold/fold scheme\footnote{We do not consider the generation of a residual program, since the correctness of the residual programs is provided by the properties of the process graph.} is given in \sectionname~\ref{Sect:SpecScheme}. In \sectionname~\ref{Sect:VerRes}, we discuss the optimality of the specialization and present results of the verification, in particular, for a number of sets of the word equations we show that every equation in any of the sets can be solved via the verification method. \sectionname~\ref{Sect:Discuss} considers related work, and \sectionname~\ref{Sect::Conclusion} concludes the paper. The proofs of the main properties of the presented algorithms are given in Appendix, as well as the source code of the interpreter models used in the specialization.

We assume that the reader is familiar with the basic notions of the program specialization, \st partial evaluation, partial deduction, supercompilation, \etc

\section{Preliminaries}\label{Sect:Language}

We denote the set of the string variables with $\VAlph$, the finite set of the letters with $\CAlph$. We assume that the bold capitals $\cA$, $\cB$ denote the letters in $\CAlph$; while the typographic small letters $\vx$, $\vy$, $\vz$ denote the variables in $\VAlph$. A term is an element of $\CAlph\cup\VAlph$. A word equation is an equation $\Phi\seq\Psi$, where $\Phi,\Psi\in\{\CAlph\cup\VAlph\}^*$. A word equation is said to be \emph{reduced} iff its sides neither start nor end with the same terms~\cite{KarhumCombWo}. 

We write an application of substitution $\xi\colon\VAlph\rightarrow\{\CAlph\cup\VAlph\}^*$ to a word $\Phi$ as $\Phi\xi$. A solution of the~equation $\Phi\seq\Psi$ is a substitution $\xi\colon\VAlph\rightarrow\CAlph^*$ \st $\Phi\xi$ and $\Psi\xi$ coincide textually. Given an equation $\Eq\colon\Phi\seq\Psi$ and substitution $\xi$, $\Eq\xi$ is $\Phi\xi\seq\Psi\xi$, by default, in the reduced form. We denote the number of occurrences of the term $\term$ in $\Phi$ with $|\Phi|_{\term}$. The equation length is $|\Phi|+|\Psi|$, $\empt$ denotes the empty word.

\subsection{Simple Logic Language}
\label{Subsect::LogLang}

Our method is based on an analysis of programs written in a simple logic language $\LogLang$ over the dataset consisting of the word equations. An $\LogLang$ program is a list of narrowings representing the variable substitutions that are to be applied to the word equations. We distinguish between the narrowings in the language $\LogLang$ and the narrowings occurring in the specialization process (\sectionname~\ref{Sect:SpecScheme}). Hence, we call the former $\LogLang$-narrowings, and the latter --- parameter narrowings (or par-narrowings). 

A word equation $\Phi\seq\Psi$ is encoded with a\hspace{-0.2ex} pair $(\Phi,\Psi)$. An\hspace{-0.3ex} $\LogLang$-narrowing is encoded with a\hspace{-0.2ex} string using the\hspace{-0.2ex} sign $\rar$. There are three possible forms of the\hspace{-0.2ex} elementary\hspace{-0.3ex} $\LogLang$-narrowings, corresponding to those used in the\hspace{-0.2ex} Matiyasevich algorithm. We consider these\hspace{-0.3ex} $\LogLang$-narrowing sequences as the\hspace{-0.2ex} programs in a\hspace{-0.2ex} simple acyclic logic programming language\hspace{-0.3ex} $\LogLang$. Thus, every statement in an\hspace{-0.3ex} $\LogLang$-program is an encoded $\LogLang$-narrowing. The syntax of the encoded equations \texttt{Eqs} and of the\hspace{-0.3ex} $\LogLang$-programs \texttt{Narrs} is given in~\figurename~\ref{tab:datasyntax}. 
\begin{figure}[ht]
\centering
\addtolength{\tabcolsep}{-4pt}
{\begin{tabular}{c|c}
Data & Programs \\\hline\small
\begin{tabular}{rllllll}
\textit{Eqs} & ::= & \textit{Eq\,\,Eqs} & $\vert$ & $\skipper\empt$ \\
\textit{Eq} & ::= & \multicolumn{2}{l}{\texttt{(}\textit{Side}\texttt{,} \textit{Side}\texttt{)}} \\
\textit{Side} & ::= & \textit{Char\,\,Side} & $\vert$ & \textit{Var\,Side} & $\vert$ & $\empt$
\end{tabular}
&\small
\begin{tabular}{rllllll}
\textit{Narrs} & ::= & \texttt{(}\textit{Narr}\texttt{)}\,\textit{Narrs} & $\vert$ & $\empt$ \\
\textit{Narr} & ::= & \kavych \textit{Var}\,$\rar$ \textit{Char\,\,Var}\kavych &$\vert$ & \kavych \textit{Var}\,$\rar$\textit{Var}$_1$ \textit{Var}\kavych &$\vert$&\kavych \textit{Var}\,$\rar\empt$\kavych 
\end{tabular}
\end{tabular}
}
\addtolength{\tabcolsep}{4pt}
\vspace{-1ex}
\caption{The syntax of the data and the programs.}
\label{tab:datasyntax}
\end{figure}

\vspace*{-2.5ex}
\noindent Here \textit{Var}, \textit{Var}$_1$\,$\in\VAlph$, \textit{Var}$\,\neq\,$\textit{Var}$_1$, \textit{Char}$\,\in\CAlph$. Let $\langle\Phi_i\seq\Psi_i\rangle_{i=1}^n$ be syntactic sugar for ($\Phi_1$, $\Psi_1$)\dots($\Phi_n$, $\Psi_n$). The $\LogLang$-narrowings sequence $(\sigma_1)\dots(\sigma_m)$ is also written as $\langle\sigma_i\rangle_{i=1}^m$.

\subsection{Interpreters' Source Pseudocode Language}
\label{Subsect::IntLang}

The interpreters considered are written in the following pseudocode for functional programs manipulating the strings and based on the pattern matching. The $\IntLang$ programs are lists of term rewriting rules. The rules in the definitions are applied using the top-down matching order. The syntax of the language $\IntLang$ is given in~\figurename~\ref{tab:syntax}. Here $\empt$ is the empty word, $\longconc$ stands for the associative concatenation constructor (both may be omitted). The~set of the~constants used as the letters\footnote{This set is wider than the set $\CAlph$, \sectionname~\ref{Sect:Language}, because it contains also the letters used in the inner encoding of the equations.} is $\Alp$, elements of which are given in bold, $\CAlph\subset \Alp$. The variables in the $\IntLang$-program rules range either over expressions or over letters. Henceforth we call these variables $\IntLang$-variables or pattern variables, in order to distinguish them from the variables occurring in the word equations.

\begin{figure}[h]
 \small\centering
\addtolength{\tabcolsep}{-3.5pt}
{\begin{tabular}{rlclclclclclc}
\textit{Rule} &::=$\upskipper^{\upskipper}$ & \multicolumn{6}{c}{\textit{FName}\hspace{0.1pt}\texttt{(}\textit{Pattern, \dots, Pattern}\texttt{)}\, =\, \textit{Exp}} \\
\textit{Pattern} &::=$\upskipper^{\upskipper}$ &\hspace*{-20pt}$\empt$ &\hspace*{-2pt}$\vert$&\textit{Variable}&$\vert$&\textit{Letter}&$\vert$&\hspace*{-30pt}\texttt{(}\textit{Pattern}\texttt{)}&
$\vert$&\textit{Pattern} $\longconc$ \textit{Pattern}\\
\textit{Exp} &::=$\upskipper^{\upskipper}$ &\hspace*{-20pt}$\empt$ &\hspace*{-2pt}$\vert$&\textit{Variable}&$\vert$&\textit{Letter}&$\vert$&\hspace*{-30pt}\texttt{(}\textit{Exp}\texttt{)}&$\vert$& \textit{FName}\hspace{0.1pt}\texttt{(}\textit{Exp, \dots, Exp}\texttt{)} &$\vert$&\textit{Exp} $\longconc$ \textit{Exp}\\
\textit{Variable} &::=$\upskipper^{\upskipper}$ &\hspace*{-10pt}$\vx$\textit{Name}&\hspace*{-2pt}$\vert$&\texttt{s}\textit{Name}
\end{tabular}}
\addtolength{\tabcolsep}{3.5pt}
\vspace{-1ex}
\caption{The syntax of the program pseudocode.}
\label{tab:syntax}
\end{figure}

\vspace*{-0.5ex}
\noindent An \emph{object} expression is either a string in $\Alp^*$, concatenation of two object expressions or (\textit{Exp}), where \textit{Exp} is an object expression. The $\IntLang$-variables with the first letter $\vx$ range over the object expressions; the~$\IntLang$-variables with the first letter $\texttt{s}$ range over the symbols in $\Alp$. We denote the set of the $\IntLang$-variables occurring in the expression \textit{Exp} with $\VarSet($\textit{Exp}$)$. Given an $\IntLang$-program, the function $\FunGo$ serves as its entry function. The delimiters \texttt{/\hspace{-0.2ex}*}\dots\texttt{*\hspace{-0.2ex}/} stand for the comments. The~semantics of the~programming language $\IntLang$ is based on the call-by-value evaluation strategy.

\section{Word Equations' Interpreters}\label{Sect:WeqInt}

In this section we introduce informally a class of simple interpreters taking $\LogLang$-programs and applying them to the lists of word equations. Given such an interpreter $\WeqInt$, a program ${\langle\sigma_i\rangle}_{i=1}^m$, and a sequence of equations $\langle\Phi_i\seq\Psi_i\rangle_{i=1}^n$, the call $\WeqInt\hspace{-0.2ex}\left(\langle\sigma_i\rangle_{i=1}^m,\,\langle\Phi_i\seq\Psi_i\rangle^n_{i=1}\right)$ returns $\OutTrue$ iff the following two conditions hold. The notions of the compatibility of an $\LogLang$-narrowing with an equation list and of the operation $\FunSim$ are given further.
\begin{itemize}
\item Every $\LogLang$-narrowing $\sigma_i$, where $1\leq i\leq m$, is compatible with the equation list resulting from the call $\FunSim$($\langle\Phi_i\sigma_1\dots\sigma_{i-1}\seq\Psi_i\sigma_1\dots\sigma_{i-1}\rangle^n_{i=1}$).
\item And for every $i$, $1\leq i\leq n$, $\Phi_i\sigma_1\dots\sigma_m$ textually coincides with $\Psi_i\sigma_1\dots\sigma_m$. 
\end{itemize}

The call $\WeqInt\hspace{-0.2ex}(\empt,\empt)$ returns $\OutTrue$. Otherwise, the call $\WeqInt\hspace{-0.2ex}\left(\langle\sigma_i\rangle_{i=1}^m,\,\langle\Phi_i\seq\Psi_i\rangle_{i=1}^n\right)$ returns $\OutFalse$.   

All the $\LogLang$-interpreters share the same structure: they take the first program statement, apply it to the equation list, and then call a simplification function $\FunSim$ that transforms the current resulting equation list to an equation list with the same solution set, and having a simpler form (see \figurename~\ref{tab:interprgenstr}). The~function $\FunSim$ is the only source part that depends on the concrete interpreter considered.

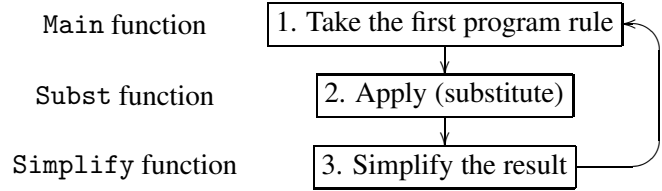
\begin{figure}[ht]
$$
\xymatrix@-4.5mm {
*{\raisebox{-0.7ex}{$\FunEq$}\txt{ function}}&*+[F-]{\txt{1. Take the first program rule}}\ar@{->}[d]& \\
*{\raisebox{-0.7ex}{$\FunSubst$}\txt{ function}}&*+[F-]{\txt{2. Apply (substitute)}}\ar@{->}[d]&\\
*{\raisebox{-0.7ex}{$\FunSim$}\txt{ function}}&*+[F-]{\txt{3. Simplify the result}}\ar@{->}`r[ur]`[uu][uu]
}
\vspace{-1ex}
$$
\caption{General structure of the interpreters.}
\label{tab:interprgenstr}
\end{figure}

We can also say that such an interpreter applies a list of the $\LogLang$-narrowings to an equation, choosing for this purpose a path to one its solution. The solution tree is a tree describing the set of all the solution paths, the tree nodes are labelled with equation lists\footnote{In order to emphasize that the the graph depends on the order of the equations in the equation system, we use ``the list of the equations'' instead of ``the system \dots''.}. Such a tree can be constructed as a (possibly infinite) directed graph representing a non-deterministic unfolding process using the substitutions listed in \figurename~\ref{tab:classicsubstrules}~(a). We see the $\LogLang$-narrowings in the first column. Given an $\LogLang$-narrowing, the second column provides the~constraint imposed on the first equation in the list, required for applying this narrowing. Thus, an $\LogLang$-narrowing $\sigma$ is compatible with the list $\langle\Phi_i\seq\Psi_i\rangle_{i=1}^n$ iff $\sigma$ is generated by the constraint imposed on the first equation in the list $\Phi_1\seq\Psi_1$ given in~\figurename~\ref{tab:classicsubstrules}~(a). Following the classical approach~\cite{KarhumCombWo,Hmelevsky,Matiyas}, no fresh variables and constants are introduced in the~$\LogLang$-narrowing rules.

\begin{figure}[ht]
\begin{minipage}{0.46\textwidth}
\small
\begin{tabular}{llll}
\hline
&\multirow{2}{*}{$\varx\rar\empt$} & \qquad$\varx\conc\Phi\seq\Psi^{\upskipper}$& 
\\ 
&&\qquad or \;$\Phi\seq\varx\conc\Psi$& \vspace*{2pt}\\\hline
&$\varx\rar t\conc\varx$& \qquad$\varx\conc\Phi\seq t\conc\Psi^{\upskipper}$& \\
&($t\in\CAlph$)&\qquad or \;$t\conc\Phi\seq\varx\conc\Psi$& \vspace*{2pt}\\\hline
&\multirow{2}{*}{$\varx\rar\vary\conc\varx$}& \qquad $\varx\conc\Phi\seq\vary\conc\Psi^{\upskipper}$& \\
&&\qquad or \;$\vary\conc\Phi\seq\varx\conc\Psi$& \vspace*{2pt}\\\hline
\end{tabular}

\vspace*{12pt}
\begin{flushleft}
\normalsize(a) The $\LogLang$-narrowings and equations compatible with them.
\end{flushleft}
\end{minipage}
\begin{minipage}{0.53\textwidth}
\vspace*{-25pt}
$${
\xymatrix@-4.5mm {&&&&&\\
&& *++[F-:<44pt>]{\cA\conc\vx\conc\vy\seq \vx\conc\vy\conc\cA}\ar[dl]|{\skipper\vx\rar\empt}\ar@{->}`d[dr]`[ur]|{\upskipper\vx\rar\cA\conc\vx\upskipper}`[l][]&\\
&*++[F-:<44pt>]{\cA\conc\vy\seq\vy\conc\cA}\ar[dl]|<(.35){\skipper\vy\rar\empt}\ar@{->}`d[dr]`[ur]!<-25pt,0pt>|{\upskipper\vy\rar\cA\conc\vy\upskipper}`[l][]& &\\
*+{\OutTrue}&&&
}}
$$

\vspace*{8pt}
\qquad(b) The solution graph of $\cA\conc\vx\conc\vy\seq\vx\conc\vy\conc\cA$.
\end{minipage}
\caption{The narrowings cases and an example of a solution graph.}
\label{tab:classicsubstrules}

\end{figure}
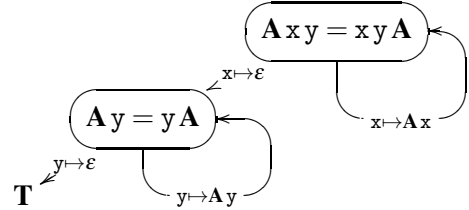

We use a slightly modified version of Matiyasevich's algorithm. The~Nielsen transformation~\cite{Day2017RegOrd}, which is the base of the algorithm, states that given $\varx\conc\Phi\seq\vary\conc\Psi$, we can replace either $\varx$ with $\vary\conc\varx$ if the length of the value of the variable $\varx$ is greater than the length of the value of $\vary$, or vice versa if the length of $\vary$ value is greater, or $\varx$ with $\vary$ if their value lengths are equal. \figurename~\ref{tab:classicsubstrules}~(a) does not show the last case: it is a composition of the substitutions $\varx\rar\vary\conc\varx$ and $\varx\rar\empt$, since we allow the substitution $\varx\rar\empt$ to be compatible with any equation whose left- or right-hand side starts with $\varx$. This modification guarantees that any solution to the equation $\Phi\seq\Psi$ can be generated\footnote{The idea behind the algorithm originated by Matiyasevich is aimed at deciding the solvability of an equation, rather than at constructing the whole set of the solutions.} as a composition of the elementary $\LogLang$-narrowings given in  \figurename~\ref{tab:classicsubstrules}~(a). The following example shows that for the classical version of the algorithm, this statement does not hold.

\begin{Example}\label{Example::Nielsen}
Let $\varx\sigma\seq\cA$, $\vary\sigma\seq\empt$, and the equation $\varx\conc\vary\seq\vary\conc\varx$ be considered. If we use the narrowings $\varx\rar\vary\conc\varx\logor \vary\rar\varx\conc\vary\logor\varx\rar\vary$ provided by the classic form of the Nielsen transformation then the solution generated by $\sigma$ cannot be obtained by any finite number of such narrowings.  
\end{Example}

We assume that an infinite path in a solution tree is to be folded iff it contains nodes $\Node_1$ and $\Node_2$ \st $\Node_1$ is an ancestor of $\Node_2$ and their labels textually coincide. Then the subtree with the~root $\Node_2$ repeats the subtree with the root $\Node_1$, and the infinite path can be represented with a cycle, thus the tree is represented with the graph. Henceforth we use almost interchangeably these two notions. An~example of a graph representing all solutions of the equation $\cA\conc\vx\conc\vy\seq\vx\conc\vy\conc\cA$ is given in  \figurename~\ref{tab:classicsubstrules}~(b). The cyclic arcs in the graph show the folding operation. For the sake of brevity, the nodes along the~folded paths are not shown. 

An~interpreter $\WeqInt\hspace{-0.2ex}\left(\Prog,\langle\eq_1\dots,\eq_n\rangle\right)$ takes a list of the equations $\eq_i$ and a logic program $\Prog$ being a list of the $\LogLang$-narrowings that have to be successively applied to the equations. If the composition of the substitutions given in $\Prog$ transforms all the equations in the list to the tautologies, then $\WeqInt\hspace{-0.2ex}\left(\Prog,\langle\eq_1\dots,\eq_n\rangle\right)$ returns the~value $\OutTrue$. If a substitution is not compatible with the current equations list or the list of substitutions is empty whereas the equations are not tautologies then $\WeqInt\hspace{-0.2ex}\left(\Prog,\langle\eq_1\dots,\eq_n\rangle\right)$ results in the value $\OutFalse$. Given any input equation list, the interpreter does at most $m$ steps shown in~\figurename~\ref{tab:interprgenstr}, where $m$ is the number of the $\LogLang$-narrowings in the list $\Prog$, and always returns either value $\OutTrue$ or value $\OutFalse$. 

The general structure of the simplification functions used in the three interpreters that we consider in this paper is given in~\figurename~\ref{tab:intscheme}. The transformations shown in this figure are also used for constructing the corresponding solution graphs. According to  \figurename~\ref{tab:intscheme} we say that the interpreter $\WeqIntBase$ models paths in a solution tree based on the scheme $\SolBase$; and so on. See  \figurename~\ref{fig:solscheme} for examples of the corresponding solution graphs.

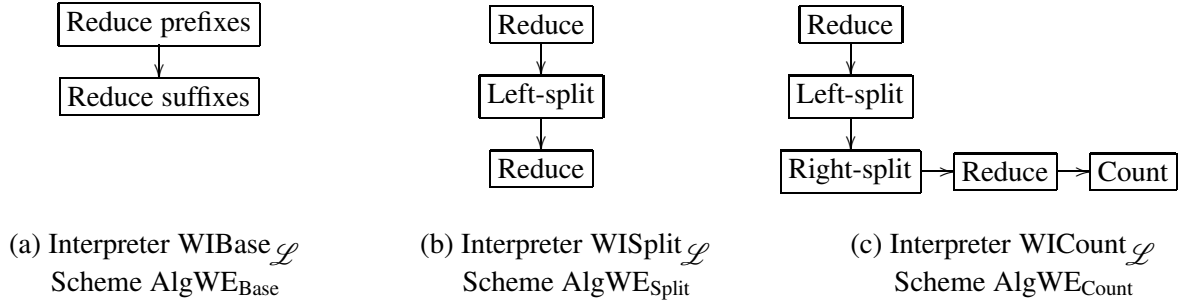
\begin{figure}[t]
\noindent\begin{tabular}{ccc}
$$
\xymatrix@-4mm {
*+[F-]{\txt{Reduce prefixes}}\ar@{->}[d] \\ *+[F-]{\txt{Reduce suffixes}}}
$$
&\quad
$$
\xymatrix@-4mm {
*+[F-]{\txt{Reduce}}\ar@{->}[d] \\ *+[F-]{\txt{Left-split}}\ar[d] \\ *+[F-]{\txt{Reduce}}}
$$\quad
&\quad
$$
\xymatrix@-4mm {
*+[F-]{\txt{Reduce}}\ar@{->}[d] \\ *+[F-]{\txt{Left-split}}\ar[d] \\ *+[F-]{\txt{Right-split}}\ar[r] & *+[F-]{\txt{Reduce}}\ar@{->}[r] & *+[F-]{\txt{Count}}}
$$\quad\\ \\ 
(a) Interpreter $\WeqIntBase$ & \qquad\quad (b) Interpreter $\WeqIntSplit$ & \qquad\quad (c) Interpreter $\WeqIntSym$ \\
\;\;\,Scheme $\SolBase$ & \qquad\quad \;\;\,Scheme $\SolSplit$ & \qquad\quad \;\;\,Scheme $\SolCount$
\end{tabular}
\caption{Schemes for simplifying the word equations.}
\label{tab:intscheme}
\end{figure}

\subsection{Basic Interpreter}\label{subsect:baseint} 
The basic interpreter $\WeqIntBase(\Prog,\Phi\seq\Psi)$ manipulates a single equation and accepts as the first input a list $\Prog$ of the elementary $\LogLang$-narrowings. When an $\LogLang$-narrowing $\sigma$ from the list $\Prog$ is applied to $\Phi\seq\Psi$, the simplification function immediately removes equal prefixes and suffixes of $\Phi\sigma$ and $\Psi\sigma$. Thus, the function $\FunSim$ constructs the reduced form of the equation $(\Phi\seq\Psi)\sigma$. Actually, this basic interpreter models the classic algorithm for solving the word equations suggested by Matiyasevich. We treat the interpreter $\WeqIntBase$ as a base for developing more complex ones.

\subsection{Splitting Interpreter}\label{subsect::splitint}

The interpreter $\WeqIntSplit\hspace{-0.2ex}\left(\Prog,\langle\eq_1,\dots,\eq_n\rangle\right)$ manipulates the lists of equations (representing equation systems) rather than a single equation. Thus, the substitution and simplification functions are applied to every equation in the list. If construction of the~reduced form of some equation in the list results in an~equation $\term_1\conc\Phi_i\seq\term_2\conc\Psi_i$, where $\term_1\in\CAlph$, $\term_2\in\CAlph$, $\term_1\neq\term_2$, the equation is immediately replaced with the trivial contradictory equation $\term_1\seq\term_2$, and all the other equations in the list are removed. The operation looking for trivial contradictions is also a part of the algorithm deriving the reduced form. 

A natural way to manually simplify an equation is to split it using the length argument~\cite{KarhumCombWo}. \textit{E.g}.~given equation $\varx\conc\cA\conc\vary\conc\cB\seq\cA\conc\varx\conc\varx\conc\varx$ we can split it into the list of $\Eq_1: \varx\conc\cA\seq\cA\conc\varx$ and $\Eq_2: \vary\seq\varx\conc\varx$, and the system represented by the list of these two equations has the same solution set as the initial equation.

Let us describe more formally the mentioned method. Given $\Phi$ and $\Psi$ in $\{\CAlph\cup\VAlph\}^+$, \st $|\Phi|\seq |\Psi|$ and for every $\varx\in\VAlph$ the equality $|\Phi|_{\varx}\seq |\Psi|_{\varx}$ holds, we say that the words $\Phi$ and $\Psi$ are \emph{variable-permutated} (briefly, var-permutated). The following proposition is trivial.

\begin{Prop}\label{Prop:split1}

Let an\hspace{-0.2ex} equation be of the\hspace{-0.2ex} form $\pref_1\conc \suff_1 \seq \pref_2\conc \suff_2$, where $\pref_i,\suff_i\in\{\CAlph\cup\VAlph\}^*$. Then the\hspace{-0.2ex} equation is equivalent to the\hspace{-0.2ex} system $\pref_1\seq \pref_2\logand \suff_1\seq \suff_2$ if at least one of these two statements holds.
\begin{enumerate}
\item the prefixes $\pref_1$ and $\pref_2$ are non-empty and var-permutated;
\item the suffixes $\suff_1$ and $\suff_2$ are non-empty and var-permutated.
\end{enumerate}
\end{Prop}

\noindent If Proposition~\ref{Prop:split1} is applied to the prefixes of an equation we say that the equation is left-split; if it is applied to the suffixes we say the equation is right-split.

The simplification function of the interpreter $\WeqIntSplit$ uses Proposition~\ref{Prop:split1} \wrt the var-permutated prefixes. Given an equation list $\langle\eq_1,\dots,\eq_n\rangle$, let the substitution $\sigma$ be applied. Then the function $\FunSim$ first reduces all the equations in the list $\langle\eq_1\sigma,\dots,\eq_n\sigma\rangle$. For every resulting equation $\eq'_i$, the simplification algorithm tries then to find the shortest non-empty var-permutated prefixes $\pref_{1,i,1}$ and $\pref_{2,i,1}$ of its left- and right-hand sides. If that succeeds, the equation $\eq'_i$ is split into the two equations $\pref_{1,i,1}\seq \pref_{2,i,1}$ and $\suff_{1,i,1}\seq \suff_{2,i,1}$. Here the third index is the number of the splitting iterations. Then the simplification algorithm construct the reduced form of the equation $\suff_{1,i,1}\seq \suff_{2,i,1}$, and tries to left-split it, \etc until $\suff_{1,i,j}$ and $\suff_{2,i,j}$ have no non-empty var-permutated prefixes (see \figurename~\ref{tab:intscheme}~(b)). The initial equation $\eq'_i$ is replaced with the generated equations resulting from $k$ successful left-split operations. The equations in the updated list are ordered as follows: $\suff_{1,i,k}\seq\suff_{2,i,k},\pref_{1,i,1}\seq\pref_{2,i,1},\dots,\pref_{1,i,k}\seq\pref_{2,i,k}$. The simplification function used in $\WeqIntSplit$ does not change the~resulting equations $\pref_{1,i,k}\seq\pref_{2,i,k}$, since they are in the reduced form by the construction.
 
\begin{Example}

Given $\sigma:\varx\rar\empt$ and the list $\langle\Eq_1,\Eq_2\rangle\seq\langle\varx\conc\varz\conc\varx\conc\cB\conc\vary\seq\cA\conc\varz\conc\varz,\,\vary\conc\varx\conc\cB\conc\varz\conc\vary\seq\cA\conc\vary\conc\varz\conc\varz\conc\varz\rangle$, the equation $\Eq_1\sigma$ is $\varz\conc\cB\conc\vary\seq\cA\conc\varz\conc\varz$, thus it is split into the equations $\varz\conc\cB\seq\cA\conc\varz$ and $\vary\seq\varz$, which replace $\Eq_1\sigma$ in the order $\langle\vary\seq\varz,\,\varz\conc\cB\seq\cA\conc\varz\rangle$,  where the non-var-permutated suffixes are written first. The equation $\Eq_2\sigma: \vary\conc\cB\conc\varz\conc\vary\seq\cA\conc\vary\conc\varz\conc\varz\conc\varz$ is first split into the equations: $\vary\conc\cB\seq\cA\conc\vary$ and $\varz\conc\vary\seq\varz\conc\varz\conc\varz$. The second equation is then transformed to the reduced form $\vary\seq\varz\conc\varz$ and cannot be split any more. These two equations replace the equation $\Eq_2\sigma$ in the following order: $\langle\vary\seq\varz\conc\varz,\,\vary\conc\cB\seq\cA\conc\vary\rangle$. The resulted list of the simplified equations is $\langle\vary\seq\varz,\,\varz\conc\cB\seq\cA\conc\varz,\,\vary\seq\varz\conc\varz,\,\vary\conc\cB\seq\cA\conc\vary\rangle$.

\end{Example} 
 
\subsection{Counting Interpreter}

The third variant of our interpreter uses the following well-known simple observation.
\begin{Prop}\label{Prop:count1}
Given an equation $\Phi\seq\Psi$, let for every $\varx\in\VAlph$ $|\Phi|_{\varx}\geq |\Psi|_{\varx}$, and $\displaystyle\sum_{\term_i\in\CAlph} |\Phi|_{\term_i}>\displaystyle\sum_{\term_i\in\CAlph} |\Psi|_{\term_i}$. 

\vspace{-2ex}
\noindent Then the equation $\Phi\seq\Psi$ has no solution.
\end{Prop}

After constructing the reduced form of the equations, the simplification function of the interpreter $\WeqIntSym$ tries to construct their left-splits (as $\WeqIntSplit$ does), and then to construct the right-splits of their suffixes resulted from the left-splits. Finally, $\WeqIntSym$ checks the property stated in Proposition~\ref{Prop:count1} of the resulting equations with non-var-permutated sides. The interpreter $\WeqIntSym$ has been used in the most benchmark tests (Section~\ref{Sect:VerRes}). 

\begin{figure}[ht]
$$
{\xymatrix@-3mm {
&&\ov{\varx\conc\varx\conc\cA\conc\vary\conc\cB\conc\varz\seq\cA\conc\varx\conc\varx\conc\varz\conc\vary}\ar[dl]|{\skipper\varx\rar\empt}
\ar[dr]|{\skipper\varx\rar\cA\conc\varx}
\\
&\ov{\vary\conc\cB\conc\varz\seq\varz\conc\vary}\ar[d]|<(.3){\skipper\varz\rar\vary\conc\varz}
\ar@{->}`l[ul]+<10pt,-0pt>`[ur]+<10pt,-10pt>|{\upskipper\vy\rar\vz\conc\vy\upskipper}`[d][] &&\ov{\begin{array}{c}\varx\conc\cA\conc\varx\conc\cA\conc\vary\conc\cB\conc\varz \\ 
\seq\cA\conc\varx\conc\cA\conc\varx\conc\varz\conc\vary\end{array}}
\ar[dl]|{\skipper\vx\rar\empt}\ar[d]|{\skipper\vx\rar\cA\conc\vx}
\\
&\ov{\cB\conc\vary\conc\varz\seq\varz\conc\vary}\ar@{->}@/^30pt/[]!<5pt,-15pt>;[u]!<-5pt,15pt>|{\skipper\varz\rar\cB\conc\varz}&\ov{\vary\conc\cB\conc\varz\seq\varz\conc\vary}\ar@{..}@/^5pt/[ul]&\ov{\begin{array}{c}\varx\conc\cA\conc\cA\conc\varx\conc\cA\conc\vary\conc\cB\conc\varz \\ 
\seq\cA\conc\varx\conc\cA\conc\cA\conc\varx\conc\varz\conc\vary\end{array}}
\ar[dl]|<(.3){\skipper\vx\rar\empt}\ar[d]|<(.3){\skipper\vx\rar\cA\conc\vx}
\\
&&\ov{\vary\conc\cB\conc\varz\seq\varz\conc\vary}\ar@{..}@/^5pt/[uul]&*++{\dots}
\\
&&*{{\txt{(a) Algorithm\;AlgWE}_{\textrm{Base}}}}
}}
$$
\vspace{20pt}
$$
{\begin{array}{cc}
\xymatrix@-3mm {
&\ov{\varx\conc\varx\conc\cA\conc\vary\conc\cB\conc\varz\seq\cA\conc\varx\conc\varx\conc\varz\conc\vary}
\ar@{=}[d]&
\\
&\ov{\begin{array}{c}\vary\conc\cB\conc\varz\seq\varz\conc\vary
\\ \varx\conc\varx\conc\cA\seq\cA\conc\varx\conc\varx\end{array}}\ar[d]|{\skipper\varz\rar\vary\conc\varz}
\ar@{->}`l[ul]+<20pt,-10pt>`[ur]+<5pt,-12.5pt>|{\upskipper\vy\rar\vz\conc\vy\upskipper}`[d]+<-10pt,-20pt>[] 
\\
&\ov{\begin{array}{c}\cB\conc\vary\conc\varz\seq\varz\conc\vary\\
\varx\conc\varx\conc\cA\seq\cA\conc\varx\conc\varx\end{array}}\ar@/^45pt/[]!<5pt,-40pt>;[u]!<-5pt,35pt>|<(.25){\skipper\varz\rar\cB\conc\varz}\\
&*{\txt{(b) Algorithm AlgWE}_{\textrm{Split}}}
}
&
\xymatrix@-3mm {
\\
&\ov{\varx\conc\varx\conc\cA\conc\vary\conc\cB\conc\varz\seq\cA\conc\varx\conc\varx\conc\varz\conc\vary}
\ar@{=}[d]&
\\
&\ov{\begin{array}{c}\bottom
\\ \varx\conc\varx\conc\cA\seq\cA\conc\varx\conc\varx\end{array}}
\\
&\save[]+<0pt,-18pt>*{\txt{(c) Algorithm AlgWE}_{\textrm{Count}}}\restore
}
\end{array}}
\vspace{-2ex}
$$
\caption{Solution graphs for $\varx\conc\varx\conc\cA\conc\vy\conc\cB\conc\vz\seq\cA\conc\varx\conc\varx\conc\vz\conc\vy.$}
\label{fig:solscheme}
\end{figure}

 \figurename~\ref{fig:solscheme} demonstrates the difference between the simplification algorithms used in the presented interpreters. The~dotted edges in the graph constructed using the algorithm $\SolBase$ show the equality of the node labels; but the paths are not folded here, because the equal nodes are not along the same path. The edges given in the double lines show the splits. The sign $\bot$ denotes the contradiction. Every solution of the equation corresponds to a non-empty set\footnote{The set may be infinite, for example the solution $\varx\seq\empt\logand \vary\seq\empt$ of the equation $\varx\conc\vary\seq\vary\conc\varx$ may be a result of the composition of the elementary substitutions $\varx\rar\vary\conc\varx,\,\varx\rar\empt,\,\vary\rar\empt$; $\varx\rar\vary\conc\varx,\,\varx\rar\vary\conc\varx,\,\varx\rar\empt,\,\vary\rar\empt$; \etc In such a case, the solutions are always resulting from concatenations not increasing these solutions' lengths.} of paths rooted in the initial node of its solution tree and ending at a leaf labelled by $\OutTrue$. Thus, if the solution graph of the equation does not have $\OutTrue$-leaves, then the equation solution set is empty. The graph constructed using the algorithm $\SolBase$ is infinite, whereas the other two graphs show that the equation $\varx\conc\varx\conc\cA\conc\vy\conc\cB\conc\vz\seq\cA\conc\varx\conc\varx\conc\vz\conc\vy$ has no solution.
\newpage

\section{Verification Task}\label{Sect:VerTask}

We use the notion of a parameter (\ie a dynamic variable) for a datum which is already given, but it is unknown to us; while a variable value is undefined and is to be assigned. The parameter values are used in this paper in order to represent possible paths in the solution tree. Thus, if the $\LogLang$-program $\Prog$ in $\WeqInt\hspace{-0.2ex}\left(\Prog,\langle\eq_1\dots,\eq_n\rangle\right)$ is replaced with a~parameter, then the stepwise unfolding of this call generates all the possible programs (\ie the $\LogLang$-narrowings' lists) that are compatible with the equation list. The unfolding stops when either the equation list is empty or no $\LogLang$-narrowing compatible with the current equation list is found. Henceforth the letters $\pu$ and $\pv$, maybe subscripted, stand for the parameters. 

Below we use the underlining sign to show encoded structures of the program to be specialized. Given an $\IntLang$-program transformation tool $\scp$, a list of the word equations $\Eqs$ and encoded sources $\Code{\WeqInt}$ of an interpreter $\WeqInt$, we consider the following specialization task.

$$\SpecTask(\WeqInt,\Eqs)\triangleq \scp\hspace{-0.2ex}\left(\Code{\WeqInt}, \Code{\FunGo(}\vpar, \Code{\wholeencode{\hspace{0.5ex}\Eqs\hspace{0.5ex}})}\right),$$

\noindent where $\FunGo$ is the name of the entry function of $\WeqInt$. Here $\vpar$ ranges over the set of the encoded $\LogLang$-programs (\figurename~\ref{fig:encoding}) that can be interpreted by $\WeqInt$, namely all possible encoded sequences of the $\LogLang$-narrowings. The lengths of the $\LogLang$-programs are unbounded, so there is no bound on the solution lengths.

\begin{figure}[h]\centering
\begin{tabular}{rcl}\small
\wholeencode{$\kavych$LHS = RHS$\kavych$} &\quad = \qquad & (\,\wholeencode{\hspace{0.6ex}LHS\hspace{0.6ex}}\,,\,\;\wholeencode{\hspace{0.6ex}RHS\hspace{0.6ex}}\,) \\
\wholeencode{$\kavych\varx\rar$Narr$\kavych$} &\quad = \qquad& \wholeencode{\hspace{0.6ex}$\varx$\hspace{0.6ex}}\;$\rar$\;\wholeencode{\hspace{0.6ex}Narr\hspace{0.6ex}} \\
\wholeencode{\hspace{0.6ex}Term $\longconc$ Expr\hspace{0.6ex}} &\quad = \qquad&\wholeencode{\hspace{0.6ex}Term\hspace{0.6ex}}\,$\longconc$\,\wholeencode{\hspace{0.6ex}Expr\hspace{0.6ex}}\\
\wholeencode{\hspace{0.7ex}Letter\hspace{0.7ex}}&\quad = \qquad& Letter \\
\wholeencode{\hspace{0.7ex}Variable\hspace{0.7ex}}&\quad = \qquad& ($\codex$\,Variable)
\end{tabular}
\caption{The encoding used for the data in the interpreters.}\label{fig:encoding}
\end{figure}

The result of this specialization is a program with the input value to be assigned to $\vpar$. We impose the following minimal requirement on the specialization, which is strengthened in \sectionname~\ref{Sect:VerRes}: the specialization succeeds if the resulting process tree generated by $\scp$ contains a leaf labelled with the value $\OutTrue$ iff the equation system $\Eqs$ has a solution. In that case the specialization tool $\scp$ verifies the existence of a sequence of substitutions that generates a solution of the system $\Eqs$ given to the interpreter $\WeqInt$. However, we do not require the specialization task to terminate on every equation list. Thus, the power of the suggested verification scheme depends on the underlying interpreter $\WeqInt$.

\section{Unfold/Fold Program Transformation Method}\label{Sect:SpecScheme}

The specialization tool $\scp$ used in the verification scheme above is based on the elementary unfold/fold technique widely used, \eg, in deforestation, supercompilation, partial evaluation, partial deduction, and so on~\cite{Burst,JonesBook,Pettorossi96,SecherSorensen99}. The algorithm transforms the $\IntLang$-programs (\sectionname~\ref{Subsect::IntLang}). The technique exploits the sub-algorithms presented briefly in this Section and more formally in the paper~\cite{NemReachability}. The unfold/fold algorithm assumes that every node $\Node$ in the process tree $\ProcTree$ of the $\IntLang$-program is labelled with a configuration, which represents the current parameterized computation state. 

\begin{Definition}\label{Definition::Config}
A configuration $C$ is a parameterized expression in the language $\IntLang$. Namely, it is either a parameter, a string in $\Sigma^*$, a parameterized expression enclosed in the parentheses, a concatenation of parameterized expressions or a function call with the parameterized expressions as its arguments.

\noindent The active call of the configuration $C$ is the function call (if any) with the leftmost closing right bracket.  
\end{Definition}

\noindent We say an~$\IntLang$-expression is \emph{ground} if it does not contain function calls, while it may contain parameter occurrences. The~set of the~parameters is denoted with $\ParSet$, and $\Ground$ stands for the set of the ground expressions. We say that a call $\Fun(\aarg_1,\dots,\aarg_n)$ matches against $\Fun(\Pat_1,\dots,\Pat_n)$, where $\aarg_i$ are object expressions, and $\Pat_i$ are patterns, if there exists a substitution $\xi$ \st $\forall i,1\leq i\leq n(\Pat_i\xi\seq\aarg_i)$.

\begin{Definition}\label{Definition::NarrowSet}
Given a program rule $R: \Fun(\Pat_1,\dots,\Pat_n)\seq \mathrm{S}$, let $C$ be of the form $\Fun(\aarg_1,\dots,\aarg_n)$, where $\aarg_i$ are ground\footnote{This property is guaranteed by the call-by-value semantics.}. We say that the substitution $\xi:\ParSet\rightarrow\Ground$ is \emph{a parameter narrowing} unifying $R$ with $C$ iff $\exists\sigma:\VarSet\hspace{-0.25ex}\left(\langle\Pat_1,\dots,\Pat_n\rangle\right)\rightarrow\Ground$ \st $\forall i\,(\Pat_i\sigma\seq \aarg_i\xi)$. We say that the set of the par-narrowings $\{\xi_i\}$ is \emph{exhaustive} \wrt $R$ if for every substitution $\tau:\ParSet\rightarrow\Sigma^*$ \st the expression $C\tau$ matches against the left-hand side of the rule $R$, there exists a par-narrowing $\xi_i$ \st $\tau$ is an instance of $\xi_i$.
\end{Definition}

Now we are ready to describe the unfold/fold algorithm. Every node in the process tree $\ProcTree$ is marked either as open (by default), or as closed with some node $\Node'$. The three steps listed below are applied to the tree $\ProcTree$ until all the~nodes in $\ProcTree$ are closed.

\begin{itemize}
\item \textbf{Unfolding step}. Given an open node $\Node$ labelled with a parameterized configuration $C$, consider the active~call $\Fun(\aarg_1,\dots,\aarg_n)$ in $C$. For every rule $R_i: \Fun(\Pat_{i,1},\dots,\Pat_{i,n})\seq \mathrm{S}_i$ in the definition of $\Fun$ (where $\Pat_{i,j}$ are patterns), construct a set of pairs $\langle \sigma_{i,k}, \xi_{i,k}\rangle$ \st $\forall j\,(\Pat_{i,j}\sigma_{i,j}\seq \aarg_j\xi_{i,j})$, and the set $\{\xi_{i,k}\}$ of the parameter narrowings is exhaustive \wrt $R_i$. For every such a par-narrowing $\xi_{i,k}$ generate an open child node $\Node_{i,k}$. Construct $C\xi_{i,k}$, and replace the active call $\Fun(\aarg_1,\dots,\aarg_n)\xi_{i,k}$ in it with $\mathrm{S}_\mathrm{i}\sigma_{i,k}$. The result is the configuration\footnote{In the case of the verification task considered, Property~\ref{Prop:unfold} implies that the par-narrowing $\xi_{i,k}$ is applied to the only active call, since the arguments of the other calls do not include parameters.} labelling the node $\Node_{i,k}$.
\item \textbf{Folding step}. Given a node $\Node$ labelled with a configuration $C$\hspace{-0.2ex}, if some its ancestor $\Node_0$ is labelled with $C$ (up to a \hspace{0.2ex}parameter renaming), then mark $\Node$ as closed with $\Node_0$ and remove all the paths originating from $\Node$.
\item \textbf{Close}. Mark an open node $\Node$ as closed with $\Node$ if either $\Node$ is labelled with a~ground expression, or all the successors of $\Node$ are closed.
\end{itemize}

\noindent In order to guarantee that the unification algorithm used in the unfolding step can always produce a finite set of the parameter narrowings, we use the~following syntactic property of the~function $\FunEq$ of the~interpreters considered. \figurename~\ref{fig:codemain} presents the source code of the function $\FunEq$ in the interpreter $\WeqIntBase$. The other interpreters use this function with some minor changes, such as applying the substitution function to the equation lists and storing information about the number of the equations in the list in the second argument of $\FunEq$. The patterns used for the first argument of $\FunEq$ in the left-hand sides of the definitions are the same in all the three interpreters considered.

\begin{Property}\label{Prop::SyntMain}
Given the interpreters $\WeqIntBase$, $\WeqIntSplit$, and $\WeqIntSym$, the program rewriting rules defining the~function $\FunEq$ in the interpreters are only of the following forms:
\begin{itemize}
\item $\FunEq(\Pat,S_1)\seq S_2$, where ${S}_2$ is an object expression (rules (1) and (5) in \figurename~\ref{fig:codemain});

\item $\FunEq(\Pat\longconc\varr,S_1)\seq \FunEq(\varr,S_2)$, where $\left(\VarSet\left({S_2}\right)\setminus\VarSet\left({S_1}\right)\right)\cap\VarSet(\Pat)\seq\varnothing$, the part $\Pat$ does not contain expression-type variables, and $\varr\notin\VarSet(\mathrm{S_2})$ (rules (2--4 a,b)). 
\end{itemize}
\end{Property} 

\noindent We recall that the verification task is $\scp\hspace{-0.5ex}\left(\Code{\WeqInt}, \Code{\FunGo(}\vpar, \Code{\wholeencode{\Eqs})}\right)$, and the rules of the function $\FunGo$ of all the~three interpreters are $\FunGo(\varr,\vareqlist)\seq \FunEq(\varr,\FunSim(\textrm{\it{Other\;args}}))$, where $\varr$ does not occur in the other arguments. This fact together with Property~\ref{Prop::SyntMain} imply the following feature.

\begin{Property}\label{Prop:unfold}
Let us consider the process tree $\ProcTree$ generated by the specialization task $\SpecTask(\WeqInt,\Eqs)$, where $\WeqInt\in\{\WeqIntBase,\WeqIntSplit,\WeqIntSym\}$.
\begin{itemize}
\item Given an arbitrary configuration $C$ labelling a node in $\ProcTree$\hspace{-0.5ex}, the only parameterized call in $C$ (if any) is of the form $\FunEq(\pv_i,\mathrm{Other\;args})$, where no parameter occurs\footnote{A rewriting rule $\FunEq(\Pat,{S}_1)\seq {S}_2$ may include letter-type pattern variables shared by $\Pat$ and ${S_1}$ that occur in ${S_2}$, but the~value matched against ${S_1}$ is always an object expression, hence these variables are assigned with letters in any pattern matching.} in the \emph{other arguments}.
\item The patterns to be unified with the parameterized data never have more than one occurrence of an~expression-type variable. 
\end{itemize}
\end{Property} 

\begin{figure}[t]
\small\centering
$$
\begin{array}{l}
/\hspace{-0.3ex}* \mathrm{(1)\quad\;\,\,There\;are\;no\;more\;}\LogLang\mathrm{-narrowings,\;and\;the\; equation\;is\;trivial.}\;*\hspace{-0.3ex}/\\
\FunEq(\empt,(\empt,\empt))\seq \OutTrue;\\
/\hspace{-0.3ex}* \mathrm{(2\,a,b)\,The\;}\LogLang\mathrm{-narrowing\;}\kavych\svar_x\rar\empt\kavych\mathrm{ \;is\;compatible\;with\;an\;equation\;whose\;side}\\
\qquad\qquad\mathrm{\;starts\;with\;a\;variable\;named\;}\svar_x.\;*\hspace{-0.3ex}/ \\
\FunEq(((\codex\svar_x)\rar\empt)\longconc\varr, (\vareq_{\rm{LHS}},(\codex\svar_x)\longconc\vareq_{\rm{RHS}})) \\\quad\seq\FunEq(\varr,\FunSim(((\codex\svar_x)\rar\empt),\FunSubst((\codex\svar_x)\rar\empt,\empt,\vareq_{\rm{LHS}}),\FunSubst((\codex\svar_x)\rar\empt,\empt,\vareq_{\rm{RHS}})));\\
\FunEq(((\codex\svar_x)\rar\empt)\longconc\varr, ((\codex\svar_x)\longconc\vareq_{\rm{LHS}},\vareq_{\rm{RHS}})) \\\quad\seq \FunEq(\varr,\FunSim(((\codex\svar_x)\rar\empt),\FunSubst((\codex\svar_x)\rar\empt,\empt,\vareq_{\rm{LHS}}),\FunSubst((\codex\svar_x)\rar\empt,\empt,\vareq_{\rm{RHS}})));\\
/\hspace{-0.3ex}* \mathrm{(3\,a,b)\,The\;}\LogLang\mathrm{-narrowing\;}\kavych\svar_x\rar\vars\svar_x\kavych\mathrm{ \;is\;compatible\;with\;an\;equation\;having\;one\;side\;starting\;with}\\
\qquad\mathrm{\qquad\;the\;equation\;variable\;named\;}\svar_x\mathrm{\;and\;the\;other\;side\;starting\;with\;a\;symbol}.\;*\hspace{-0.3ex}/ \\
\FunEq(((\codex\svar_x)\rar\vars\conc(\codex\svar_x))\longconc\varr, ((\codex\svar_x)\longconc\vareq_{\rm{LHS}},\vars\longconc\vareq_{\rm{RHS}})) \\\quad\seq \FunEq(\varr,\FunSim(((\codex\svar_x)\rar\vars\conc(\codex\svar_x)),\FunSubst((\codex\svar_x)\rar\vars\conc(\codex\svar_x),(\codex\svar_x),\vareq_{\rm{LHS}}),
\\\qquad\qquad\qquad\qquad\qquad\qquad\qquad\qquad\qquad\qquad\;\,\,\hspace{0.2ex}\FunSubst((\codex\svar_x)\rar\vars\conc(\codex\svar_x),\empt,\vareq_{\rm{RHS}})));\\
\FunEq(((\codex\svar_x)\rar\vars\conc(\codex\svar_x))\longconc\varr, (\vars\longconc\vareq_{\rm{LHS}},(\codex\svar_x)\longconc\vareq_{\rm{RHS}})) \\\quad\seq \FunEq(\varr,\FunSim(((\codex\svar_x)\rar\vars\longconc(\codex\svar_x)),\FunSubst((\codex\svar_x)\rar\vars\conc(\codex\svar_x),\empt,\vareq_{\rm{LHS}}),
\\\qquad\qquad\qquad\qquad\qquad\qquad\qquad\qquad\qquad\qquad\;\,\,\hspace{0.2ex}\FunSubst((\codex\svar_x)\rar\vars\conc(\codex\svar_x),(\codex\svar_x),\vareq_{\rm{RHS}})));\\
/\hspace{-0.3ex}* \mathrm{(4\,a,b)\,The\;}\LogLang\mathrm{-narrowing\;}\kavych\svar_x\rar\svar_y\svar_x\kavych\mathrm{ \;is\;compatible\;with\;an\;equation\;having\;sides\;starting\;with}\\\qquad\mathrm{\qquad\;different\;equation\;variables\;named\;}\svar_x\mathrm{\;and\;}\svar_y.\;*\hspace{-0.3ex}/ \\
\FunEq(((\codex\svar_x)\rar(\codex\svar_y)\conc(\codex\svar_x))\longconc\varr,((\codex\svar_y)\longconc\vareq_{\rm{LHS}},(\codex\svar_x)\longconc\vareq_{\rm{RHS}})) \\\quad\seq\FunEq(\varr,\FunSim(((\codex\svar_x)\rar(\codex\svar_y)\conc(\codex\svar_x)),\FunSubst((\codex\svar_x)\rar(\codex\svar_y)\conc(\codex\svar_x),\empt,\vareq_{\rm{LHS}}),
\\\qquad\qquad\qquad\qquad\qquad\qquad\qquad\qquad\qquad\qquad\;\,\,\hspace{1.8ex}\FunSubst((\codex\svar_x)\rar(\codex\svar_y)\conc(\codex\svar_x),(\codex\svar_x),\vareq_{\rm{RHS}})));\\
\FunEq(((\codex\svar_x)\rar(\codex\svar_y)\conc(\codex\svar_x))\longconc\varr,((\codex\svar_x)\longconc\vareq_{\rm{LHS}},(\codex\svar_y)\longconc\vareq_{\rm{RHS}})) \\\quad\seq\FunEq(\varr,\FunSim(((\codex\svar_x)\rar(\codex\svar_y)\conc(\codex\svar_x)),\FunSubst((\codex\svar_x)\rar(\codex\svar_y)\conc(\codex\svar_x),(\codex\svar_x),\vareq_{\rm{LHS}}),
\\\qquad\qquad\qquad\qquad\qquad\qquad\qquad\qquad\qquad\qquad\;\,\,\hspace{1.8ex}\FunSubst((\codex\svar_x)\rar(\codex\svar_y)\conc(\codex\svar_x),\empt,\vareq_{\rm{RHS}})));\\
/\hspace{-0.3ex}* \mathrm{(5)\quad\;\,\,Stop\;the\;computation\;in\;the\;default\;case.}\;*\hspace{-0.3ex}/\\
\FunEq(\varr,(\vareq_{\rm{LHS}},\vareq_{\rm{RHS}})) \seq \OutFalse; 
\end{array}
$$
\caption{The function $\FunEq$ accepts two arguments: the first is an encoded substitutions list, and the second is an encoded pair representing an equation. Here $\svar_x$ takes a name of an equation variable, if $x\in\VAlph$, while $\vars$ takes a character. The second argument of $\FunSubst$ function is an accumulator.}
\label{fig:codemain}
\end{figure}
\normalsize

\noindent Henceforth we say that a configuration $C$ is primary if its unfolding results in parameter narrowings, and we call a node primary if configuration labelling it is primary. Property~\ref{Prop:unfold} implies that the par-narrowings constructed by the unfolding step are always substituted only to the function call $\FunEq$ being the active call. Hence, for the verification task considered, $C$ is primary iff $C$ is of the form $\FunEq(\pv_i,\mathrm{Other\;args})$, where the other arguments do not contain function calls. 

Properties~\ref{Prop:unfold} and~\ref{Prop::SyntMain} together imply that in the case of the verification task considered the exhaustive narrowing set always consists of the only par-narrowing, thus, a unification with one rewriting rule results in a single par-narrowing. Depending on the~form of the rule defining the function $\FunEq$, the par-narrowing would be either $\pv_i\rar\empt$, $\pv_i\rar L'\longconc \pv'_i$, where $L'$ is an object expression, or trivial $\pv_i\rar\pv_i$, if parameter $\pv_i$ is unified with the only variable $\varr$, in rule (5). Hence, every non-trivial par-narrowing corresponds to an $\LogLang$-narrowing of the variables of the equation being transformed. Provided this feature\footnote{Here the unfolding step has another important property: for all function rules excluding the last one (rule (5)), the narrowings imposed on the parameter $\vpar_i$ are always disjoint, provided that the equations in the list given to the function $\FunEq$ are of the reduced form. The rule (5) accepts an arbitrary input, serving as the \texttt{otherwise} branch. The~right-hand side of this rule is an object expression, hence there is no need to propagate negative constraints imposed on the parameter value to the~successor configurations.}, the unification process is always finite~\cite{NemytykhUnfold}. Note that no trivial $\texttt{otherwise}$ branch corresponds to a branch in the equation solution tree.

We call a node transient~\cite{TurGener}, if the one-step unfolding of the configuration labelling it produces no narrowing on the parameters; in particular, if all the $\aarg_i$ in $\Fun(\aarg_1,\dots,\aarg_n)$ are object expressions. A transient node has the~only child in the process tree.

\section{Results of Specialization}\label{Sect:VerRes}

This section discusses some conditions under which the verification succeeds, and presents several sets of word equation systems, which have been solved by means of the specialization task $\SpecTask(\WeqInt,\Eqs)$, where $\WeqInt$ is either $\WeqIntBase$, $\WeqIntSplit$ or $\WeqIntSym$. 

Given an equation list $\Eqs$ and an interpreter $\WeqInt$, the final result of the stepwise unfolding of $\SpecTask(\WeqInt,\Eqs)$ can be considered as a possibly infinite process tree modelling the solution tree of $\Eqs$. The folding occurs if a node in the process tree has an ancestor labelled with the same configuration modulo the parameter renaming. If the unfolding can produce infinite paths in the tree then the specialization process does not terminate, unless two equal configurations exist along every infinite path. Thus, the specialization terminates iff the relation of the~textual equality is a~well-quasi order over the configuration sequence along every infinite path in the process tree.

\subsection{Optimality of Specialization}\label{Subsect::Optimality}

In this section, we show that given the structure of the interpreters considered, the residual graphs produced by the specialization may be reduced to the solution graphs of the equations considered. For this purpose we have to consider the unfolding and the folding operations, which generate both the process graph and the solution graph in very similar ways. In \sectionname~\ref{Sect:SpecScheme} we have shown that every node in the process tree, whose one-step unfolding results in the set $\{\xi_i\}$ of the disjoint par-narrowings marking the outgoing arcs, corresponds to a node in the solution tree, and there is a bijection from the arc set marked by $\xi_i$ to the arc set marked by the corresponding $\LogLang$-narrowings, thus it remains to show that the folding does work exactly on the same nodes where the par-narrowings and $\LogLang$-narrowings are generated.  

In general, a partial process tree of the specialization task may require to construct a folding arc connecting transient nodes. That would cause problems with the reasoning on the process graphs in the terms of the solution graphs, because the transient nodes do not correspond to any nodes in the solution graph. Informally, we can say that the specialization result is optimal if no folding arcs connect transient nodes. The structure of the interpreters $\WeqIntBase$, $\WeqIntSplit$, $\WeqIntSym$ guarantees that all the non-transient nodes are also primary (see Property~\ref{Prop:unfold}). Thus, we define now a notion of the optimal specialization for the verification task given in \sectionname~\ref{Sect:VerTask}.  

\begin{Definition}\label{Definition:optim}
Given a task $\SpecTask(\WeqInt,\Eqs)$, its specialization result is said to be optimal iff all the arcs folding computation paths in the process graph connect the primary nodes.
\end{Definition}

The interpreters considered above satisfy the property that every parameter narrowing occurring in a process tree either generates an $\LogLang$-narrowing or results in an $\OutFalse$-node never unfolded. If the optimality holds, then all the intermediate steps of the specialization of the interpreters, including specialization of substitution and simplification, correspond to the nodes in which the folding never occurs, hence every transient path segment in the process graph may be represented with a single arc. We can therefore reason on the process graphs using the solution graphs of the equations \wrt which the interpreters are specialized. Moreover, the~optimality guarantees that the residual programs generated by a specialization tool contain no part of the~interpreters' source code, except the encoding $\wholeencode{P}$ of the $\LogLang$-programs $P$. Thus, the introduced optimality can be considered as an analogue of the Jones-optimality~\cite{BenAbram,JonesBook2} for the given verification task. The Jones-optimality demands that the interpretation overheads should be completely removed from the residual programs. The notion of the optimality given in Definition~\ref{Definition:optim} implies also that all the interpretation overheads are removed from the specialization result, although some pieces of the encoded $\LogLang$-programs will be present in it, since the parameters are narrowed according to their values.

Since we consider the input sequence of $\LogLang$-narrowings given to an interpreter $\WeqInt$ as a straight-line program, we can also use the following reasoning~\cite{NemytykhCSR}. In the classical first Futamura projection~\cite{Futamura}, which corresponds to the specialization task $\scp\hspace{-0.5ex}\left(\Code{\WeqInt},\Code{\FunGo(\langle\wholeencode{\hspace{0.4ex}\sigma_i\hspace{0.4ex}}\rangle},\vpar\Code{)}\right)$, the input data (the equation list) is dynamic, while the program given to the interpreter is static. Here we parameterize the program. 

\begin{Lemma}\label{lemma:optimality}
For every word equation $\Phi\seq\Psi$, the result of specializing $\SpecTask (\WeqInt, \Phi\seq\Psi)$, where $\WeqInt$ is either $\WeqIntBase$, $\WeqIntSplit$ or $\WeqIntSym$, is optimal.
\end{Lemma}

\noindent The idea of the proof is as follows. If the nodes $\Node_1$ and $\Node_2$ labelled with the equal non-primary configurations exist along the same path, then the closest primary ancestor of $\Node_1$ and the closest primary ancestor of $\Node_2$ are also labelled with the equal configurations. The property holds because the structure of the interpreters implies the following two statements. First, given any primary configuration $\FunEq(\vpar_i,\Eqs)$ along the path segment $[\Node_1;\Node_2]$ and the first primary configuration $\FunEq(\vpar_j,\Eqs')$ labelling a successor node of $\Node_2$, the length of the equation lists $\Eqs$ and $\Eqs'$ are equal. Second, all the equation transformations done along the segment $[\Node_1;\Node_2]$ are injective. The detailed proof is given in Appendix, \sectionname~\ref{proof:optlemma}.

Provided that the optimality holds, given a class of equations $\mathsf{K}$, we say the verification by specialization of the interpreter $\WeqInt$ succeeds over $\mathsf{K}$ iff for every equation $\eq\in\mathsf{K}$ the solution graph constructed by the corresponding algorithm solving the equation is finite.  

Let us show how the process graph generated by the specialization of an interpreter \wrt an equation corresponds to a solution graph of the equation, using the following example. Here we consider an equation such that the choice of the simplification algorithm has no impact on its solution.

\begin{Example}\label{Example::mappingoptimal}

Let us consider the task $\SpecTask (\WeqInt, \vx\conc\vy\conc\cA\seq\cA\conc\vx\conc\vy)$. Using the operations given above, we construct its process graph. Some structures of this process graph are relevant only to the interpreter, namely to the structures of the parameter narrowings and the function calls. If we delete them, as well as the \emph{otherwise} branches, we will get a solution graph of the equation, as is shown in~\figurename~\ref{fig:mappinggraphs}. 

\begin{figure}[t]
\begin{tabular}{cc}\small
$$
\hspace*{-10pt}\xymatrix@+3mm {&&&&&\\
&& *+[F-:<44pt>]{\begin{array}{l}\FunEq(\pu_0,\\\;\;\wholeencode{\kavych\vx\conc\vy\conc\cA\seq\cA\conc\vx\conc\vy\kavych})\end{array}}\ar[]!<-10pt,18pt>;[dl]|<(.4){{\begin{array}{l}\skipper\pu_0\rar \\(\wholeencode{\smkavych\vx\rar\empt\smkavych})\longconc\pu_1\end{array}}}\ar[d]|<(.4){\begin{array}{l}\upskipper\pu_0\rar\\(\wholeencode{\smkavych\vx\rar\cA\conc\vx\smkavych})\longconc\pu_1\upskipper\end{array}}\ar[]!<-50pt,50pt>;[dr]|<(.5){\skipper\skipper\skipper\txt{otherwise}}&\\
&*++{\dots}\ar[d]|<(.3){\begin{array}{l}\txt{\it no parameter}\\ \txt{\it narrowings}\end{array}}&*++{\dots}\ar[d]|{\begin{array}{l}\txt{\it no parameter}\\ \txt{\it narrowings}\end{array}}&*+{\OutFalse}&\\
&*+[F-:<44pt>]{\begin{array}{l}\FunEq(\pu_1,\\\;\;\wholeencode{\kavych\vy\conc\cA\seq\cA\conc\vy\kavych})\end{array}}\ar[d]|<(.35){{\begin{array}{r}\upskipper\pu_1\rar\qquad\skipper\\(\wholeencode{\smkavych\vy\rar\cA\conc\vy\smkavych})\longconc\pu_2\upskipper\end{array}}}\ar[]!<-10pt,20pt>;[dl]!<-15pt,0pt>|<(.25){\begin{array}{l}\skipper\pu_1\rar \\(\wholeencode{\smkavych\vy\rar\empt\smkavych})\longconc\pu_2\end{array}}&
\save[]+<0pt,-15pt>*+[F-:<44pt>]{\begin{array}{l}\FunEq(\pu_1,\\\;\;\wholeencode{\kavych\vx\conc\vy\conc\cA\seq\cA\conc\vx\conc\vy\kavych})\end{array}}\ar@{-->}@/_3.8pc/[uu]!<-3pt,-3pt>
\restore\\
*++{\dots}&*++{\dots}&\save[]+<-20pt,-10pt>*+{\OutFalse}\ar@{<-}[ul]!<20pt,13pt>|<(.25){\skipper\skipper\skipper\txt{otherwise}}\restore\\
\save[]+<-10pt,-20pt>*+[F-:<44pt>]{\begin{array}{l}\FunEq\\\;(\pu_2,\empt)\end{array}}\ar@{<-}[u]|<(.5){\begin{array}{l}\txt{\it no parameter}\\ \txt{\it narrowings}\end{array}}\ar[]!<9pt,-29.5pt>;[ddr]|<(.5){\skipper\skipper\skipper\txt{otherwise}}\ar[dd]|<(.33){\begin{array}{l}\skipper\pu_2\rar \empt\end{array}}\restore&
\save[]+<0pt,-33pt>*+[F-:<44pt>]{\begin{array}{l}\FunEq(\pu_2,\\\;\;\wholeencode{\kavych\vy\conc\cA\seq\cA\conc\vy\kavych})\end{array}}\ar@{<-}[u]|<(.5){\begin{array}{l}\txt{\it no parameter}\\ \txt{\it narrowings}\end{array}}
\ar@{-->}@/_5pc/[uu]!<-5pt,-5pt>\restore
\\
&\\
*+{\OutTrue}&*+{\OutFalse}&
}
$$
&\small
$$
\hspace*{-200pt}\xymatrix@-2.5mm {\\\\\\\\\\\\\\
&& *+++[F-:<44pt>]{\begin{array}{l}\vx\conc\vy\conc\cA\seq\cA\conc\vx\conc\vy\end{array}}\ar[dl]|<(.25){\begin{array}{l}\vx\rar\empt\end{array}}\ar[d]|<(.55){\begin{array}{l}\vx\rar\cA\conc\vx\end{array}}&\\
&*+++[F-:<44pt>]{\begin{array}{l}\vy\conc\cA\seq\cA\conc\vy\end{array}}&
\save[]+<0pt,-15pt>*+++[F-:<44pt>]{\begin{array}{l}\vx\conc\vy\conc\cA\seq\cA\conc\vx\conc\vy\end{array}}\ar@{-->}@/_3pc/[u]
\restore
&\\
\save[]+<-10pt,-20pt>*++{\begin{array}{l}\OutTrue\end{array}}\ar@{<-}[ur]|<(.4){\begin{array}{l}\skipper\vy\rar\empt\end{array}}\restore&
\save[]+<0pt,-33pt>*+++[F-:<44pt>]{\begin{array}{l}\vy\conc\cA\seq\cA\conc\vy\end{array}}\ar@{<-}[u]|<(.25){\begin{array}{l}\upskipper\vy\rar\cA\conc\vy\upskipper\end{array}}\ar@{-->}@/_4pc/[u]\restore
}
$$
\end{tabular}
\caption{A process graph and the corresponding solution graph. The quoted data (encoded in the process graph) are preserved in the solution graph.}
\label{fig:mappinggraphs}
\end{figure}
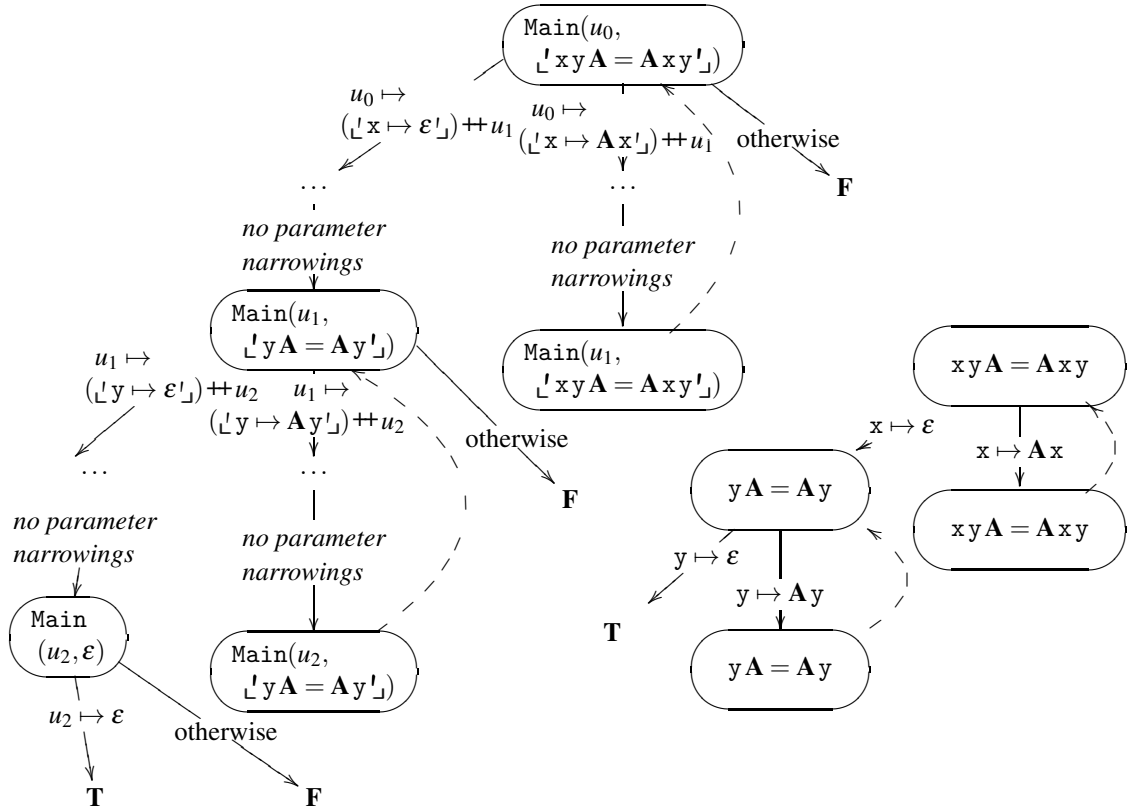
\end{Example}

\subsection{Specialization of Basic Interpreter}\label{subsection:weqintbase}

\begin{Definition}
A word equation $\Phi\seq\Psi$ is said to be \emph{quadratic} iff for all $\varx\in\VAlph$, $|\Phi|_{\varx}+|\Psi|_{\varx}\leq 2$.
\end{Definition}

\noindent For every quadratic equation, the solution graph constructed with the use of the algorithm $\SolBase$ is finite. This fact is well-known due to the works by Matiyasevich~\cite{DiekertJewels,Le2018,Lin2018,Matiyas}. Thus, specialization of $\WeqIntBase$ \wrt quadratic equations provides a basic test on the optimality of the program model. Namely, the optimality lemma implies the following proposition.

\begin{Prop}\label{Prop:quad}
For any quadratic equation $\Phi\seq\Psi$, specialization of $\SpecTask(\WeqIntBase,\Phi\seq\Psi)$ succeeds.
\end{Prop}

\subsection{Specialization of Splitting Interpreter}\label{subsect::regrep}

The interpreter $\WeqIntSplit$ was introduced as an optimized version of $\WeqIntBase$, but the experiments have shown that the specialization of $\WeqIntSplit$ succeeds in significantly more cases. One interesting class of the word equations solvable with the help of $\WeqIntSplit$ consists of a special kind of equations whose~solution sets are regular languages. 

\begin{Definition}
Given $\Phi\in\{\CAlph\cup\VAlph\}^*$, let $\xi(\Phi)$ map any $\cA\in\CAlph$ explicitly occurring in $\Phi$ to $\empt$, preserving the other part of $\Phi$. We say an equation $\Phi\seq\Psi$ is \emph{strictly regular-ordered with repetitions} iff $\xi(\Phi)$ is textually equal to $\xi(\Psi)$.
\end{Definition}

\noindent Thus, if the equation $\Phi\seq\Psi$ is strictly regular-ordered with repetitions, then $\forall \varx\in\VAlph\left(|\Phi|_{\varx}\seq |\Psi|_{\varx}\right)$ and the variable occurrences are ordered in $\Phi$ and $\Psi$ in the same way. The set of the strictly regular-ordered equations with repetitions generalizes the set of the regular ordered equations in which every variable occurs twice~\cite{Day2017RegOrd}. 

\begin{Example}
The solution sets of the three equations $\cA\conc\varx\seq\varx\conc\cA$, $\cA\conc\cA\conc\varx\seq\varx\conc\cA\conc\cA$, $\cA\conc\varx\conc\varx\seq\varx\conc\varx\conc\cA$ are equal, namely the sets are $\cA^*$. The first two equations are quadratic; the third is strictly regular-ordered with repetitions, but is not quadratic. Its solution graph constructed with the use of the algorithm $\SolBase$ is infinite. 
\end{Example}

The termination of the specialization is provided by the following lemma. 

\begin{Lemma}\label{lemma:regcommeq}
Given any strictly regular-ordered equation with repetitions $\Phi\seq\Psi$, every infinite path in its solution tree generated with the use of the algorithm $\SolSplit$ includes at least two nodes with equal labels.
\end{Lemma}

\noindent The idea of the proof is as follows. Every such an equation $\Phi\seq\Psi$ is split into the several quadratic equations after a number of substitutions applied to it. The detailed proof is given in Appendix (\sectionname~\ref{proof::regcommeq}).  

\begin{Coroll}
Given any strictly regular-ordered equation with repetitions $\Phi\seq\Psi$, specialization of the verification task $\SpecTask\hspace{-0.25ex}\left(\WeqIntSplit,\langle\Phi\seq\Psi\rangle\right)$ succeeds.
\end{Coroll}

\subsection{Specialization of Counting Interpreter}\label{Subsect::SpecCount}

The interpreter $\WeqIntSym$ uses more simplifying operations as compared to the interpreter $\WeqIntSplit$. The specialization of this interpreter succeeds additionally in solving one-variable equations. The success of the verification is guaranteed by the following lemma.

\begin{Lemma}\label{lemma:onevareq}
Given an equation $\Phi\seq\Psi$, where $\Phi,\Psi\in \{\CAlph\cup \{\varx\}\}^*$, and $|\Phi\conc\Psi|_{\varx}>2$, every infinite path of its solution graph constructed with the use of the algorithm $\SolCount$ contains a split. 
\end{Lemma}

\noindent The idea of the proof is similar to the one of Lemma~\ref{lemma:regcommeq}: if $\Phi\seq\Psi$ has a solution, then after an application of a number of the substitutions, the resulted equation will have var-permutated prefixes or suffixes. See the Appendix for the details (\sectionname~\ref{proof:onevareq}).

\begin{Coroll}
For any one-variable equation $\Phi\seq\Psi$, specialization of the task $\SpecTask(\WeqIntSym, \langle\Phi\seq\Psi\rangle)$ succeeds.
\end{Coroll}

\noindent In order to experimentally test the verification technique presented in this paper, we have generated a~benchmark consisting of 50 equation systems\footnote{The archive containing the equations is given on the~web-page~\url{https://github.com/TonitaN/TestEquations}.}: the tests 1--10 are the regular-ordered equations with repetitions; the tests 11--20 are similar to the regular-ordered equations with repetitions, but the variable occurrences order may be different on the equation sides where the variables occur; the tests 21--30 contain equations of the form $\varx\conc\Phi\seq \Psi\conc\varx$, where $\Phi\seq\Psi$ is a regular-ordered equation with repetitions and neither $\Phi$ nor $\Psi$ contain $\varx$; the tests 31--40 present systems of the regular-ordered equations with repetitions mixed with equations of the form $\Phi\conc\Psi\seq\Psi\conc\Phi$; the tests 41--50 are equations of no special form sharing several variables in right- and left-hand sides.  

The supercompiler SCP4~\cite{NemytykhBook} was mainly used as $\SpecTask$ in the tests. The experimental supercompiler~\mscp\;was also used and has shown the same solvability results\footnote{See the web-page \url{http://refal.botik.ru/mscp/weq\_int\_readme.html} for details.} on the tests above, but it spends much more time for producing the results as compared with \texttt{SCP4}. The comparative verification results between the approach presented in this paper and the external SMT-solvers \texttt{CVC4}, \texttt{Z3str3} are presented in \figurename~\ref{fig:bench}, the last row. The results show that the scheme $\SpecTask(\WeqIntSym,\Eqs)$ is quite stable modulo small changes in the equations which are guaranteed to be solved by it.

Finally, we have tested the scheme $\SpecTask(\WeqIntSym,\Eqs)$ on the equation set provided by the paper~\cite{Day2019Woorjpe} as a benchmark for the string constraint solver \texttt{Woorpje}, namely Track 1 consisting of 200  equations guaranteed to have a solution; and Track 5 consisting of 200 equation systems\footnote{The reader can find the residual programs encoding paths in the solution graphs for the equations generated in the experiments at the web-page~\url{https://github.com/TonitaN/TestEquations}.}. We have removed the~length constraints from the Track 5 benchmark before the specialization starts. The results are quite successful, provided that we use the general-purpose specialization tool for the verification. First of all, the residual programs constructed by $\SpecTask(\WeqIntSym, \Eqs)$ never contain functions returning $\OutTrue$ if the system $\Eqs$ has been found unsatisfiable by the other solvers. Moreover, if the system $\Eqs$ has solutions, then $\SpecTask(\WeqIntSym, \Eqs)$ always generates programs containing functions with the output $\OutTrue$, if terminates.  That is a practical evidence that the specialization produces sound and complete solution graphs, which is a corollary of the optimality lemma. Second, the equations are successfully solved in 179 out of 200 cases in Track 1 and in 181 out of 200 cases in Track 5. This result is comparable with the verification results done by \texttt{Z3str3}~\cite{Z3}; for 17 equations in Track 1 the specialization process does not terminate. In the remaining cases, the specialization process is theoretically terminating but takes too much time. The equations for which the specialization is the most time-consuming all are linear, \ie every variable occurs in at most once per such an equation. 

The average runtime of $\SpecTask(\WeqIntSym, \Eqs)$ on the tests considered is about 3.5 minutes per task. While the runtime of all the tests solved successfully by \texttt{CVC4} or \texttt{Z3str3} is less than 2 minutes. Although we have used the 3-hour timeout, the long-running tests resulting in the verification success occurred only for the scheme $\SpecTask(\WeqIntSym, \Eqs)$. There are two main reasons of such a difference. First, the scheme uses the general-purpose specialization tool, employing time-consumable transformation operations, such as the residualization. Second, the scheme considered searches for all the solutions of the equations rather than for at least one. That is the main reason for the long runtimes, and is immanent to the problem solved. Many equations have the solution set exponentially-sized \wrt their length.  
   
\begin{figure}[h]
\center
\begin{tabular}{|l|c|c|c|c|}\hline
\multirow{2}{*}{\qquad Benchmark}& \multirow{2}{*}{Tests (total)}&\multicolumn{3}{c|}{Timeout / not terminated}\\\cline{3-5}
 &  &$\upskipper^{\upskipper}$\texttt{CVC4}\,\, &$\upskipper^{\upskipper}$\texttt{Z3str3}\,\, &$\upskipper^{\upskipper}\WeqIntSym$\,\, \\\hline
Track 1 (Woorpje)$\upskipper^{\upskipper}$ & 200 & 8 & 13 & 21 \\
Track 5 (Woorpje)$\upskipper^{\upskipper}$ & 200 & 4 & 14 & 19 \\
Our benchmark$\upskipper^{\upskipper}$ & 50 & 21 & 28 & 10 \\\hline 
\end{tabular}
\caption{The verification results.}
\label{fig:bench}
\end{figure}

\section{Discussion}\label{Sect:Discuss}
 
The discussed specialization tasks above have been solved by using the following two program specializers developed for the string manipulating functional language Refal~\cite{TurRefal5}, namely the model supercompiler\footnote{The supercompiler is presented on the web-page \url{http://refal.botik.ru/mscp/mscp-a\_eng.html}} \mscp~and the experimental supercompiler \texttt{SCP4}~\cite{NemytykhBook}. Albeit we used the supercompilers, the properties of the residual programs, which we are interested in, do not depend on specific features of the supercompilation method~\cite{Tur86} and can be reproduced with other specializers based on partial evaluation, partial deduction, and so on \cite{JonesBook,Pettorossi96}. We used the language $\IntLang$ with built-in concatenation constructor, however the method can be used also over the lisp-data with some minor changes in the interpreters' source code. First, all the parameter narrowings are constructed only by specializing the head of the list of the $\LogLang$-narrowings. Second, the optimality lemma guarantees that the additional loops in the intermediate steps of the interpretation would not cause any folding operations, since the data given to the intermediate functions is not parameterized. Thus, the structure of the residual programs is preserved.

Our approach is able to solve the regular-ordered equations with repetitions~(see \ref{subsect::regrep}), which are hard to solve for the known existing solvers, especially in the case when the solution set is empty. For example, neither \texttt{Z3str3} nor \texttt{CVC4} terminates on the equation $\cA\conc\cB\conc\vx\conc\vx\conc\vy\conc\vy\seq\vx\conc\vx\conc\vy\conc\vy\conc\cB\conc\cA$, which is proved to have no solution by the $\WeqIntSplit$ or $\WeqIntSym$ specialization. The feature to solve equations with the empty solution set is especially interesting provided the fact that this case is the hardest for the most theoretical algorithms, and is also the bottleneck for the SMT-solvers used in our tests. 

The domain of the described verification method is not exhausted by the sets of the equations considered above. One more interesting class of the equations with the variables shared by left- and right-hand sides consists of the equations of the form $\Phi\conc\vary\seq\vary\conc\Psi$, where $\Phi\seq\varx\conc\Phi_1\conc\varx\dots\varx\conc\Phi_n$, $\Psi\seq\Psi_1\conc\varx\dots\varx\conc\Psi_n\conc\varx$, where $\Phi_i,\Psi_i\in\CAlph^*$ and $\exists k\in\mathbb{N}\,\forall i,j\,|\Phi_i|\seq k\logand |\Psi_i|\seq k$. Examining their solution graphs, we can prove that the specialization task $\SpecTask(\WeqIntSym,\Eqs)$ successfully solves such equations. The experiments with the randomly chosen equations mentioned above promise to find other interesting classes of the word equations that can be solved by automated specialization tools.

\subsection{Related Works}\label{Subsect::relwork}

A number of efficient string constraint solving tools were designed, which look for the word equation solutions bounded by a given length, \eg \cite{Bjorner,Day2019Woorjpe}\footnote{Actually, if the upper bound is assigned dynamically, such a tool can decide solvability of every word equation, because a minimal solution length is at most doubly exponential in the equation length~\cite{Jez}.}. Reasonings on the unbounded case can provide efficient methods for solution search if the equations considered satisfy some special properties. For example, a number of efficient solving algorithms have been designed for the set of straight-line word equations, \eg\cite{Abdulla,Ostrich}, whereas our specialization \wrt such equations is too time-consuming. The difference is again rooted in the tasks considered: our approach looks for a description of the whole set of equation solutions, while SMT-solvers aim to find at least one solution. 

Several recent works exploit the unfold/fold technique with Nielsen's transformation for solving quadratic word equations, in the way originated by Matiyasevich in 1968~\cite{Matiyas}. In the paper~\cite{Lin2018}, the algorithm using non-deterministic counter systems for searching solutions of the unbounded-length word equations with regular constraints via Nielsen's transformation is introduced, and the completeness of the given algorithm has been shown for the set of the regular ordered equations. Maybe the regular ordered equations with repetitions, being split in the way shown in $\SolSplit$ algorithm~(\sectionname~\ref{Sect:WeqInt}), can be solved by this method as well. In the paper~\cite{Le2018}, Nielsen's transformation is used for solving unbounded-length quadratic word equations. As in the original Matiyasevich work, the algorithm does not terminate if the input equation contains more than two occurrences of some variable. Thus, these algorithms are not able to solve non-quadratic regular-ordered equations with repetitions, which are solvable by our specialization scheme. Advanced SMT-solvers such as \texttt{CVC4}~\cite{CVC4}, \texttt{Z3str3}~\cite{Z3}, or \texttt{S3P} presented in~\cite{Trinh} demonstrate a very good efficiency in many practical cases, however the paper~\cite{Le2018} shows that their algorithms are not complete even \wrt the~set of the quadratic equations, \eg do not solve on the equation $\varx\conc\vvarv\conc\vary\seq\vary\conc\varw\conc\varx$, whose solution set is quite complex; see Hmelevskij's work~\cite{Hmelevsky} for the very first proof of this fact. The tests of our benchmark have shown that the current version of \texttt{CVC4} solves equations of the given form, but fails to solve more complex quadratic equations, \eg $\varx_1\conc\varx_2\conc\varx_3\conc\cA\cB\cA\cB\cA\cB\seq \cA\cA\cA\cB\cB\cB\conc\varx_2\conc\varx_3\conc\varx_1$ (which is solved by specialization even of the basic interpreter $\WeqIntBase$). Based on the results of the verification presented in this paper we may conclude that the most troublesome cases for the SMT-solvers are the ones when the equation system has no solution, and this fact cannot be shown by reasoning on solution lengths. In that cases, the verification scheme $\SpecTask(\WeqIntSym, \Eqs)$ has the best success rate as compared to \texttt{CVC4} and \texttt{Z3str3}. 

\section{Conclusion}
\label{Sect::Conclusion}

We have shown that general-purpose specializers can be useful for solving some classes of the word equations. Instead of modifying the specialization tools, we modify the word equation interpreters specialized according with the verification scheme. This technique uses a modification of the classical first Futamura projection~\cite{Futamura} and simplifies the work of interest. Starting from the simplest interpreter $\WeqIntBase$, every new refinement extends significantly the set of the equations solvable via the specialization method. The specialization-time overheads are high as compared with the direct work of the existing string constraint solvers, but the specialization method presented in this paper aims at supporting development of the solver prototypes with a minimal effort. Experiments with the prototypes provide a fruitful research material on the sets of the word equations over which the verification algorithm terminates. Moreover, we have shown that theoretical methods for solving the word equations can be useful in the automatic approach, hence these methods are able to prove unsatisfiability of word equation systems, which, as our experiments show, is hard for some well-known state-of-art SMT-solvers.

Another interesting aspect of the presented verification experiments is the optimality. The non-deterministic algorithms for solving the word equations are well-designed in order to be used in the~intermediate interpretation. This paper considers the optimality property in the case of the basic folding, however our experiments show that the constructed interpreters provide possibility for the optimal verification, if one uses a more complex path termination criterion based on the homeomorphic embedding relation~\cite{Sorensen95}. Thereby advanced specialization tools can also be used for solving the word equations, and the additional strategies developed for program transformation may support more efficient algorithms as compared with the basic unfold/fold method.

\section*{Acknowledgements}
We would like to thank A.~P.~Nemytykh, who contributed greatly to the improvement of the paper, and the anonymous referees for the thoughtful comments which helped a much to clarify the presentation. 

\label{sect:bib}

\nocite{*}
\bibliographystyle{eptcs}
\bibliography{weqs_nepan2020}

\begin{thebibliography}{10}
\providecommand{\bibitemdeclare}[2]{}
\providecommand{\surnamestart}{}
\providecommand{\surnameend}{}
\providecommand{\urlprefix}{Available at }
\providecommand{\url}[1]{\texttt{#1}}
\providecommand{\href}[2]{\texttt{#2}}
\providecommand{\urlalt}[2]{\href{#1}{#2}}
\providecommand{\doi}[1]{doi:\urlalt{http://dx.doi.org/#1}{#1}}
\providecommand{\bibinfo}[2]{#2}

\bibitemdeclare{article}{Abdulla}
\bibitem{Abdulla}
\bibinfo{author}{Parosh~Aziz \surnamestart Abdulla\surnameend},
  \bibinfo{author}{Mohamed~Faouzi \surnamestart Atig\surnameend},
  \bibinfo{author}{Yu-Fang \surnamestart Chen\surnameend},
  \bibinfo{author}{Bui~Phi \surnamestart Diep\surnameend},
  \bibinfo{author}{Luk\'{a}\v{s} \surnamestart Hol\'{\i}k\surnameend},
  \bibinfo{author}{Ahmed \surnamestart Rezine\surnameend} \&
  \bibinfo{author}{Philipp \surnamestart R\"{u}mmer\surnameend}
  (\bibinfo{year}{2017}): \emph{\bibinfo{title}{Flatten and Conquer: {A}
  Framework for Efficient Analysis of String Constraints}}.
\newblock {\sl \bibinfo{journal}{SIGPLAN Not.}}
  \bibinfo{volume}{52}(\bibinfo{number}{6}), pp. \bibinfo{pages}{602--617},
  \doi{10.1145/3140587.3062384}.

\bibitemdeclare{article}{Leuschel2008}
\bibitem{Leuschel2008}
\bibinfo{author}{S.~\surnamestart Barker\surnameend},
  \bibinfo{author}{M.~\surnamestart Leuschel\surnameend} \&
  \bibinfo{author}{M.~\surnamestart Varea\surnameend} (\bibinfo{year}{2008}):
  \emph{\bibinfo{title}{Efficient and Flexible Access Control via
  {J}ones-optimal Logic Program Specialisation}}.
\newblock {\sl \bibinfo{journal}{High. Order Symb. Comput.}}
  \bibinfo{volume}{21}(\bibinfo{number}{1--2}), pp. \bibinfo{pages}{5--35},
  \doi{10.1007/s10990-008-9030-8}.

\bibitemdeclare{article}{BenAbram}
\bibitem{BenAbram}
\bibinfo{author}{A.~\surnamestart Ben-Amram\surnameend} \&
  \bibinfo{author}{N.~\surnamestart Jones\surnameend} (\bibinfo{year}{2000}):
  \emph{\bibinfo{title}{Computational Complexity via Programming Languages:
  Constant Factors Do Matter}}.
\newblock {\sl \bibinfo{journal}{Acta Informatica}}
  \bibinfo{volume}{2}(\bibinfo{number}{37}), pp. \bibinfo{pages}{83--120},
  \doi{10.1007/s002360000038}.

\bibitemdeclare{inproceedings}{Bjorner}
\bibitem{Bjorner}
\bibinfo{author}{Nikolaj \surnamestart Bj{\o}rner\surnameend},
  \bibinfo{author}{Nikolai \surnamestart Tillmann\surnameend} \&
  \bibinfo{author}{Andrei \surnamestart Voronkov\surnameend}
  (\bibinfo{year}{2009}): \emph{\bibinfo{title}{Path Feasibility Analysis for
  String-Manipulating Programs}}.
\newblock In \bibinfo{editor}{Stefan \surnamestart Kowalewski\surnameend} \&
  \bibinfo{editor}{Anna \surnamestart Philippou\surnameend}, editors: {\sl
  \bibinfo{booktitle}{Tools and Algorithms for the Construction and Analysis of
  Systems}}, \bibinfo{publisher}{Springer Berlin Heidelberg},
  \bibinfo{address}{Berlin, Heidelberg}, pp. \bibinfo{pages}{307--321},
  \doi{10.1007/978-3-642-00768-2\_27}.

\bibitemdeclare{article}{Burst}
\bibitem{Burst}
\bibinfo{author}{R.~M. \surnamestart Burstall\surnameend} \&
  \bibinfo{author}{John \surnamestart Darlington\surnameend}
  (\bibinfo{year}{1977}): \emph{\bibinfo{title}{A Transformation System for
  Developing Recursive Programs}}.
\newblock {\sl \bibinfo{journal}{J. ACM}}
  \bibinfo{volume}{24}(\bibinfo{number}{1}), pp. \bibinfo{pages}{44--67},
  \doi{10.1145/321992.321996}.

\bibitemdeclare{article}{Ostrich}
\bibitem{Ostrich}
\bibinfo{author}{Taolue \surnamestart Chen\surnameend},
  \bibinfo{author}{Matthew \surnamestart Hague\surnameend},
  \bibinfo{author}{Anthony~W. \surnamestart Lin\surnameend},
  \bibinfo{author}{Philipp \surnamestart R\"{u}mmer\surnameend} \&
  \bibinfo{author}{Zhilin \surnamestart Wu\surnameend} (\bibinfo{year}{2019}):
  \emph{\bibinfo{title}{Decision Procedures for Path Feasibility of
  String-Manipulating Programs with Complex Operations}}.
\newblock {\sl \bibinfo{journal}{POPL}} \bibinfo{volume}{3}, pp.
  \bibinfo{pages}{1--30}, \doi{10.1145/3290362}.

\bibitemdeclare{article}{KarhumCombWo}
\bibitem{KarhumCombWo}
\bibinfo{author}{Christian \surnamestart Choffrut\surnameend} \&
  \bibinfo{author}{Juhani \surnamestart Karhum\"{a}ki\surnameend}
  (\bibinfo{year}{1997}): \emph{\bibinfo{title}{Combinatorics of Words}}.
\newblock {\sl \bibinfo{journal}{Handbook of Formal Languages}}, pp.
  \bibinfo{pages}{329--438}, \doi{10.1007/978-3-642-59136-5\_6}.

\bibitemdeclare{inproceedings}{Day2019Woorjpe}
\bibitem{Day2019Woorjpe}
\bibinfo{author}{J.~D. \surnamestart Day\surnameend}, \bibinfo{author}{Thorsten
  \surnamestart Ehlers\surnameend}, \bibinfo{author}{Mitja \surnamestart
  Kulczynski\surnameend}, \bibinfo{author}{Florin \surnamestart
  Manea\surnameend}, \bibinfo{author}{Dirk \surnamestart Nowotka\surnameend} \&
  \bibinfo{author}{Danny~Bogsted \surnamestart Poulsen\surnameend}
  (\bibinfo{year}{2019}): \emph{\bibinfo{title}{On Solving Word Equations Using
  {S}{A}{T}}}.
\newblock In: {\sl \bibinfo{booktitle}{Reachability Problems. RP 2019}},
  \bibinfo{volume}{11674}, \bibinfo{publisher}{Lecture Notes in Computer
  Science}, pp. \bibinfo{pages}{93--106}, \doi{10.1007/978-3-030-30806-3\_8}.

\bibitemdeclare{inproceedings}{Day2017RegOrd}
\bibitem{Day2017RegOrd}
\bibinfo{author}{Joel~D. \surnamestart Day\surnameend}, \bibinfo{author}{Florin
  \surnamestart Manea\surnameend} \& \bibinfo{author}{Dirk \surnamestart
  Nowotka\surnameend} (\bibinfo{year}{2017}): \emph{\bibinfo{title}{The
  Hardness of Solving Simple Word Equations}}.
\newblock In: {\sl \bibinfo{booktitle}{42nd International Symposium on
  Mathematical Foundations of Computer Science (MFCS 2017)}}, {\sl
  \bibinfo{series}{Leibniz International Proceedings in Informatics
  (LIPIcs)}}~\bibinfo{volume}{83}, pp. \bibinfo{pages}{18:1--18:14},
  \doi{10.4230/LIPIcs.MFCS.2017.18}.

\bibitemdeclare{article}{Pettorossi2019}
\bibitem{Pettorossi2019}
\bibinfo{author}{Emanuele \surnamestart {De Angelis}\surnameend},
  \bibinfo{author}{Fabio \surnamestart Fioravanti\surnameend},
  \bibinfo{author}{Alberto \surnamestart Pettorossi\surnameend} \&
  \bibinfo{author}{Maurizio \surnamestart Proietti\surnameend}
  (\bibinfo{year}{2018}): \emph{\bibinfo{title}{Solving {H}orn Clauses on
  Inductive Data Types Without Induction}}.
\newblock {\sl \bibinfo{journal}{Theory Pract. Log. Program.}}
  \bibinfo{volume}{18}(\bibinfo{number}{3--4}), pp. \bibinfo{pages}{452--469},
  \doi{10.1017/S1471068418000157}.

\bibitemdeclare{article}{Gallagher2019}
\bibitem{Gallagher2019}
\bibinfo{author}{Jes{\'{u}}s~J. \surnamestart Dom{\'{e}}nech\surnameend},
  \bibinfo{author}{John~P. \surnamestart Gallagher\surnameend} \&
  \bibinfo{author}{Samir \surnamestart Genaim\surnameend}
  (\bibinfo{year}{2019}): \emph{\bibinfo{title}{Control-Flow Refinement by
  Partial Evaluation, and its Application to Termination and Cost Analysis}}.
\newblock {\sl \bibinfo{journal}{Theory Pract. Log. Program.}}
  \bibinfo{volume}{19}(\bibinfo{number}{5--6}), pp. \bibinfo{pages}{990--1005},
  \doi{10.1017/S1471068419000310}.

\bibitemdeclare{article}{Futamura}
\bibitem{Futamura}
\bibinfo{author}{Yoshihiko \surnamestart Futamura\surnameend}
  (\bibinfo{year}{1999}): \emph{\bibinfo{title}{Partial Evaluation of
  Computation Process --- An Approach to a Compiler-Compiler}}.
\newblock {\sl \bibinfo{journal}{Higher-Order and Symbolic Computation}}
  \bibinfo{volume}{12}, pp. \bibinfo{pages}{381--391},
  \doi{10.1023/A:1010095604496}.

\bibitemdeclare{inproceedings}{Hamilton15}
\bibitem{Hamilton15}
\bibinfo{author}{Geoff~W. \surnamestart Hamilton\surnameend}
  (\bibinfo{year}{2015}): \emph{\bibinfo{title}{Verifying Temporal Properties
  of Reactive Systems by Transformation}}.
\newblock In: {\sl \bibinfo{booktitle}{Proceedings of the Third International
  Workshop on Verification and Program Transformation, VPT@ETAPS 2015, London,
  United Kingdom, 11th April 2015.}}, pp. \bibinfo{pages}{33--49},
  \doi{10.4204/EPTCS.199.3}.

\bibitemdeclare{article}{Hmelevsky}
\bibitem{Hmelevsky}
\bibinfo{author}{Ju.~I. \surnamestart Hmelevskij\surnameend}
  (\bibinfo{year}{1971}): \emph{\bibinfo{title}{Equations in a Free Semigroup.
  (in {R}ussian)}}.
\newblock {\sl \bibinfo{journal}{Trudy Mat. Inst. Steklov}}
  \bibinfo{volume}{107}, p. \bibinfo{pages}{286}.

\bibitemdeclare{article}{Jez}
\bibitem{Jez}
\bibinfo{author}{Artur \surnamestart Jez\surnameend} (\bibinfo{year}{2016}):
  \emph{\bibinfo{title}{Recompression: {A} Simple and Powerful Technique for
  Word Equations}}.
\newblock {\sl \bibinfo{journal}{J. ACM}}
  \bibinfo{volume}{63}(\bibinfo{number}{1}), \doi{10.1145/2743014}.

\bibitemdeclare{book}{JonesBook2}
\bibitem{JonesBook2}
\bibinfo{author}{Neil \surnamestart Jones\surnameend} (\bibinfo{year}{2002}):
  \emph{\bibinfo{title}{Computability and Complexity from a Programming
  Perspective}}.
\newblock \bibinfo{volume}{62}, \bibinfo{publisher}{NATO Science Series,
  Springer}, \doi{10.1007/978-94-010-0413-8\_4}.

\bibitemdeclare{book}{JonesBook}
\bibitem{JonesBook}
\bibinfo{author}{Neil \surnamestart Jones\surnameend}, \bibinfo{author}{Carsten
  \surnamestart Gomard\surnameend} \& \bibinfo{author}{Peter \surnamestart
  Sestoft\surnameend} (\bibinfo{year}{1993}): \emph{\bibinfo{title}{Partial
  Evaluation and Automatic Program Generation}}.
\newblock \bibinfo{publisher}{Prentice Hall International}.

\bibitemdeclare{book}{DiekertJewels}
\bibitem{DiekertJewels}
\bibinfo{author}{Juhani \surnamestart Karhum\"{a}ki\surnameend},
  \bibinfo{author}{Hermann \surnamestart Maurer\surnameend},
  \bibinfo{author}{Gheorghe \surnamestart Paun\surnameend} \&
  \bibinfo{author}{Grzegorz \surnamestart Rozenberg\surnameend}
  (\bibinfo{year}{1999}): \emph{\bibinfo{title}{Jewels are Forever,
  {C}ontributions on Th. Computer Science in Honor of {A}rto {S}alomaa}}.
\newblock \bibinfo{publisher}{Springer, Berlin, Heidelberg},
  \doi{10.1007/978-3-642-60207-8\_28}.

\bibitemdeclare{inbook}{Le2018}
\bibitem{Le2018}
\bibinfo{author}{Quang~Loc \surnamestart Le\surnameend} \&
  \bibinfo{author}{Mengda \surnamestart He\surnameend} (\bibinfo{year}{2018}):
  \emph{\bibinfo{title}{A Decision Procedure for String Logic with Quadratic
  Equations, Regular Expressions and Length Constraints}}, pp.
  \bibinfo{pages}{350--372}.
\newblock \bibinfo{volume}{11275}, \bibinfo{publisher}{Lecture Notes in
  Computer Science}, \doi{10.1007/978-3-030-02768-1\_19}.

\bibitemdeclare{article}{CVC4}
\bibitem{CVC4}
\bibinfo{author}{Tianyi \surnamestart Liang\surnameend},
  \bibinfo{author}{Andrew \surnamestart Reynolds\surnameend},
  \bibinfo{author}{Nestan \surnamestart Tsiskaridze\surnameend},
  \bibinfo{author}{Cesare \surnamestart Tinelli\surnameend},
  \bibinfo{author}{Clark \surnamestart Barrett\surnameend} \&
  \bibinfo{author}{Morgan \surnamestart Deters\surnameend}
  (\bibinfo{year}{2016}): \emph{\bibinfo{title}{An Efficient {S}{M}{T} Solver
  for String Constraints}}.
\newblock {\sl \bibinfo{journal}{Form. Methods Syst. Des.}}
  \bibinfo{volume}{48}(\bibinfo{number}{3}), pp. \bibinfo{pages}{206--234},
  \doi{10.1007/s10703-016-0247-6}.

\bibitemdeclare{inproceedings}{Lin2018}
\bibitem{Lin2018}
\bibinfo{author}{Anthony~Widjaja \surnamestart Lin\surnameend} \&
  \bibinfo{author}{Rupak \surnamestart Majumdar\surnameend}
  (\bibinfo{year}{2018}): \emph{\bibinfo{title}{Quadratic Word Equations with
  Length Constraints, Counter Systems, and {P}resburger Arithmetic with
  Divisibility}}.
\newblock In: {\sl \bibinfo{booktitle}{Automated Technology for Verification
  and Analysis. ATVA 2018}}, \bibinfo{volume}{11138},
  \bibinfo{publisher}{Lecture Notes in Computer Science}, pp.
  \bibinfo{pages}{352--369}, \doi{10.1007/978-3-030-01090-4\_21}.

\bibitemdeclare{inproceedings}{NemytykhCSR}
\bibitem{NemytykhCSR}
\bibinfo{author}{Alexei \surnamestart Lisitsa\surnameend} \&
  \bibinfo{author}{Andrei~P. \surnamestart Nemytykh\surnameend}
  (\bibinfo{year}{2007}): \emph{\bibinfo{title}{A Note on Specialization of
  Interpreters}}.
\newblock In \bibinfo{editor}{Volker \surnamestart Diekert\surnameend},
  \bibinfo{editor}{Mikhail~V. \surnamestart Volkov\surnameend} \&
  \bibinfo{editor}{Andrei \surnamestart Voronkov\surnameend}, editors: {\sl
  \bibinfo{booktitle}{Computer Science -- Theory and Applications}},
  \bibinfo{publisher}{Springer Berlin Heidelberg}, pp.
  \bibinfo{pages}{237--248}, \doi{10.1007/978-3-540-74510-5\_25}.

\bibitemdeclare{article}{NemReachability}
\bibitem{NemReachability}
\bibinfo{author}{Alexei \surnamestart Lisitsa\surnameend} \&
  \bibinfo{author}{Andrei~P. \surnamestart Nemytykh\surnameend}
  (\bibinfo{year}{2008}): \emph{\bibinfo{title}{Reachability Analysis in
  Verification via Supercompilation}}.
\newblock {\sl \bibinfo{journal}{Int. J. Foundations of Computer Science}}
  \bibinfo{volume}{19}(\bibinfo{number}{4}), pp. \bibinfo{pages}{953--970},
  \doi{10.1142/S0129054108006066}.

\bibitemdeclare{inproceedings}{Z3}
\bibitem{Z3}
\bibinfo{author}{V.~Ganesh \surnamestart M.~Berzish\surnameend} \&
  \bibinfo{author}{Y.~\surnamestart Zheng\surnameend} (\bibinfo{year}{2017}):
  \emph{\bibinfo{title}{Z3str3: {A} String Solver with Theory-aware
  Heuristics}}.
\newblock In: {\sl \bibinfo{booktitle}{Formal Methods in Computer Aided Design
  (FMCAD)}}, pp. \bibinfo{pages}{55--59}, \doi{10.23919/FMCAD.2017.8102241}.

\bibitemdeclare{article}{Makanin77}
\bibitem{Makanin77}
\bibinfo{author}{Gennadiy~S. \surnamestart Makanin\surnameend}
  (\bibinfo{year}{1977}): \emph{\bibinfo{title}{The Problem of Solvability of
  Equations in a~Free Semigroup}}.
\newblock {\sl \bibinfo{journal}{Math. USSR-Sb.}}
  \bibinfo{volume}{32}(\bibinfo{number}{2}), pp. \bibinfo{pages}{129--198},
  \doi{10.1070/SM1977v032n02ABEH002376}.

\bibitemdeclare{article}{Matiyas}
\bibitem{Matiyas}
\bibinfo{author}{Yuri \surnamestart Matiyasevich\surnameend}
  (\bibinfo{year}{1968}): \emph{\bibinfo{title}{A Connection between Systems of
  Word and Length Equations and {H}ilbert's {T}enth {P}roblem (in {R}ussian)}}.
\newblock {\sl \bibinfo{journal}{Sem. Mat. V. A. Steklov Math. Inst.
  Leningrad}} \bibinfo{volume}{8}, pp. \bibinfo{pages}{132--144}.

\bibitemdeclare{book}{NemytykhBook}
\bibitem{NemytykhBook}
\bibinfo{author}{A.~P. \surnamestart Nemytykh\surnameend}
  (\bibinfo{year}{2007}): \emph{\bibinfo{title}{The Supercompiler {S}{C}{P}4:
  General Structure (in {R}ussian)}}.
\newblock \bibinfo{publisher}{URSS, Moscow}.

\bibitemdeclare{inproceedings}{NemytykhUnfold}
\bibitem{NemytykhUnfold}
\bibinfo{author}{Andrei~P. \surnamestart Nemytykh\surnameend}
  (\bibinfo{year}{2014}): \emph{\bibinfo{title}{On Unfolding for Programs Using
  Strings as a Data Type}}.
\newblock In: {\sl \bibinfo{booktitle}{{VPT} 2014. Second International
  Workshop on Verification and Program Transformation, July 17--18, 2014,
  Vienna, Austria, co-located with the 26th International Conference on
  Computer Aided Verification {CAV 2014}}}, pp. \bibinfo{pages}{66--83},
  \doi{10.29007/m8rr}.

\bibitemdeclare{article}{NepeivodaArt}
\bibitem{NepeivodaArt}
\bibinfo{author}{Antonina \surnamestart Nepeivoda\surnameend}
  (\bibinfo{year}{2016}): \emph{\bibinfo{title}{Ping-pong Protocols as Prefix
  Grammars: {M}odelling and Verification via Program Transformation}}.
\newblock {\sl \bibinfo{journal}{Journal of Logical and Algebraic Methods in
  Programming}} \bibinfo{volume}{85}(\bibinfo{number}{5}), pp.
  \bibinfo{pages}{782--804}, \doi{10.1016/j.jlamp.2016.06.001}.
\newblock \bibinfo{note}{Special Issue on Automated Verification of Programs
  and Web Systems}.

\bibitemdeclare{inproceedings}{Pettorossi96}
\bibitem{Pettorossi96}
\bibinfo{author}{Alberto \surnamestart Pettorossi\surnameend} \&
  \bibinfo{author}{Maurizio \surnamestart Proietti\surnameend}
  (\bibinfo{year}{1996}): \emph{\bibinfo{title}{A Comparative Revisitation of
  Some Program Transformation Techniques}}.
\newblock In: {\sl \bibinfo{booktitle}{Partial Evaluation}},
  \bibinfo{publisher}{Springer Berlin Heidelberg}, \bibinfo{address}{Berlin,
  Heidelberg}, pp. \bibinfo{pages}{355--385}, \doi{10.1007/3-540-61580-6\_18}.

\bibitemdeclare{inproceedings}{Plandowski}
\bibitem{Plandowski}
\bibinfo{author}{Wojciech \surnamestart Plandowski\surnameend}
  (\bibinfo{year}{2006}): \emph{\bibinfo{title}{An Efficient Algorithm for
  Solving Word Equations}}.
\newblock In: {\sl \bibinfo{booktitle}{Proceedings of 38th Annual ACM Symposium
  on Theory of Computing}}, \bibinfo{publisher}{Association for Computing
  Machinery}, \bibinfo{address}{New York, NY, USA}, pp.
  \bibinfo{pages}{467--476}, \doi{10.1145/1132516.1132584}.

\bibitemdeclare{inproceedings}{Saxena}
\bibitem{Saxena}
\bibinfo{author}{P.~\surnamestart Saxena\surnameend},
  \bibinfo{author}{D.~\surnamestart Akhawe\surnameend},
  \bibinfo{author}{S.~\surnamestart Hanna\surnameend},
  \bibinfo{author}{F.~\surnamestart Mao\surnameend},
  \bibinfo{author}{S.~\surnamestart McCamant\surnameend} \&
  \bibinfo{author}{D.~\surnamestart Song\surnameend} (\bibinfo{year}{2010}):
  \emph{\bibinfo{title}{A Symbolic Execution Framework for {J}avascript}}.
\newblock In: {\sl \bibinfo{booktitle}{SP}}, pp. \bibinfo{pages}{513--528},
  \doi{10.1109/SP.2010.38}.

\bibitemdeclare{inproceedings}{SecherSorensen99}
\bibitem{SecherSorensen99}
\bibinfo{author}{Jens~P. \surnamestart Secher\surnameend} \&
  \bibinfo{author}{Morten~Heine \surnamestart S{\o}rensen\surnameend}
  (\bibinfo{year}{1999}): \emph{\bibinfo{title}{On Perfect Supercompilation}}.
\newblock In: {\sl \bibinfo{booktitle}{Perspectives of System Informatics,
  Third International Andrei Ershov Memorial Conference, PSI'99, Akademgorodok,
  Novosibirsk, Russia, July 6--9, 1999, Proceedings}}, pp.
  \bibinfo{pages}{113--127}, \doi{10.1007/3-540-46562-6\_10}.

\bibitemdeclare{inproceedings}{Sorensen95}
\bibitem{Sorensen95}
\bibinfo{author}{Morten~H. \surnamestart S{\o}rensen\surnameend} \&
  \bibinfo{author}{Robert \surnamestart Gl\"{u}ck\surnameend}
  (\bibinfo{year}{1995}): \emph{\bibinfo{title}{An Algorithm of Generalization
  in Positive Supercompilation}}.
\newblock In: {\sl \bibinfo{booktitle}{Proceedings of ILPS'95, the
  International Logic Programming Symposium}}, \bibinfo{publisher}{MIT Press},
  pp. \bibinfo{pages}{465--479}, \doi{10.7551/mitpress/4301.003.0048}.

\bibitemdeclare{article}{Jones}
\bibitem{Jones}
\bibinfo{author}{Morten~H. \surnamestart S\o{}rensen\surnameend},
  \bibinfo{author}{Robert \surnamestart Gl\"uck\surnameend} \&
  \bibinfo{author}{Neil~D. \surnamestart Jones\surnameend}
  (\bibinfo{year}{1993}): \emph{\bibinfo{title}{A Positive Supercompiler}}.
\newblock {\sl \bibinfo{journal}{Journal of Functional Programming}}
  \bibinfo{volume}{6}, pp. \bibinfo{pages}{465--479},
  \doi{10.1017/s0956796800002008}.

\bibitemdeclare{inproceedings}{Trinh}
\bibitem{Trinh}
\bibinfo{author}{M.-T. \surnamestart Trinh\surnameend}, \bibinfo{author}{D.-H.
  \surnamestart Chu\surnameend} \& \bibinfo{author}{D.-H. \surnamestart
  Jaffar\surnameend} (\bibinfo{year}{2016}): \emph{\bibinfo{title}{Progressive
  Reasoning over Recursively-Defined Strings}}.
\newblock In: {\sl \bibinfo{booktitle}{Proc. CAV 2016 (LNCS)}},
  \bibinfo{volume}{9779}, pp. \bibinfo{pages}{218--240},
  \doi{10.1007/978-3-319-41528-4\_12}.

\bibitemdeclare{article}{Tur86}
\bibitem{Tur86}
\bibinfo{author}{Valentin~F. \surnamestart Turchin\surnameend}
  (\bibinfo{year}{1986}): \emph{\bibinfo{title}{The Concept of a
  Supercompiler}}.
\newblock {\sl \bibinfo{journal}{ACM Transactions on Programming Languages and
  Systems}} \bibinfo{volume}{8}(\bibinfo{number}{3}), pp.
  \bibinfo{pages}{292--325}, \doi{10.1145/5956.5957}.

\bibitemdeclare{book}{TurRefal5}
\bibitem{TurRefal5}
\bibinfo{author}{Valentin~F. \surnamestart Turchin\surnameend}
  (\bibinfo{year}{1989}): \emph{\bibinfo{title}{Refal-5, Programming Guide and
  Reference Manual}}.
\newblock \bibinfo{publisher}{New England Publishing Co.},
  \bibinfo{address}{Holyoke, Massachusetts}.
\newblock \bibinfo{note}{Electronic version:
  \url{http://www.botik.ru/pub/local/scp/refal5/}}.

\bibitemdeclare{inproceedings}{TurGener}
\bibitem{TurGener}
\bibinfo{author}{Valentin~F. \surnamestart Turchin\surnameend}
  (\bibinfo{year}{1996}): \emph{\bibinfo{title}{On Generalization of Lists and
  Strings in Supercompilation}}.
\newblock In: {\sl \bibinfo{booktitle}{Technical Report CSc. TR 96-002}},
  \bibinfo{publisher}{City College of the City University of New York}, pp.
  \bibinfo{pages}{1--28}.

\bibitemdeclare{article}{Vidal12}
\bibitem{Vidal12}
\bibinfo{author}{Germ{\'{a}}n \surnamestart Vidal\surnameend}
  (\bibinfo{year}{2012}): \emph{\bibinfo{title}{Annotation of Logic Programs
  for Independent {A}{N}{D}-parallelism by Partial Evaluation}}.
\newblock {\sl \bibinfo{journal}{Theory Pract. Log. Program.}}
  \bibinfo{volume}{12}(\bibinfo{number}{4--5}), pp. \bibinfo{pages}{583--600},
  \doi{10.1017/S1471068412000191}.

\bibitemdeclare{inproceedings}{Yu}
\bibitem{Yu}
\bibinfo{author}{Fang \surnamestart Yu\surnameend}, \bibinfo{author}{Tevfik
  \surnamestart Bultan\surnameend} \& \bibinfo{author}{Oscar~H. \surnamestart
  Ibarra\surnameend} (\bibinfo{year}{2011}): \emph{\bibinfo{title}{Relational
  String Verification Using Multi-track Automata}}.
\newblock In \bibinfo{editor}{Michael \surnamestart Domaratzki\surnameend} \&
  \bibinfo{editor}{Kai \surnamestart Salomaa\surnameend}, editors: {\sl
  \bibinfo{booktitle}{Implementation and Application of Automata}},
  \bibinfo{publisher}{Springer Berlin Heidelberg}, \bibinfo{address}{Berlin,
  Heidelberg}, pp. \bibinfo{pages}{290--299},
  \doi{10.1007/978-3-642-18098-9\_31}.

\end{thebibliography}

\section*{Appendix}
\label{section:Appendix}

\subsection{Proofs and Auxiliary Propositions on Optimality}

Given a word $\Phi$, a prefix (suffix) of $\Phi$ is said to be proper if it does not coincide with $\Phi$. Henceforth we consider only the non-empty prefixes, \eg the var-permutated prefixes are non-empty by their definition.

\begin{Prop}\label{Prop:consimprop}
Let $\Eq$ be a reduced word equation $\Phi\seq\Psi$ with the var-permutated sides and without var-permutated suffixes and prefixes; $\sigma$ be an arbitrary substitution given in  \figurename~\ref{tab:classicsubstrules}. Let $\Eq'$ be $\Phi\sigma\seq\Psi\sigma$ after the reduction. Then the two following properties hold.
\begin{enumerate}
\item If $\Eq'$ is split into $\pref_1\conc\suff_1\seq\pref_2\conc\suff_2$ where $\pref_1$ and $\pref_2$ are var-permutated, then either $\sigma$ is $\varx\rar\empt$ or the equations $\pref_1\seq\pref_2$ and $\suff_1\seq\suff_2$ cannot be reduced.
\item If the length of $\Eq'$ is lesser than the length of $\Eq$, then $\sigma$ is $\varx\rar\empt$.
\end{enumerate} 
\end{Prop}
\begin{proof}
\begin{enumerate}
\item Let a substitution $\sigma$ be of the form $\sigma\colon\varx\rar\term\conc\varx$, where $\term\in\CAlph\cup\VAlph$. We may present the equation sides as $\Phi\seq\Phi_1\varx\dots\Phi_{n-1}\varx\conc\Phi_n$ and $\Psi\seq\Psi_1\varx\dots\Psi_{n-1}\varx\conc\Psi_n$, where $\forall i,j\,(|\Phi_i|_\varx\seq 0\logand |\Psi_j|_\varx\seq 0)$. The substitution $\sigma$ results in equation $\Phi_1\term\conc\varx\dots\Phi_{n-1}\conc\term\conc\varx\conc\Phi_n\seq\Psi_1\term\conc\varx\dots\Psi_{n-1}\conc\term\conc\varx\conc\Psi_n$. 

We consider the left-split operation finding var-permutated prefixes with the minimal length. The case of the right-split operation uses the analogous reasoning. If the var-permutated prefixes found by the operation are of the form $\Phi_1\conc\term\conc\varx\dots\conc\term\conc\varx\conc\Phi'_i$ and $\Psi_1\conc\term\conc\varx\dots\conc\term\conc\varx\conc\Psi'_i$, where $\Phi'_i$ and $\Psi'_i$ are prefixes of $\Phi_i$ and $\Psi_i$ respectively, then the words $\Phi_1\conc\varx\dots\conc\varx\conc\Phi'_i$ and $\Psi_1\conc\varx\dots\conc\varx\conc\Psi'_i$ are also var-permutated and the equation should be split until the substitution $\sigma$ is applied. Thereby the proper var-permutated prefixes $\pref_1$ and $\pref_2$ can be only of the forms $\Phi_1\conc\term\conc\varx\dots\conc\term\conc\varx\conc\Phi'_i$ (where $\Phi'_i$ is a prefix of $\Phi_i$) and $\Psi_1\conc\term\conc\varx\dots\conc\term\conc\varx\conc\Psi_i\conc\term$, or vice versa. We consider only the first case, because they are symmetric. 

Thus, $\suff_1$ starts with a term other than $\varx$, while $\suff_2$ starts with $\varx$. Moreover, if the last term of $\Phi'_i$ can be reduced with $\term$, then the word $\Phi_1\conc\varx\dots\conc\varx\conc\Phi'_i$ and $\Psi_1\conc\varx\dots\conc\varx\conc\Psi_i$ are var-permutated\footnote{This reasoning still holds if $i\seq 1$. In the case of $\Psi_1\seq\empt$ or $\Phi_1\seq\empt$, one reduction is to be done until the splitting, but the overall reasoning is the same.}.  
\item Let a substitution $\sigma$ be of the form $\varx\rar\term\conc\varx$, where $\term\in\CAlph\cup\VAlph$ and $\varx$ occurs either in $\Phi$ or in $\Psi$. Otherwise, the substitution has no impact on the equation length. Following the first case proven above, we do not consider reduction operations after splitting the equation. Consider the possible reductions of the equation $\Eq': \Phi\sigma\seq\Psi\sigma$ before it is split. A reduction may occur only if $\Phi$ starts with $\varx$, and $\Psi$ starts with $\term$ (or vice versa).  Let $k \seq |\Phi|+|\Psi|$, then $|\Phi\sigma|+|\Psi\sigma|\geq k+2$, and the considered reduction decreases the length of $\Eq'$ by 2. Thus, after an application of such a substitution $\sigma$ the overall length of the equation cannot decrease.  
\end{enumerate}  
\end{proof}

The following proposition does not hold when the algorithm of interest is $\SolSplit$ or $\SolCount$.

\begin{Prop}\label{Prop:reversibility}
Given a substitution $\sigma\colon\varx\rar\term\conc\varx$, where $\term\in\VAlph\cup\CAlph$ ($\term\neq\varx$), $\sigma$ is an injection on the set of reduced equations when they are simplified by the algorithm $\SolBase$ unless the equations are trivial contradictions.
\end{Prop}
\begin{proof}
Let $\eq$ be a reduced equation $\Phi\seq\Psi$ \st $\Phi$ and $\Psi$ start with different terms, and let us assume that there exists $\eq'$ \st $\eq\seq\eq'\sigma$, where $\sigma$ is compatible with $\eq'$. Let $\eq'$ be $\Phi'\seq\Psi'$. If $\eq'$ is of the reduced form then at most one elementary reduction can be done in $\eq'\sigma$, namely we can reduce the first terms in the left- and right-hand sides of the equation. Moreover, the reduction occurs iff $\Phi'$ starts with $\varx$ and $\Psi'$ with $\term$ (or vice versa). If none of $\Phi$ and $\Psi$ starts with $\varx$, then no reduction is possible in $\eq'\sigma$ and $\eq'$ can be computed as a result of the formal inverse substitution $\sigma^{-1}\colon \term\conc\varx\rar\varx$, namely $\Phi'\seq\Phi\sigma^{-1}$, $\Psi'\seq\Psi\sigma^{-1}$. Let $\Phi\seq\varx\conc\Phi_1$. Then $\Phi'$ is $\varx\conc(\Phi_1\sigma^{-1})$, $\Psi'$ is $\term\conc(\Psi\sigma^{-1})$. 
\end{proof}

Proposition~\ref{Prop:reversibility} states that given an equation $\eq$ and a substitution $\sigma$ of a special kind \st a node $\Node$ in a solution tree generated with the use of the algorithm $\SolBase$ is labelled by $\eq$ and has the ingoing arc labelled with $\sigma$, then the label of the parent of $\Node$ can be restored. But this proposition does not require $\sigma$ to be generated according to the substitutions given in  \figurename~\ref{tab:classicsubstrules}: both $\Phi$ and $\Psi$ may start with terms differing from $\varx$. 

The previous propositions refer to the equation solution trees. The following proof of Lemma~\ref{lemma:optimality} refers to the process trees (graphs). We recall that the notion of a configuration is given in~Definition~\ref{Definition::Config}.

\begin{proof}[Proof of Lemma~\ref{lemma:optimality} --- Optimality Lemma]\label{proof:optlemma}

Let $[\Node_1;\Node_2]$ denote a path segment starting at node $\Node_1$ and ending at node $\Node_2$ .

First, consider the interpreter $\WeqIntBase$. See \sectionname~\ref{subsection::source} for its source code. The node configurations in the process tree of $\WeqIntBase$ are of the following forms:
\begin{enumerate}
\item $\MainStep(\vpar_i, \wholeencodespace{\eq})$ (primary configurations);
\item $\MainStep(\vpar_i,\FunSim(\wholeencodespace{\sigma},\mathrm{\it Other\;args}))$, where $\wholeencodespace{\sigma}$ is static data encoding the substitution that is last applied\footnote{This substitution is stored as an argument of the function $\FunSim$ as an additional annotation of the calls.} to the equation list, and the other arguments may include a call of the function $\FunSubst$.
\end{enumerate}
The function $\MainStep$ applies the first substitution in the list given in its first argument to the equation list given in the second argument; the function $\FunSim$ simplifies, \ie reduces, the equations. If only the primary configurations are folded, the specialization is already optimal. Let the nodes $\Node_1$ and $\Node_2$ be labelled with the configurations of the form $\FunEq(\vpar_i,\FunSim({\wholeencodespace{\sigma},\mathrm{\it Other\;args}}))$ \st the configurations coincide up to the parameter renaming. We will prove now that their ancestors labelled with the primary configurations also should have equal labels, thus, the folding operation should be applied before the whole path part $[\Node_1;\Node_2]$ is unfolded\footnote{The scheme of the proof given in~\figurename~\ref{Fig:OptLemma} refers to the interpreters $\WeqIntSplit$ and $\WeqIntSym$, although the only significant difference between the reasonings is that the simplification function and $\FunEq$ function in $\WeqIntBase$ do not use an additional information about the number of equations in the list, namely, $n$ on the scheme.}. 

If $\Node_1$ and $\Node_2$ are non-primary and are labelled with the equal configurations, then the primary ancestors of the nodes $\Node_1$ and $\Node_2$ always exist. Note that the first call of $\FunSim$ initialized by $\FunGo$, which is not preceded with a primary configuration, is of the form $\FunSim(\empt,\dots)$, while all other calls of $\FunSim$ have a non-empty first argument. This implies that two configurations of the form $\MainStep(\vpar_i,\FunSim(\empt,\mathrm{\it Other\;args}))$ cannot be folded --- that would imply that no par-narrowing is generated along the path segment $[\Node_1;\Node_2]$, thus, the function $\FunSim$ would be non-terminating. Thus, the configurations that label the nodes $\Node_1$ and $\Node_2$, assumed to be folded, can be only of the form $\FunEq(\vpar_i,\FunSim({\wholeencodespace{\sigma},\mathrm{\it Other\;args}}))$ with non-empty $\sigma$. Note that there may be several primary nodes along the segment $[\Node_1;\Node_2]$.

Given the nodes $\Node_1$ and $\Node_2$, let us take their closest primary ancestors, named $\Node'_1$ and $\Node'_2$ respectively. Let the configurations labelling them be of the forms $\MainStep(\vpar, \wholeencodespace{\eq})$ and $\MainStep(\vpar', \wholeencodespace{\eq'})$. Let the closest primary successors of $\Node_1'$ and $\Node_2'$, named $\Node_{\mathrm{loop\_1}}$ and $\Node_{\mathrm{loop\_2}}$ respectively, be labelled with the configurations $\FunEq(\vpar_i,\wholeencodespace{\hspace{0.2ex}\eq_i\hspace{0.2ex}})$. The path segments $[\Node_1;\Node_{\mathrm{loop\_1}}]$ and $[\Node_2;\Node_{loop\_2}]$ contain only transient nodes. Thus, the configurations labelling $\Node_{\mathrm{loop\_1}}$ and $\Node_{\mathrm{loop\_2}}$ coincide up to the name of the parameter $\vpar_i$. Let $\eq_i\seq\eq_0$. No $\LogLang$-narrowings generating a substitution $\varx\rar\empt$ can occur along the segment $[\Node_{\rm{loop\_1}};\Node_{\rm{loop\_2}}]$, otherwise the equation $\eq_0$ would not be preserved in $\Node_{\mathrm{loop\_2}}$. Thus $\sigma$ is $\varx\rar\term\conc\varx$, $t\in\CAlph\cup\VAlph$. Here the substitution $\sigma$ is the same one used in the configurations labelling $\Node_1$ and $\Node_2$, because it is the first argument of the function call $\FunSim$, and, when applied to the equations $\eq$ and $\eq'$, $\sigma$ generates the same equation $\eq_0$. Hence by Proposition~\ref{Prop:reversibility} $\eq$ and $\eq'$ coincide. This proves the lemma in the case of the interpreter $\WeqIntBase$.

Let us consider now specialization of interpreters $\WeqIntSplit$ and $\WeqIntSym$. In these cases the possible configurations are exhausted with the following ones:
\begin{enumerate}
\item $\MainStep(\vpar, (n)\longconc\wholeencodespace{\langle\eq_i\rangle^{n'}_{i=1}})$, where $n$ coincides with $n'$ unless $\langle\eq_i\rangle^{n'}_{i=1}$ consists of a single unsatisfiable equation;
\item $\FunEq(\vpar,\FunSim(n,\wholeencodespace{\sigma},\mathrm{\it Other\;args}))$, where $n$ is the length of the equation list in the configuration labelling the closest primary ancestor of the node considered, and $\sigma$ is the last substitution that was applied to the equations; the other arguments may contain function calls;
\item $\FunEq(\vpar,\FunSort(n,\wholeencodespace{\sigma},\mathrm{\it Other\;args}))$.
\end{enumerate}
The argumentation for the case 3 does not differ from the one that will be given in the case 2, thus we assume the fold operation works only with the nodes $\Node_1$ and $\Node_2$ labelled with the configurations of the form $\FunEq(\vpar_i,\FunSim(n,\wholeencodespace{\sigma},\mathrm{\it Other\;args}))$. Once again we consider their closest ancestor nodes $\Node'_1$ and $\Node'_2$ labelled with the primary configurations $\MainStep(\vpar_0, (n)\longconc\wholeencodespace{\langle\eq_i\rangle^n_{i=1}})$ and $\FunEq(\vpar_{k-1}, (n)\longconc\wholeencodespace{\langle\eq'_i\rangle^n_{i=1}})$, as that was done for $\WeqIntBase$ case. 

\begin{figure}[ht]
\begin{minipage}{0.25\textwidth}
\begin{flushright}${\begin{array}{l}\rm{No~splits}\\\rm{No~}\varx_i\rar\empt\\\rm{No~contradictions}\end{array}}\Biggggl\{$

\vspace*{50pt}\end{flushright}
\end{minipage}
\begin{minipage}{0.73\textwidth}
$$
\xyoption{frame}
\xymatrix@-5mm {
& *+[F]{\Node'_1}\ar[d]\ar@/_2pc/@{=>}[dd]_{\begin{array}{l}\rm{Injection}
\\ \varx\rar\term\conc\varx\end{array}}&*+{\FunEq(\vpar_0,(n)\longconc\wholeencodespace{\langle\eq_i\rangle_{i=1}^n})}\ar@{~>}[l]
 \\
& \ov{\;\Node_1\;}\ar[d] & *{\begin{array}{l}\\ \FunEq(\vpar_1,\FunSim(n,\wholeencode{\hspace{0.4ex}\varx\rar\term\conc\varx\hspace{0.4ex}},\\\quad\qquad\quad\,\;\;\;\;\;\qquad\quad{\rm{\it Other~args}}))\end{array}}\ar@{~>}[l] \\
& *+[F]{\Node_{\rm{loop\_1}}}\ar[d] \\
& *{\skipper\dots\skipper}\ar[d] \\
& *+[F]{\Node'_2}\ar[d]\ar@/_2pc/@{=>}[dd]_{\begin{array}{l}\rm{Injection}\\ \varx\rar\term\conc\varx\end{array}}&*+{\FunEq(\vpar_{k-1},(n)\longconc\wholeencodespace{\langle\eq'_i\rangle_{i=1}^n}\,)}\ar@{~>}[l] \\
& \ov{\;\Node_2\;}\ar@/_2.5pc/@{--}[uuuu]\ar[d] & *{\begin{array}{l}\\\FunEq(\vpar_k,\FunSim(n,\wholeencode{\hspace{0.4ex}\varx\rar\term\conc\varx\hspace{0.4ex}},\\\quad\qquad\quad\;\,\;\;\;\;\;\qquad\quad{\rm{\it Other~args}}))\end{array}}\ar@{~>}[l] \\
& *+[F]{\Node_{\rm{loop\_2}}}&&&\\
&&&&\\
}
$$
\vspace*{-15pt}
$$
\xymatrix@-7mm{
*++[F]{\;}&*{\hspace{-21pt}\--~\rm{primary~configuration}}&&&&&&\\
\ov{\;\upskipper\;}&*{\--~\rm{non\hspace{-2.5pt}-\hspace{-2.5pt}primary~configuration}}&&&&&
}
$$
\end{minipage}
\caption{Scheme of the Optimality Lemma proof.}
\label{Fig:OptLemma}
\end{figure}
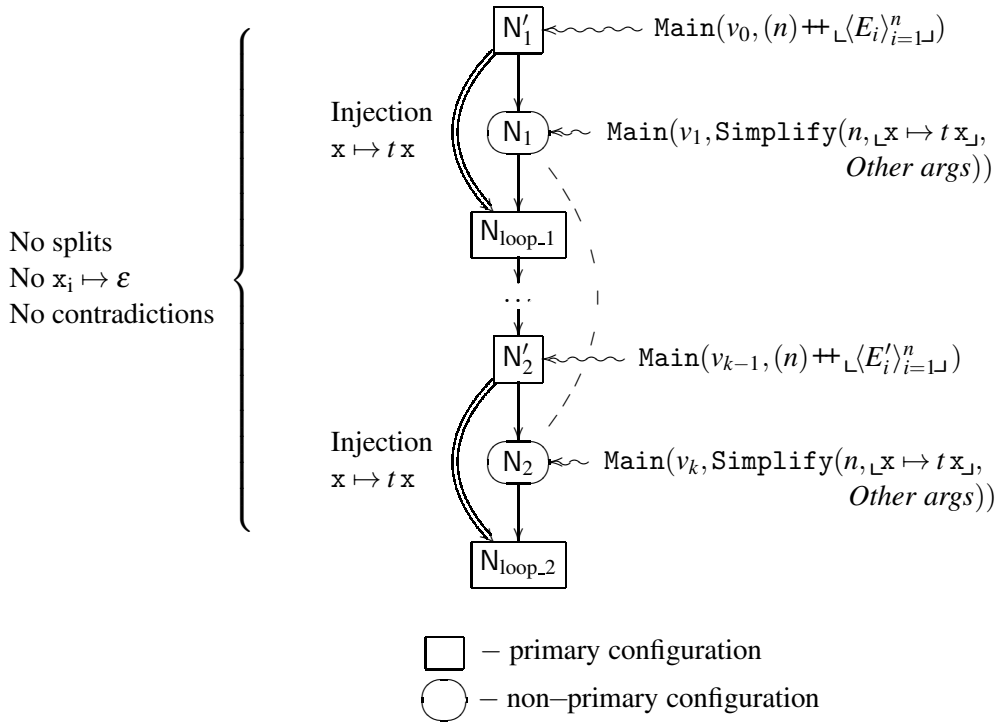

Here the number $n$ is the same in the both configurations, because it is the same in the configurations of the nodes $\Node_1$, $\Node_2$. Function $\FunSim$ takes its first argument from the last $\FunEq$ call, hence they are the same in all the four configurations. The substitution $\varx\rar\empt$ cannot be applied along the path segment $[\Node_1;\Node_2]$, by the reasoning above. Let $\sigma$ be $\varx\rar\term\conc\varx$, $t\in\CAlph\cup\VAlph$. By Proposition~\ref{Prop:consimprop}, the number of the equations in $[\Node'_1;\Node'_2]$ cannot decrease. Thus, the number of the equations is a constant along the path segment, namely $n$, and no equation in a configuration in $[\Node'_1;\Node'_2]$ can be split. That means every equation in a configuration along the path starting at $\Node'_1$ is transformed as it would be transformed by the algorithms implemented in $\WeqIntBase$. Given the first primary successors $\Node_{\rm{loop\_1}}$ and $\Node_{\rm{loop\_2}}$ of $\Node_1$ and $\Node_2$, by their choice, the segments $[\Node_1;\Node_{\rm{loop\_1}}]$ and $[\Node_2;\Node_{\rm{loop\_2}}]$ consist only of the transient nodes, Thus, $\Node_{\rm{loop\_1}}$ and $\Node_{\rm{loop\_2}}$ are labelled with the equal lists of equations, thus by Proposition~\ref{Prop:reversibility} the lists of equations $\langle\eq'_i\rangle^n_{i=1}$ and $\langle\eq_i\rangle^n_{i=1}$ also coincide. 
\end{proof}

\subsection{Proofs and Auxiliary Propositions on Regularly-Ordered Equations with Repetitions}

The next propositions consider the equation solution algorithms presented in \sectionname~\ref{Sect:WeqInt} and do not refer to the interpreters' structure. The equation solution graphs are considered here, like ones shown in \figurename~\ref{fig:solscheme}. 

\begin{Prop}
Let $Q_i,Q'_i\in\CAlph^+$. Given a list of equations $\Eqs\seq \langle Q_1\conc\varx_1\seq\varx_1\conc Q'_1,\dots, Q_n\conc\varx_n\seq\varx_n\conc Q'_n\rangle$, where for some $i,j$, $\varx_i\seq\varx_j$ may hold, the label of any node in the solution tree constructed with the use of the algorithm $\SolSplit(\Eqs)$ is a list consisting of at most $n$ equations. If it consists exactly of $n$ equations $\langle \Phi'_1\seq\Psi'_1,\dots, \Phi'_n\seq\Psi'_n\rangle$, then $|\Phi'_i|\leq |Q_i|+1$, $|\Psi'_i|\leq |Q'_i|+1$.
\end{Prop}
\begin{proof}
 \figurename~\ref{tab:classicsubstrules} shows that the possible substitutions that can be compatible with the initial equation list $\langle Q_1\conc\varx_1\seq\varx_1\conc Q'_1,\dots, Q_n\conc\varx_n\seq\varx_n\conc Q'_n\rangle$ are either $\varx_i\rar\empt$ or $\varx_i\rar\cA_i\conc\varx_i$. The first one transforms equations including $\varx_i$ to either tautologies or contradictions. The second one transforms an equation $Q_i\conc\varx_i\seq\varx_i\conc Q'_i$ either to a contradiction or to an equation of the form $R_i\conc\varx_i\seq\varx_i\conc Q'_i$, where $R_i$ is a cyclic permutation of the word $Q_i$, $|R_i|\seq|Q_i|$. 
\end{proof}

\begin{Coroll}\label{Coroll:commut}
If $Q_i,Q'_i\in\CAlph^+$ then every infinite path of a solution graph constructed with the use of the algorithm $\SolSplit$ applied to the list $\langle Q_1\conc\varx_1\seq\varx_1\conc Q'_1,\dots, Q_n\conc\varx_n\seq\varx_n\conc Q'_n\rangle$ contains two nodes with equal labels. 
\end{Coroll}

\begin{Prop}
Let $\Phi\seq\Psi$ be a strictly regular-ordered equation with repetitions. Then the following statements hold.
\begin{enumerate}
\item Given the shortest non-empty var-permutated prefixes $\Phi_1$ and $\Psi_1$ \st $\Phi\seq\Phi_1\conc\Phi_2$, $\Psi\seq\Psi_1\conc\Psi_2$, the equations $\Phi_1\seq\Psi_1$ and $\Phi_2\seq\Psi_2$ are strictly regular-ordered with repetitions.
\item The solution tree constructed with the use of the algorithm $\SolSplit(\Phi\seq\Psi)$ never includes an application of the $\LogLang$-narrowing $\varx\rar\vary\conc\varx$ given in  \figurename~\ref{tab:classicsubstrules}.
\item Every infinite path in the solution tree constructed with the use of $\SolSplit(\Phi\seq\Psi)$ contains a finite number of the split operations.
\end{enumerate}
\end{Prop}
\begin{proof}
\begin{enumerate}
\item For every variable $\varx_i\in\VAlph$, the var-permutated property gives $|\Phi_1|_{\varx_i}\seq|\Psi_1|_{\varx_i}$, and hence $|\Phi_2|_{\varx_i}\seq|\Psi_2|_{\varx_i}$. The order of the variable occurrences in $\Phi\seq\Psi$ is preserved in the prefixes and suffixes as well.
\item Consider the variable $\varx$ leading in $\Phi$ (and occurring in $\Psi$ before any other variable). Then $\Phi\seq\varx\conc\Phi'$, $\Psi\seq\Psi_0\conc\varx\conc\Psi'$ (or vice versa), where $\Psi_0\in\CAlph^+$. Let $\Psi_0$ be $\cA\conc\Psi'_0$ then an $\LogLang$-narrowing unfolding the equation $\varx\conc\Phi'\seq\Psi_0\conc\varx\conc\Psi'$ is either $\varx\rar\empt$ or $\varx\rar\cA\conc\varx$. Both of the substitutions preserve the strictly-regular-ordered property, as well as the splitting operation does.
\item The statement (2) implies that given a list $\langle \Phi_1\seq\Psi_1,\dots,\Phi_n\seq\Psi_n\rangle$ of equations labelling a node in the solution tree of $\Phi\seq\Psi$, for every $\varx_i\in\VAlph$, $\displaystyle\sum_{j=1}^n|\Phi_j|_{\varx_i}+|\Psi_j|_{\varx_i}\leq|\Phi|_{\varx_i}+|\Psi|_{\varx_i}$. And every split operation generates two equations containing at least two variables.      
\end{enumerate}
\end{proof}



Here we repeat Lemma~\ref{lemma:regcommeq} before giving its detailed proof.

\noindent{\textbf{Lemma 2.}} \emph{Given a strictly regular-ordered equation with repetitions $\Phi\seq\Psi$, every infinite path in its solution tree constructed with the use of the algorithm $\SolSplit(\Phi\seq\Psi)$ includes at least two nodes with equal labels.}

\begin{proof}\label{proof::regcommeq}
Every infinite path in the tree generated with the use of $\SolSplit(\Phi\seq\Psi)$ has an infinite subpath satisfying the following two conditions:
\begin{enumerate}
\item equations are never split along the subpath;
\item variables are never mapped to $\empt$ along the subpath.
\end{enumerate}
Let the first node in such a subpath be $\Node$, and the label of $\Node$ be a list $\langle\Theta\conc\varx\conc\Phi_1\seq\varx\conc\Psi_1,\dots, \Phi_n\seq\Psi_n\rangle$, $\Theta\in\CAlph^+$. There may be only the following three options. 
\begin{enumerate}
\item For every $j$, if $|\Phi_j|_{\varx}>0$ then $\Phi_j\seq\varx\conc\Phi'_j$, $\Psi_j\seq\Theta_j\conc\varx\conc\Psi'_j$ (or vice versa), $\Theta_j\in\CAlph^+$, $|\Phi'_{j}|_\varx\seq|\Psi'_{j}|_\varx\seq 0$. Given such an equation, a~substitution $\sigma\colon\varx\rar \term\conc\varx$ ($\term\in\CAlph$), followed by the reduction, preserves its length. The number of the equations in the lists labelling the nodes along the path, as well as the lengths of the equations, cannot grow, hence the $\LogLang$-narrowings do not generate fresh variables, and thus the set of the labels of the nodes along the path is finite.
\item Some equation $\Eq_j\colon \Phi_j\seq\Psi_j$ is of the form $\Phi_{0,j}\conc\varx\conc\Phi_{1,j}\conc\varx\conc\Phi_{2,j}\seq\varx\conc\Psi_{1,j}\conc\varx\conc\Psi_{2,j}$, where $\Phi_{0,j}\in\CAlph^+$, $|\Phi_{1,j}|_{\varx}\seq 0$, $|\Psi_{1,j}|_{\varx}\seq 0$. Equation $\Eq_j$ is strictly regular-ordered, hence $\forall k\,(|\Phi_{1,j}|_{\varx_k}\seq|\Psi_{1,j}|_{\varx_k})$. Let $m\seq \left| |\Phi_{0,j}|+|\Phi_{1,j}|-|\Psi_{1,j}|\right|$; $m\neq 0$, otherwise $\Eq_j$ would be split. Consider the path segment starting at $\Node$ and having the length $m$. The arcs in this segment are labelled by the substitutions $\varx\rar c_1\conc\varx,\dots,\varx\rar c_m\conc\varx$, where $c_i\in\CAlph$ are letters of $\Theta$. Let $\Theta'\seq c_1\dots c_m$. The ending node of the $m$-length path segment is labelled with the list containing the following equation: $\langle\dots,\Phi'_{0,j}\conc\varx\conc\Phi_{1,j}\conc\Theta'\conc\varx\conc\Phi'_{2,j}\seq\varx\conc\Psi_{1,j}\conc\Theta'\conc\varx\conc\Psi'_{2,j},\dots\rangle$, $|\Phi'_{0,j}|\seq|\Phi_{0,j}|$. If $|\Phi_{0,j}|+|\Phi_{1,j}|-|\Psi_{1,j}|>0$, then the prefixes $\Phi'_{0,j}\conc\varx\conc\Phi_{1,j}$ and $\varx\conc\Psi_{1,j}\conc\Theta$ are var-permutated, otherwise the prefixes $\Phi'_{0,j}\conc\varx\conc\Phi_{1,j}\conc\Theta$ and $\varx\conc\Psi_{1,j}$ are var-permutated. In any case, there exists such a $k\leq m$ that\footnote{$k<m$, if $\Phi_{1,j}$ or $\Psi_{1,j}$ do not end with variables, and $k=m$ otherwise.} a split takes place in $k$-th node along the path segment.
\item Some equation $\Eq_j\colon \Phi_j\seq\Psi_j$ is of the form $\Phi_{1,j}\conc\varx\conc\Phi_{2,j}\seq\Psi_{1,j}\conc\varx\conc\Psi_{2,j}$, $|\Phi_{1,j}|_{\varx}\seq 0$, $|\Psi_{1,j}|_{\varx}\seq 0$, $\Phi_{1,j},\Psi_{1,j}\notin\CAlph^+$. Let $m\seq \left| |\Phi_{1,j}|-|\Psi_{1,j}|\right|$. The same arguments show that given such an $\Eq_j$ a~split would occur along the path at most after $m$ substitutions. If $\Psi_{1,j}$ or $\Phi_{1,j}$ do not end with a variable, then the split can be applied earlier. That case is given in  \figurename~\ref{fig:proofTh}, where $\Phi_{1,j}\seq\Phi'_0\conc\vary\conc\Phi'_1$, $\Psi_{1,j}\seq\Psi'_0\conc\vary\conc\Psi'_1$, $\Phi'_1,\Psi'_1\in\CAlph^+$.
\end{enumerate}
Thus, either the nodes with the equal labels exist along the given subpath or a split is constructed. That contradicts the choice of the subpath.
\end{proof}

\noindent
\begin{figure}[t]
\begin{tabular}{lll}
\footnotesize
\begin{tikzpicture}[scale=0.9]
\draw (0,0) rectangle (4.5,0.5);
\draw [pattern=north east lines,  pattern color=gray] (0,0) rectangle (1,0.5);
\node at (0.52,0.25) {$\Phi'_0$};
\draw [fill=gray!10](1.4,0) rectangle (1.85,0.5);
\node at (1.2,0.22) {$\vary$};
\node at (1.65,0.25) {$\Phi'_1$};
\node at (2.05,0.25) {$\varx$};
\draw [pattern=north west lines,  pattern color=gray](2.25,0) rectangle (4.5,0.5);
\node at (3.3,0.25) {$\Phi_2$};
\draw [thick, ->] (1.85,0.8) -- (1.85,0.5);

\draw (0,-1) rectangle (5.5,-0.5);
\draw [pattern=north east lines,  pattern color=gray] (0,-1) rectangle (3,-0.5);
\node at (1.5,-0.75) {$\Psi'_0$};
\draw [fill=gray!10](3.4,-1) rectangle (4,-0.5);
\node at (3.2,-0.77) {$\vary$};
\node at (3.7,-0.75) {$\Psi'_1$};
\node at (4.22,-0.75) {$\varx$};
\draw [pattern=north west lines,  pattern color=gray](4.45,-1) rectangle (5.5,-0.5);
\node at (5,-0.75) {$\Psi_2$};
\draw [thick, ->] (4,-0.2) -- (4,-0.5);
\end{tikzpicture}
&\begin{tikzpicture}[scale=0.7]
\path [fill=white](0,0) rectangle (0.7,0.7);
\draw [thick, ->] (0.1,1.05) -- (0.55,1.05);
\draw [thick, ->] (0.1,1.25) -- (0.55,1.25);
\draw [thick, ->] (0.1,0.85) -- (0.55,0.85);
\end{tikzpicture}&
\begin{tikzpicture}[scale=0.9]\footnotesize
\draw (0,0) rectangle (5.9,0.5);
\draw [pattern=north east lines,  pattern color=gray] (0,0) rectangle (1,0.5);
\node at (0.52,0.25) {$\Phi'_0$};
\draw [fill=gray!10](1.4,0) rectangle (1.85,0.5);
\draw [fill=gray!10](1.85,0) rectangle (3.4,0.5);
\node at (1.2,0.22) {$\vary$};
\node at (1.65,0.25) {$\Phi'_1$};
\node at (2.6,0.25) {$\Theta$};
\node at (3.6,0.25) {$\varx$};
\draw [pattern=north west lines,  pattern color=gray] (3.8,0) rectangle (5.9,0.5);
\node at (4.8,0.25) {$\Phi'_2$};
\draw (0,-1) rectangle (6.9,-0.5);
\draw [pattern=north east lines,  pattern color=gray] (0,-1) rectangle (3,-0.5);
\node at (1.5,-0.75) {$\Psi'_0$};
\draw [fill=gray!10](3.4,-1) rectangle (4,-0.5);
\draw [fill=gray!10](4,-1) rectangle (5.55,-0.5);
\node at (3.2,-0.77) {$\vary$};
\node at (3.7,-0.75) {$\Psi'_1$};
\node at (4.72,-0.75) {$\Theta$};
\node at (5.75,-0.75) {$\varx$};
\draw [pattern=north west lines,  pattern color=gray](5.95,-1) rectangle (6.9,-0.5);
\node at (6.45,-0.75) {$\Psi'_2$};
\draw [dashed,thick](-0.2,-1.2) -- (-0.2,0.7);
\draw [dashed,thick](-0.2,0.7) -- (3.4,0.7);
\draw [dashed,thick](-0.2,-1.2) -- (3.4,-1.2);
\draw [dashed,ultra thick](3.4,-1.2) -- (3.4,0.7);
\end{tikzpicture}
\end{tabular}

\begin{tabular}{ll}
\begin{tikzpicture}
\draw [fill=gray!10](0,1)rectangle (0.4,1.4);
\node at (1.4,1.2) { --- in $\CAlph^*$;};
\draw [pattern=north west lines,  pattern color=gray](0,0)rectangle (0.4,0.4);
\node at (3.18,0.2) { --- in $\{\CAlph\cup\VAlph\}^+$, may contain $\varx$.};
\draw [pattern=north east lines,  pattern color=gray](0,.5)rectangle (0.4,0.9);
\node at (3.5,0.7) { --- in $\{\CAlph\cup\VAlph\}^+$, does not contain $\varx$;};
\end{tikzpicture}
\end{tabular}
\caption{Splitting a strictly regular-ordered equation with repetitions resulting from substitution $\varx\sigma\seq\Theta\conc\varx$, where $|\Theta|\seq |\Psi'_0|-|\Phi'_0|-|\Phi'_1|$.}
\label{fig:proofTh}
\end{figure}
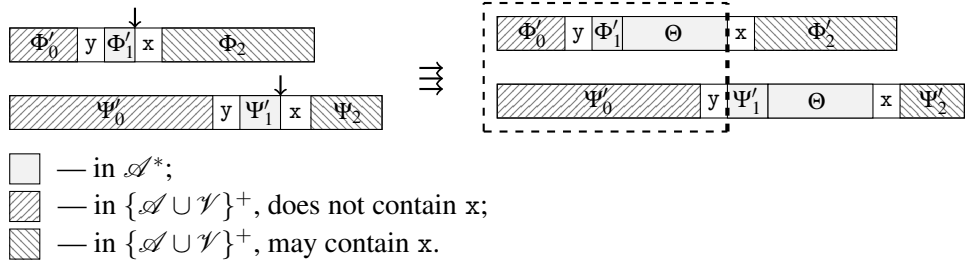

\subsection{Proofs and Auxiliary Propositions on One-Variable Equations}

The next two propositions use the following notations. The letters $T$ and $T'$ stand for words in $\CAlph\cup\{\varx\}$; $t_i$, $u_i$ are letters in $\CAlph$.

\begin{Prop}\label{prop:firstsym}
Given an equation $t_1\dots t_n\conc\varx\conc T \seq \varx\conc u_1\dots u_m \conc\varx\conc T'$ \st $n > 0$, $m\geq 0$, every infinite path of its solution graph contains a split.
\end{Prop}
\begin{proof}
If $m\geq n$, the initial equation generates a split. Let $n\seq m+k$, $k>0$, $\varx\sigma_i\seq t_i\conc\varx$. The $k$-th unfolding step observes that the equation is either already split or is of the following form:

$$t_{k+1}\dots t_n\conc t_1\dots t_k\conc\varx\conc T\sigma_k \seq \varx\conc u_1\dots u_m\conc t_1\dots t_k\conc \varx \conc T'\sigma_k.$$
This equation has var-permutated prefixes.
\end{proof}

Here we repeat Lemma~\ref{lemma:onevareq} before giving its whole proof.

\noindent{\textbf{Lemma 3.}} \emph{
Given an equation $\Phi\seq\Psi$, where $\Phi,\Psi\in \{\CAlph\cup \{\varx\}\}^*$, $|\Phi\conc\Psi|_{\varx}>2$, every infinite path of its solution graph contains a split.}
 
\begin{proof}\label{proof:onevareq}
If at least one side of the equation does not contain $\varx$, the equation always generates a finite solution tree. Thus, we consider only the following three cases.
\begin{enumerate}
\item $t_1\dots t_n\conc\varx\conc T \seq \varx\conc u_1\dots u_m \conc\varx\conc T'$, where $t_i\in\CAlph$, $u_j\in \CAlph$, $n > 0$. This case is considered in Proposition~\ref{prop:firstsym}.
\item $\varx\conc u_1\dots u_m\seq t_1\dots t_n\conc\varx\conc \Phi\conc\varx$, where $t_i\in\CAlph$, $u_i\in\CAlph$. Thus, the left-hand side of the equation contains the only occurrence of $\varx$.
\item $\varx\seq t_1\dots t_n\conc\varx\conc \Phi\conc\varx\conc u_1\dots u_k$, $k\geq 1$. The equation is contradictory according to Proposition~\ref{Prop:count1}.
\end{enumerate}
Only the case (2) is of our interest. We assume $m>n$, otherwise the equation is split. Let $m\seq n*i+j$, $j<n$, $\varx\sigma\seq t_1\dots t_n\varx$, $\varx\sigma_j\seq t_1\dots t_j\conc\varx$, $\xi\seq \sigma^i\conc\sigma_j$. Thus, $\sigma$ is the composition of $n$ elementary substitutions, and $\sigma_j$ is the composition of $j$ elementary substitutions. On the $m$-th unfolding step either the equation is split already or is of the following form: $\varx\conc u_1 \dots u_m \seq t_{j+1}\dots t_n\conc t_1\dots t_j\conc\varx \Phi\xi\conc (t_1\dots t_n)^i\conc t_1\dots t_j\conc\varx$. This equation has the var-permutated suffixes and is split.
\end{proof}
\newpage

\subsection{Source Pseudocode of Interpreters}\label{subsection::source}

We recall that the expression-type variables start with $\varx$, and may be subscripted (\eg $\varr$) or followed by other letters (\eg $\varms$), or both. The symbol-type variables are denoted with $\vars$, maybe subscripted. The constructor $\longconc$ is sometimes omitted, mainly in expressions enclosed in the parentheses.

The data encoding is given in~\figurename~\ref{fig:encoding}; thus, the equation variables are encoded with the two symbols enclosed in the parentheses. The parentheses are also used to form a structure of functions' arguments, \eg $\FunEq$ takes two arguments, where the second one is a pair; the function $\FunSim$ returns a pair. 

We use the following syntactic sugar in the source pseudocode of the interpreters. Variables $\nnumber$ and $\newnumber$ range over the natural numbers. The operations +1 and -1 may be applied to these variable values instead of the corresponding arithmetic operations taking numbers given in the unary Peano system. An equation $\Phi_{\rm{lhs}}\seq\Phi_{\rm{rhs}}$ is encoded as $(\Phi_{\rm{lhs}},\Phi_{\rm{rhs}})$ instead of $((\Phi_{\rm{lhs}})(\Phi_{\rm{rhs}}))$ (see also~\figurename~\ref{fig:encoding}). For example, the encoding of the equation $\cA\conc\varx\conc\vary\seq\varx\conc\vary\conc\cA$ is as follows: $(\cA\conc(\codex\conc\cX)\conc(\codex\conc\cY),(\codex\conc\cX)\conc(\codex\conc\cY)\conc\cA)$.

\subsubsection{Basic Interpreter}
\footnotesize
$$
\begin{array}{lll}
\FunGo(\varr,\vareqlist)\seq \FunEq(\varr, \FunSim(\empt,\vareqlist));\\
\\
\FunEq(\empt,(\empt,\empt))\seq \OutTrue; \\
\FunEq(((\codex\svar_x)\rar\empt)\longconc\varr, (\vareq_{\rm{LHS}},(\codex\svar_x)\longconc\vareq_{\rm{RHS}})) \\\quad\seq\FunEq(\varr,\FunSim(((\codex\svar_x)\rar\empt),\FunSubst((\codex\svar_x)\rar\empt,\empt,\vareq_{\rm{LHS}}),\FunSubst((\codex\svar_x)\rar\empt,\empt,\vareq_{\rm{RHS}})));\\
\FunEq(((\codex\svar_x)\rar\empt)\longconc\varr, ((\codex\svar_x)\longconc\vareq_{\rm{LHS}},\vareq_{\rm{RHS}})) \\\quad\seq \FunEq(\varr,\FunSim(((\codex\svar_x)\rar\empt),\FunSubst((\codex\svar_x)\rar\empt,\empt,\vareq_{\rm{LHS}}),\FunSubst((\codex\svar_x)\rar\empt,\empt,\vareq_{\rm{RHS}})));\\
\FunEq(((\codex\svar_x)\rar\vars\conc(\codex\svar_x))\longconc\varr, ((\codex\svar_x)\longconc\vareq_{\rm{LHS}},\vars\longconc\vareq_{\rm{RHS}})) \\\quad\seq \FunEq(\varr,\FunSim(((\codex\svar_x)\rar\vars\conc(\codex\svar_x)),
\\\qquad\qquad\qquad\qquad\quad\FunSubst((\codex\svar_x)\rar\vars\conc(\codex\svar_x),(\codex\varx),\vareq_{\rm{LHS}}),
\\\qquad\qquad\qquad\qquad\quad\FunSubst((\codex\svar_x)\rar\vars\conc(\codex\svar_x),\empt,\vareq_{\rm{RHS}})));\\
\FunEq(((\codex\svar_x)\rar\vars\conc(\codex\svar_x))\longconc\varr, (\vars\longconc\vareq_{\rm{LHS}},(\codex\svar_x)\longconc\vareq_{\rm{RHS}})) \\\quad\seq \FunEq(\varr,\FunSim(((\codex\svar_x)\rar\vars\conc(\codex\svar_x)),\\
\qquad\qquad\qquad\qquad\quad\FunSubst((\codex\svar_x)\rar\vars\conc(\codex\svar_x),\empt,\vareq_{\rm{LHS}}),
\\\qquad\qquad\qquad\qquad\quad\FunSubst((\codex\svar_x)\rar\vars\conc(\codex\svar_x),(\codex\svar_x),\vareq_{\rm{RHS}})));\\
\FunEq(((\codex\svar_x)\rar(\codex\svar_y)\conc(\codex\svar_x))\longconc\varr,((\codex\svar_y)\longconc\vareq_{\rm{LHS}},(\codex\svar_x)\longconc\vareq_{\rm{RHS}})) \\\quad\seq\FunEq(\varr,\FunSim(((\codex\svar_x)\rar(\codex\svar_y)\conc(\codex\svar_x)),\\
\qquad\qquad\qquad\qquad\quad\FunSubst((\codex\svar_x)\rar(\codex\svar_y)\conc(\codex\svar_x),\empt,\vareq_{\rm{LHS}}),
\\\qquad\qquad\qquad\qquad\quad\FunSubst((\codex\svar_x)\rar(\codex\svar_y)\conc(\codex\svar_x),(\codex\svar_x),\vareq_{\rm{RHS}})));\\
\FunEq(((\codex\svar_x)\rar(\codex\svar_y)\conc(\codex\svar_x))\longconc\varr,((\codex\svar_x)\longconc\vareq_{\rm{LHS}},(\codex\svar_y)\longconc\vareq_{\rm{RHS}})) \\\quad\seq\FunEq(\varr,\FunSim(((\codex\svar_x)\rar(\codex\svar_y)\conc(\codex\svar_x)),\\
\qquad\qquad\qquad\qquad\quad\FunSubst((\codex\svar_x)\rar(\codex\svar_y)\conc(\codex\svar_x),(\codex\svar_x),\vareq_{\rm{LHS}}),
\\\qquad\qquad\qquad\qquad\quad\FunSubst((\codex\svar_x)\rar(\codex\svar_y)\conc(\codex\svar_x),\empt,\vareq_{\rm{RHS}})));\\
\FunEq(\varr,(\vareq_{\rm{LHS}},\vareq_{\rm{RHS}})) \seq \OutFalse;
\\\\
/\hspace{-0.3ex}*\;\mathrm{The\;first\;argument\;of\;}\FunSubst\mathrm{\;is\;the\;substitution,\;the\;second\;one\;serves\;as\;an\;accumulator.}\;*\hspace{-0.3ex}/\\
\FunSubst((\codex\svar_x)\rar\varv,\varres,\empt)\seq\varres;\\
\FunSubst((\codex\svar_x)\rar\varv,\varres,(\codex\svar_x)\longconc\varexpr)\seq
\FunSubst((\codex\svar_x)\rar\varv,\varres\longconc\varv,\varexpr);\\
\FunSubst((\codex\svar_x)\rar\varv,\varres,\vars\longconc\varexpr)\seq
\FunSubst((\codex\svar_x)\rar\varv,\varres\longconc\vars,\varexpr);\\
\FunSubst((\codex\svar_x)\rar\varv,\varres,(\codex\svar_y)\longconc\varexpr)\seq
\FunSubst((\codex\svar_x)\rar\varv,\varres\longconc(\codex\svar_y),\varexpr);
\\ \\
\FunSim(\varx_{\rm{subst}},(\codex\svar_x)\longconc\varx_{\rm{LHS}},(\codex\svar_x)\longconc\varx_{\rm{LHS}})\seq
\FunSim(\varx_{\rm{subst}},\varx_{\rm{LHS}},\varx_{\rm{LHS}}); \\
\FunSim(\varx_{\rm{subst}},\vars\longconc\varx_{\rm{LHS}},\vars\longconc\varx_{\rm{LHS}})\seq
\FunSim(\varx_{\rm{subst}},\varx_{\rm{LHS}},\varx_{\rm{LHS}}); \\
\FunSim(\varx_{\rm{subst}},\vars_1\longconc\varx_{\rm{LHS}},\vars_2\longconc\varx_{\rm{LHS}})\seq
(\vars_1,\vars_2); \\
\FunSim(\varx_{\rm{subst}},\varx_{\rm{LHS}}\longconc(\codex\svar_x),\varx_{\rm{LHS}}\longconc(\codex\svar_x))\seq
\FunSim(\varx_{\rm{subst}},\varx_{\rm{LHS}},\varx_{\rm{LHS}}); \\
\FunSim(\varx_{\rm{subst}},\varx_{\rm{LHS}}\longconc\vars,\varx_{\rm{LHS}}\longconc\vars)\seq
\FunSim(\varx_{\rm{subst}},\varx_{\rm{LHS}},\varx_{\rm{LHS}}); \\
\FunSim(\varx_{\rm{subst}},\varx_{\rm{LHS}}\longconc\vars_1,\varx_{\rm{LHS}}\longconc\vars_2)\seq
(\vars_1,\vars_2); \\
\FunSim(\varx_{\rm{subst}},\varx_{\rm{LHS}},\varx_{\rm{LHS}})\seq
(\varx_{\rm{LHS}},\varx_{\rm{RHS}}); 
\end{array} $$
\newpage
\normalsize
\subsubsection{Splitting Interpreter}
This interpreter has the following refinements as compared to $\WeqIntBase$. In manipulates a list of equations rather than a single equations, and uses additional simplifying functions.
\begin{itemize}
\item The functions $\FunSim$ and $\FunSort$ use an additional argument $\nnumber$ --- a natural number which is 0 if the equations in the list are unchecked or contradictory and is the length of the list otherwise. This argument is used as the annotation that prevents unwanted fold operations in process trees. 
\item The second argument of the function $\FunEq$ is a list of equations concatenated with the natural number $\nnumber$, described above. 
\item The function $\FunSplit$ and the auxiliary multiset-handling function are added. In order to guarantee that all the equations resulting from a split are reduced, we introduce the $\FunSimPrim$ function reducing a given single equation. The new function $\FunSort$ transforms a list of equations to a single unsatisfiable equation if at least one contradiction is found in the list, otherwise the function $\FunSort$ counts how many equations are included in the list.
\item A number of rewriting rules marked with the corresponding comments are given in a sugared syntax. If a symbol and a variable are treated in the same way, then instead of the two (or four) rules we write the only one, where the term considered is replaced by the letter $t$, maybe subscripted.
\end{itemize}

\footnotesize
$$
\begin{array}{lll}
\FunGo(\varr,\vareqlist) \seq \FunEq(\varr, \FunSim(0,\empt,\empt,\vareqlist));\\
\\
\FunEq(\empt,(\nnumber)\longconc((\empt,\empt))) \seq \OutTrue; \\
\FunEq(\varr,(\nnumber)\longconc((\empt,\empt)\longconc\othereq)) \seq \FunEq(\varr,(\nnumber-1)\longconc\othereq); \\
\FunEq(((\codex\svar_x)\rar\empt)\longconc\varr,(\nnumber)\longconc((\vareq_{\rm{LHS}},(\codex\svar_x)\conc\vareq_{\rm{RHS}})\longconc\othereq)) \\\quad\seq \FunEq(\varr,\FunSim(\nnumber,((\codex\svar_x)\rar\empt),\empt,
\\\qquad\qquad\qquad\qquad\quad(\FunSubst((\codex\svar_x)\rar\empt,\empt,\vareq_{\rm{LHS}}),\FunSubst((\codex\svar_x)\rar\empt,\empt,\vareq_{\rm{RHS}}))
\\\qquad\qquad\qquad\qquad\quad\longconc\FunSubstAll((\codex\svar_x)\rar\empt,\othereq)));\\
\FunEq(((\codex\svar_x)\rar\empt)\longconc\varr, (\nnumber)\longconc(((\codex\svar_x)\conc\vareq_{\rm{LHS}},\vareq_{\rm{RHS}})\longconc\othereq)) \\\quad \seq \FunEq(\varr,\FunSim(\nnumber,((\codex\svar_x)\rar\empt),\empt,
\\\qquad\qquad\qquad\qquad\quad(\FunSubst((\codex\svar_x)\rar\empt,\empt,\vareq_{\rm{LHS}}),\FunSubst((\codex\svar_x)\rar\empt,\empt,\vareq_{\rm{RHS}}))
\\\qquad\qquad\qquad\qquad\quad\longconc\FunSubstAll((\codex\svar_x)\rar\empt,\othereq)));\\
\FunEq(((\codex\svar_x)\rar\vars\conc(\codex\svar_x))\longconc\varr, (\nnumber)\longconc(((\codex\svar_x)\conc\vareq_{\rm{LHS}},\vars\conc\vareq_{\rm{RHS}})\longconc\othereq)) \\\quad\seq \FunEq(\varr,\FunSim(\nnumber,((\codex\svar_x)\rar\vars\conc(\codex\svar_x)),\empt,
\\\qquad\qquad\qquad\qquad\quad(\FunSubst((\codex\svar_x)\rar\vars\conc(\codex\svar_x),(\codex\svar_x),\vareq_{\rm{LHS}}),\FunSubst((\codex\svar_x)\rar\vars\conc(\codex\svar_x),\empt,\vareq_{\rm{RHS}}))
\\\qquad\qquad\qquad\qquad\quad\longconc\FunSubstAll((\codex\svar_x)\rar\vars\conc(\codex\svar_x),\othereq)));\\
\FunEq(((\codex\svar_x)\rar\vars\conc(\codex\svar_x))\longconc\varr, (\nnumber)\longconc((\vars\conc\vareq_{\rm{LHS}},(\codex\svar_x)\conc\vareq_{\rm{RHS}})\longconc\othereq)) \\\quad\seq \FunEq(\varr,\FunSim(\nnumber,((\codex\svar_x)\rar\vars\conc(\codex\svar_x)),\empt,
\\\qquad\qquad\qquad\qquad\quad(\FunSubst((\codex\svar_x)\rar\vars\conc(\codex\svar_x),\empt,\vareq_{\rm{LHS}}),\FunSubst((\codex\svar_x)\rar\vars\conc(\codex\svar_x),(\codex\svar_x),\vareq_{\rm{RHS}}))
\\\qquad\qquad\qquad\qquad\quad\longconc\FunSubstAll((\codex\svar_x)\rar\vars\conc(\codex\svar_x),\othereq)));\\
\FunEq(((\codex\svar_x)\rar(\codex\svar_y)\conc(\codex\svar_x))\longconc\varr, (\nnumber)\longconc(((\codex\svar_y)\conc\vareq_{\rm{LHS}},(\codex\svar_x)\conc\vareq_{\rm{RHS}})\longconc\othereq)) \\\quad\seq \FunEq(\varr,\FunSim(\nnumber,((\codex\svar_x)\rar(\codex\svar_y)\conc(\codex\svar_x)),\empt,
\\\qquad\qquad\qquad\qquad\quad(\FunSubst((\codex\svar_x)\rar(\codex\svar_y)\conc(\codex\svar_x),\empt,\vareq_{\rm{LHS}}),\FunSubst((\codex\svar_x)\rar(\codex\svar_y)\conc(\codex\svar_x),(\codex\svar_x),\vareq_{\rm{RHS}}))
\\\qquad\qquad\qquad\qquad\quad\longconc\FunSubstAll((\codex\svar_x)\rar(\codex\svar_y)\conc(\codex\svar_x),\othereq)));\\
\FunEq(((\codex\svar_x)\rar(\codex\svar_y)\conc(\codex\svar_x))\longconc\varr,(\nnumber)\longconc(((\codex\svar_x)\conc\vareq_{\rm{LHS}},(\codex\svar_y)\conc\vareq_{\rm{RHS}})\longconc\othereq)) \\\quad\seq \FunEq(\varr,\FunSim(\nnumber,((\codex\svar_x)\rar(\codex\svar_y)\conc(\codex\svar_x)),\empt,
\\\qquad\qquad\qquad\qquad\quad(\FunSubst((\codex\svar_x)\rar(\codex\svar_y)\conc(\codex\svar_x),(\codex\svar_x),\vareq_{\rm{LHS}}),\FunSubst((\codex\svar_x)\rar(\codex\svar_y)\conc(\codex\svar_x),\empt,\vareq_{\rm{RHS}}))
\\\qquad\qquad\qquad\qquad\quad\longconc\FunSubstAll((\codex\svar_x)\rar(\codex\svar_y)\conc(\codex\svar_x),\othereq)));\\
\FunEq(\varr,(\nnumber)\longconc\othereq)\seq \OutFalse;\\\\
\FunSubst((\codex\svar_x)\rar\varv,\varres,\empt)\seq\varres;\\
\FunSubst((\codex\svar_x)\rar\varv,\varres,(\codex\svar_x)\longconc\varexpr)\seq
\FunSubst((\codex\svar_x)\rar\varv,\varres\longconc\varv,\varexpr);\\
\FunSubst((\codex\svar_x)\rar\varv,\varres,\vars\longconc\varexpr)\seq
\FunSubst((\codex\svar_x)\rar\varv,\varres\longconc\vars,\varexpr);\\
\FunSubst((\codex\svar_x)\rar\varv,\varres,(\codex\svar_y)\longconc\varexpr)\seq
\FunSubst((\codex\svar_x)\rar\varv,\varres\longconc(\codex\svar_y),\varexpr);
\end{array}
$$
$$
\begin{array}{lll}
\FunSubstAll((\codex\svar_x)\rar\varv,(\vareq_{\rm{LHS}},\vareq_{\rm{RHS}})
\longconc\othereq)\\\quad\seq(\FunSubst((\codex\svar_x)\rar\varv,\empt,\vareq_{\rm{LHS}}),\FunSubst((\codex\svar_x)\rar\varv,\empt,\vareq_{\rm{RHS}}))\longconc\FunSubstAll((\codex\svar_x)\rar\varv,\othereq);\\
\FunSubstAll((\codex\svar_x)\rar\varv,\empt)\seq\empt;
\\\\
\FunSim(\nnumber,\varx_{\rm{subst}},\varres,((\codex\svar_x)\conc\varx_{\rm{LHS}},(\codex\svar_x)\conc\varx_{\rm{RHS}})\longconc\othereq)\\\quad\seq
\FunSim(\nnumber,\varx_{\rm{subst}},\varres,(\varx_{\rm{LHS}},\varx_{\rm{RHS}})\longconc\othereq); \\
\FunSim(\nnumber,\varx_{\rm{subst}},\varres,(\vars\conc\varx_{\rm{LHS}},\vars\conc\varx_{\rm{RHS}})\longconc\othereq)\\\quad\seq
\FunSim(\nnumber,\varx_{\rm{subst}},\varres,(\varx_{\rm{LHS}},\varx_{\rm{RHS}})\longconc\othereq); \\
\FunSim(\nnumber,\varx_{\rm{subst}},\varres,(\varx_{\rm{LHS}}\conc(\codex\svar_x),\varx_{\rm{RHS}}\conc(\codex\svar_x))\longconc\othereq)\\\quad\seq\FunSim(\nnumber,\varx_{\rm{subst}},\varres,(\varx_{\rm{LHS}},\varx_{\rm{RHS}})\longconc\othereq); \\
\FunSim(\nnumber,\varx_{\rm{subst}},\varres,(\varx_{\rm{LHS}}\conc\vars,\varx_{\rm{RHS}}\conc\vars)\longconc\othereq)\\\quad\seq\FunSim(\nnumber,\varx_{\rm{subst}},\varres,(\varx_{\rm{LHS}},\varx_{\rm{RHS}})\longconc\othereq); \\
\FunSim(\nnumber,\varx_{\rm{subst}},\varres,(\vars_1\conc\varx_{\rm{LHS}},\vars_2\conc\varx_{\rm{RHS}})\longconc\othereq)\seq (0)\longconc(\vars_1,\vars_2); \\
\FunSim(\nnumber,\varx_{\rm{subst}},\varres,(\varx_{\rm{LHS}}\conc\vars_1,\varx_{\rm{RHS}}\conc\vars_2)\longconc\othereq)\seq (0)\longconc(\vars_1,\vars_2); \\
\FunSim(\nnumber,\varx_{\rm{subst}},\varres,(\empt,\empt)\longconc\othereq) \seq\FunSim(\nnumber,\varx_{\rm{subst}},\varres,\othereq);\\
\FunSim(\nnumber,\varx_{\rm{subst}},\varres,(\varx_{\rm{LHS}},\varx_{\rm{RHS}})\longconc\othereq) \\\quad\seq
\FunSim(\nnumber,\varx_{\rm{subst}},\varres\longconc\FunSplit(\empt,\OutNone\longconc((\Constdef 0))\longconc((\Constdef 0)),\empt,\empt,\varx_{\rm{LHS}},\varx_{\rm{RHS}}),\othereq); \\
\FunSim(\nnumber,\varx_{\rm{subst}},\varres,\empt)\seq \FunSort(\nnumber,0,\varx_{\rm{subst}},\empt,\varres);
\\ \\
\FunSort(\nnumber,\newnumber,\varx_{\rm{subst}},\varres,(\vars_1,\vars_2)\longconc\othereq)\seq (0)\longconc(\vars_1,\vars_2);\\
\FunSort(\nnumber,\newnumber,\varx_{\rm{subst}},\varres,(\varx_{\rm{eq}})\longconc\othereq)\\\quad\seq\FunSort(\nnumber,\newnumber+1,\varx_{\rm{subst}},\varres\longconc (\varx_{\rm{eq}}),\othereq);\\
\FunSort(\nnumber,\newnumber,\varx_{\rm{subst}},\varres,\empt)\seq(\newnumber)\longconc\varres;
\\\\
\FunSimPrim((\codex\conc\svar_x)\longconc\vareq_{\rm{LHS}},(\codex\conc\svar_x)\longconc\vareq_{\rm{RHS}})\seq\FunSimPrim(\vareq_{\rm{LHS}},\vareq_{\rm{RHS}});\\
\FunSimPrim(\vars\longconc\vareq_{\rm{LHS}},\vars\longconc\vareq_{\rm{RHS}})\seq\FunSimPrim(\vareq_{\rm{LHS}},\vareq_{\rm{RHS}});\\
\FunSimPrim(\vareq_{\rm{LHS}}\longconc(\codex\conc\svar_x),\vareq_{\rm{RHS}}\longconc(\codex\conc\svar_x))\seq\FunSimPrim(\vareq_{\rm{LHS}},\vareq_{\rm{RHS}});\\
\FunSimPrim(\vareq_{\rm{LHS}}\longconc\vars,\vareq_{\rm{RHS}}\longconc\vars)\seq\FunSimPrim(\vareq_{\rm{LHS}},\vareq_{\rm{RHS}});\\
\FunSimPrim(\svar_1\longconc\vareq_{\rm{LHS}},\svar_2\longconc\vareq_{\rm{RHS}})\seq(\svar_1,\svar_2);\\
\FunSimPrim(\vareq_{\rm{LHS}}\longconc\svar_1,\vareq_{\rm{RHS}}\longconc\svar_2)\seq(\svar_1,\svar_2);\\
\FunSimPrim(\vareq_{\rm{LHS}},\vareq_{\rm{RHS}})\seq(\vareq_{\rm{LHS}},\vareq_{\rm{RHS}});
\\ \\
/\hspace{-0.3ex}*\;\mathrm{Here\;we\;use\;a\;syntactic\;sugar:}\;t_i\mathrm{\;denotes\;either\;an\;encoded\;variable\;or\;a\;symbol\;(one\;rule\;corresponds\;to}\\
\mathrm{\quad\;\;the\;four\;desugared\;rules).\;The\;last\;two\;arguments\;recursively\;decrease.}\;*\hspace{-0.3ex}/\\
\FunSplit(\varres,\OutNone\longconc(\varms_1)\longconc(\varms_2),\varpref_{\rm{LHS}},\varpref_{\rm{RHS}},\varterm_1\longconc\vareq_{\rm{LHS}},\varterm_2\longconc\vareq_{\rm{RHS}}) \\\quad \seq \FunSplit(\varres,
\\\qquad\qquad\quad\FunCountMS(\FunInclude(\varterm_1,\empt,(\varms_1)),\FunInclude(\varterm_2,\empt,(\varms_2))),\varpref_{\rm{LHS}}\longconc\varterm_1,\varpref_{\rm{RHS}}\longconc\varterm_2,\vareq_{\rm{LHS}},\vareq_{\rm{RHS}});\\
\FunSplit(\varres,\OutFalse\longconc(\varms_1)\longconc(\varms_2),\varpref_{\rm{LHS}},\varpref_{\rm{RHS}},\varterm_1\longconc\vareq_{\rm{LHS}},\varterm_2\longconc\vareq_{\rm{RHS}}) \\\quad\seq \FunSplit(\varres,
\\\qquad\qquad\quad\FunCountMS(\FunInclude(\varterm_1,\empt,(\varms_1)),\FunInclude(\varterm_2,\empt,(\varms_2))),\varpref_{\rm{LHS}}\longconc\varterm_1,\varpref_{\rm{RHS}}\longconc\varterm_2,\vareq_{\rm{LHS}},\vareq_{\rm{RHS}});\\
\FunSplit(\varres,\OutTrue\longconc(\varms_1)\longconc(\varms_2),\varpref_{\rm{LHS}},\varpref_{\rm{RHS}},\vareq_{\rm{LHS}},\vareq_{\rm{RHS}})\\\quad\seq
\FunSplit(\varres\longconc(\varpref_{\rm{LHS}},\varpref_{\rm{RHS}}),
\OutNone\longconc((\Constdef 0))\longconc((\Constdef 0)),\empt,\empt,\vareq_{\rm{LHS}},\vareq_{\rm{RHS}});
\\
\FunSplit(\varres,\vars\longconc(\varms_1)\longconc(\varms_2),\empt,\empt,\empt,\empt)\seq \varres;\\
\FunSplit(\varres,\vars\longconc(\varms_1)\longconc(\varms_2),\varpref_{\rm{LHS}},\varpref_{\rm{RHS}},\vareq_{\rm{LHS}},\vareq_{\rm{RHS}})\\\quad\seq
\FunSimPrim(\varpref_{\rm{LHS}}\longconc\vareq_{\rm{LHS}},\varpref_{\rm{RHS}}\longconc\vareq_{\rm{RHS}})\longconc\varres;
\\\\\FunInclude(\vars,\varprev,(\varms\longconc(\Constdef\nnumber)))\seq
(\varprev\longconc\varms\longconc(\Constdef\nnumber+1));\\
\FunInclude((\codex\svar_x),\varprev,(((\codex\svar_x)\;\nnumber)\longconc\varms))\seq
(\varprev\longconc((\codex\svar_x)\;\nnumber+1)\longconc\varms);\\
\FunInclude((\codex\svar_x),\varprev,(((\codex\svar_y)\;\nnumber)\longconc\varms))\seq
\FunInclude((\codex\svar_x),\varprev\longconc((\codex\svar_y)\;\nnumber),\varms);\\
\FunInclude((\codex\svar_x),\varprev,((\Constdef\nnumber)))\seq
(((\codex\svar_x)\;1)\longconc\varprev\longconc(\Constdef\nnumber));
\\\\
/\hspace{-0.3ex}*\;\mathrm{Here\;we\;use\;a\;syntactic\;sugar:}\;t\mathrm{\;denotes\;an\;expression\;enclosed\;in\;parentheses\;having\;either}\\
\mathrm{\quad\;an\;encoded\;variable\;or\;}\Constdef
\mathrm{\;as\;the\;prefix\;(one\;rule\;corresponds\;to\;the\;two\;desugared\;rules).}\;*\hspace{-0.3ex}/\\
\FunCountMS((\varterm\longconc\varms_1),(\varms_2))\seq\FunAreEqual(\varms_1,\FunElMinus(\varterm,\empt,\varms_2))\longconc(\varterm\longconc\varms_1)\longconc(\varms_2);
\end{array}
$$
\newpage
$$
\begin{array}{lll}
\FunAreEqual(\varms_1,\varms_2\longconc\OutFalse)\seq\OutFalse; \\
\FunAreEqual(\empt,\empt)\seq\OutTrue; \\
\FunAreEqual(\varms_1,\empt)\seq\OutFalse; \\
\FunAreEqual(\empt,\varms_2)\seq\OutFalse; \\
\FunAreEqual(((\codex\vars)\conc\nnumber)\longconc\varms_1,\varms_2)\seq \FunAreEqual(\varms_1,\FunElMinus(((\codex\vars)\conc\nnumber),\empt,\varms_2));\\
\FunAreEqual((\Constdef\conc\nnumber)\longconc\varms_1,\varms_2)\seq \FunAreEqual(\varms_1,\FunElMinus((\Constdef\conc\nnumber),\empt,\varms_2));
\\ \\
\FunElMinus((\Constdef 0),\empt,\varms\longconc(\Constdef 0))\seq
\varms; \\
\FunElMinus((\Constdef\nnumber_1+1),\empt,\varms\longconc(\Constdef\nnumber_2+1))\\\qquad\seq
\FunElMinus((\Constdef\nnumber_1),\empt,\varms\longconc(\Constdef\nnumber_2)); \\
\FunElMinus(((\codex\svar_x)\;0),\varprev,((\codex\svar_x)\;0)\longconc\varms)\seq
\varprev\longconc\varms; \\
\FunElMinus(((\codex\svar_x)\;\nnumber_1+1),\varprev,((\codex\svar_x)\;\nnumber_2+1)\longconc\varms)\\\qquad\seq
\FunElMinus(((\codex\svar_x)\;\nnumber_1),\varprev,((\codex\svar_x)\;\nnumber_2)\longconc\varms); \\
\FunElMinus(((\codex\svar_x)\;\nnumber_1),\varprev,((\codex\svar_x)\;\nnumber_2)\longconc\varms)\seq\OutFalse; \\
\FunElMinus(((\codex\svar_x)\;\nnumber_1),\varprev,((\codex\svar_y)\;\nnumber_2)\longconc\varms)\\\qquad\seq\FunElMinus(((\codex\svar_x)\;\varcount_1),\varprev\longconc((\codex\svar_y)\;\nnumber_2),\varms); \\
\FunElMinus(((\codex\svar_x)\;\nnumber_1),\varprev,(\Constdef\nnumber_2))\seq\OutFalse;
\end{array}
$$

\normalsize
\subsubsection{Counting Interpreter}
This interpreter uses the function definitions given for the interpreter $\WeqIntSplit$ plus some additional functions, provided that the function $\FunElMinus$ is modified and the last rule of the function $\FunSplit$ is replaced with the following rewriting rule.

\footnotesize
$$
\noindent\begin{array}{ll}
/\hspace{-0.3ex}* \mathrm{This~rule~replaces~the~last~rule~of~\FunSplit~definition~given~in~}\WeqIntSplit.*\hspace{-0.3ex}/\\
\FunSplit(\varres,\vars\longconc(\varms_1)\longconc(\varms_2),\varpref_{\rm{LHS}},\varpref_{\rm{RHS}},\vareq_{\rm{LHS}},\vareq_{\rm{RHS}})\\\quad\seq
\FunSplitRight(\varres,
\OutNone\longconc((\Constdef 0))\longconc((\Constdef 0)),\empt,\empt,\FunSimPrim(\varpref_{\rm{LHS}}\longconc\vareq_{\rm{LHS}},\varpref_{\rm{RHS}}\longconc\vareq_{\rm{RHS}}));\\
\end{array}
$$

\normalsize
The additional function definitions are given below. The following definition replaces the version of the $\FunElMinus$ source code given in $\WeqIntSplit$ source code.

\footnotesize
$$
\begin{array}{ll}
\FunElMinus((\Constdef \nnumber_1), \empt, \varms\longconc(\Constdef \nnumber_2))\seq \varms\longconc\FunCountMinus(\nnumber_1,\nnumber_2);
\\
\FunElMinus(((\codex\vars)\;\nnumber_1),\varprev, ((\codex\vars)\;\nnumber_2)\longconc\varms)\seq
\varprev\longconc\varms\longconc\FunCountMinus(\nnumber_1,\nnumber_2);
\\
\FunElMinus(((\codex\vars)\;\nnumber_1),\varprev, (\Constdef \nnumber_2))\seq\OutGreater\longconc\OutFalse;\\
\FunElMinus(((\codex\vars_1)\;\nnumber_1), \varprev, ((\codex\vars_2)\;\nnumber_2)\longconc\varms)\\\quad\seq 	
\FunElMinus(((\codex\vars_1)\;\nnumber_1),\varprev\longconc((\codex\vars_2)\;\nnumber_2),\varms);
\\
\\
\FunCountMinus(0,0)\seq\empt;\\
\FunCountMinus(\nnumber_1+1,\nnumber_2+1)\seq\FunCountMinus(\nnumber_1,\nnumber_2);\\
\FunCountMinus(\empt,\nnumber)\seq\OutLesser\longconc\OutFalse;\\
\FunCountMinus(\nnumber,\empt)\seq\OutGreater\longconc\OutFalse;
\\\\
/\hspace{-0.3ex}*\;\mathrm{Here\;we\;use\;a\;syntactic\;sugar:}\;t_i\mathrm{\;denotes\;either\;an\;encoded\;variable\;or\;a\;symbol\;(one\;rule\;corresponds\;to}\\
\quad\;\;\mathrm{the\;four\;desugared\;rules).\;The\;last\;two\;arguments\;recursively\;decrease.}\;*\hspace{-0.3ex}/\\
\FunSplitRight(\varres,\vars\longconc(\varms_1)\longconc(\varms_2),\empt,\empt,\empt,\empt)\seq\varres;\\
\FunSplitRight(\varres,\OutNone\longconc(\varms_1)\longconc(\varms_2),\varsuff_{\rm{LHS}},\varsuff_{\rm{RHS}},\vareq_{\rm{LHS}}\longconc\varterm_1,\vareq_{\rm{RHS}}\longconc\varterm_2)\\\quad\seq\FunSplitRight(\varres,
\\\qquad\qquad\quad\FunCountMS(\FunInclude(\varterm_1,\empt,(\varms_1)),\FunInclude(\varterm_2,\empt,(\varms_2))), \varterm_1\longconc\varsuff_{\rm{LHS}},\varterm_2\longconc\varsuff_{\rm{RHS}},\vareq_{\rm{LHS}},\vareq_{\rm{RHS}});\\
\\
\FunSplitRight(\varres,\OutFalse\longconc(\varms_1)\longconc(\varms_2),\varsuff_{\rm{LHS}},\varsuff_{\rm{RHS}},\vareq_{\rm{LHS}}\longconc\varterm_1,\vareq_{\rm{RHS}}\longconc\varterm_2)\\\quad\seq\FunSplitRight(\varres,
\\\qquad\qquad\quad\FunCountMS(\FunInclude(\varterm_1,\empt,(\varms_1)),\FunInclude(\varterm_2,\empt,(\varms_2))), \varterm_1\longconc\varsuff_{\rm{LHS}},\varterm_2\longconc\varsuff_{\rm{RHS}},\vareq_{\rm{LHS}},\vareq_{\rm{RHS}});\\
\end{array}
$$$$\begin{array}{lll}
\FunSplitRight(\varres,\OutTrue\longconc(\varms_1)\longconc(\varms_2),\varsuff_{\rm{LHS}},\varsuff_{\rm{RHS}},\vareq_{\rm{LHS}},\vareq_{\rm{RHS}})\\\quad\seq\FunSplitRight(\varres\longconc(\varsuff_{\rm{LHS}},\,\varsuff_{\rm{RHS}}),\OutNone\longconc(\Constdef 0)\longconc(\Constdef 0),\empt,\empt,\vareq_{\rm{LHS}},\vareq_{\rm{RHS}});\\
\FunSplitRight(\varres,\vars\longconc(\varms_1)\longconc(\varms_2),\varsuff_{\rm{LHS}},\varsuff_{\rm{RHS}},\vareq_{\rm{LHS}},\vareq_{\rm{RHS}})\\\quad\seq\FunSubjEq(\FunYieldCheck(
\FunAddElstoMS (\vareq_{\rm{LHS}},(\varms_1)),\FunAddElstoMS (\vareq_{\rm{RHS}},(\varms_2))),
\\\qquad\qquad\qquad\qquad\quad\vareq_{\rm{LHS}}\longconc\varsuff_{\rm{LHS}},\vareq_{\rm{RHS}}\longconc\varsuff_{\rm{RHS}})\longconc\varres;
\\\\
\FunSubjEq(\OutFalse,\vareq_{\rm{LHS}},\vareq_{\rm{RHS}})\seq(\vareq_{\rm{LHS}},\vareq_{\rm{RHS}});\\
\FunSubjEq(\OutTrue,\vareq_{\rm{LHS}},\vareq_{\rm{RHS}})\seq(\cA,\cB);\\
\\
\FunSubtractEl(\OutGreater,\varms_1,\empt)\seq\OutTrue;\\
\FunSubtractEl(\OutLesser,\empt,\varms_2)\seq\OutTrue;\\
\FunSubtractEl(\OutGreater,((\codex\vars)\;\nnumber)\longconc\varms_1,\varms_2)\seq
\FunCheckInfo(\OutGreater,\FunElMinus(((\codex\vars)\;\nnumber),\empt,\varms_2),\varms_1);\\
\FunSubtractEl(\OutGreater,(\Constdef\nnumber)\longconc\varms_1,\varms_2)\seq
\FunCheckInfo(\OutGreater,\FunElMinus((\Constdef\nnumber),\empt,\varms_2),\varms_1);\\
\FunSubtractEl(\OutLesser,((\codex\vars)\;\nnumber)\longconc\varms_1,\varms_2)\seq
\FunCheckInfo(\OutLesser,\FunElMinus(((\codex\vars)\;\nnumber),\empt,\varms_2),\varms_1);\\
\FunSubtractEl(\OutLesser,(\Constdef\nnumber)\longconc\varms_1,\varms_2)\seq
\FunCheckInfo(\OutLesser,\FunElMinus((\Constdef\nnumber),\empt,\varms_2),\varms_1);\\
\FunSubtractEl(\vars,\varms_1,\varms_2)\seq\OutFalse;
\\ \\
\FunCheckInfo(\vars,\varms_2\longconc\vars\longconc\OutFalse,\varms_1)\seq\FunSubtractEl(\vars,\varms_1,\varms_2);\\
\FunCheckInfo(\OutGreater,\varms_2\longconc\OutLesser\longconc\OutFalse,\varms_1)\seq \OutFalse;\\
\FunCheckInfo(\OutLesser,\varms_2\longconc\OutGreater\longconc\OutFalse,\varms_1)\seq\OutFalse;\\
\FunCheckInfo(\vars,\varms_2,\varms_1)\seq\FunSubtractEl(\vars,\varms_1,\varms_2);
\\ \\
\FunYieldCheck(\varms_1\longconc(\Constdef\varcount_1),\varms_2\longconc(\Constdef\varcount_2))\\\quad\seq\FunYieldCheckAux(\FunElMinus ((\Constdef \varcount_1),\empt,(\Constdef \varcount_2)),\varms_1,\varms_2);
\\\\
\FunYieldCheckAux(\empt,\varms_1,\varms_2)\seq\OutFalse;\\
\FunYieldCheckAux(\varms\longconc\vars\longconc\OutFalse,\varms_1,\varms_2)\seq\FunSubtractEl(\vars,\varms_1,\varms_2);
\\\\
\FunAddElstoMS((\codex\vars)\longconc\varexpr,\varms)\seq\FunAddElstoMS(\varexpr,\FunInclude((\codex\vars),\empt,\varms));\\
\FunAddElstoMS(\vars\longconc\varexpr,\varms)\seq\FunAddElstoMS(\varexpr,\FunInclude(\vars,\empt,\varms));\\
\FunAddElstoMS(\empt,(\varms))\seq\varms;
\end{array}
$$

\end{document}